\pdfoutput=1
\documentclass[a4paper,english,fleqn,11pt,final]{scrartcl}

\usepackage[utf8]{inputenc}
\usepackage[T1]{fontenc}
\usepackage{lmodern}
\usepackage[english]{babel}

\usepackage[protrusion=true,expansion=true,final]{microtype}
\usepackage{textcomp}

\usepackage{xcolor}
\usepackage[heavycircles]{stmaryrd}
\usepackage[centertags,sumlimits,intlimits,namelimits,fleqn]{amsmath}
\usepackage[fixamsmath,disallowspaces]{mathtools}
\usepackage{bussproofs}
\usepackage{lplfitch}

\usepackage{relsize}
\usepackage{scalerel}

\usepackage{braket}

\usepackage{amsfonts}
\usepackage{amssymb}

\usepackage[backref=none,pagebackref=false,hyperindex=true,hyperfootnotes=false,bookmarks=true,pdfpagelabels=true,final]{hyperref}

\hypersetup{
    colorlinks,
    linkcolor={red!80!black},
    citecolor={blue!50!black},
    urlcolor={blue!80!black}
}

\usepackage{amsthm}

\usepackage[babel,autostyle=true,strict]{csquotes}
\MakeOuterQuote{"}
\usepackage[xspace]{ellipsis}
\usepackage[section]{placeins} 

\usepackage{bookmark}
\usepackage{pdfcolmk} 
\usepackage{authblk}

\hypersetup{
    colorlinks,
    linkcolor={red!50!black},
    citecolor={blue!50!black},
    urlcolor={blue!80!black}
}

\usepackage{tabularx}
\usepackage{booktabs}

\usepackage{tikz}
\usepackage{pgf}
\usetikzlibrary{arrows,positioning}
\usetikzlibrary{shapes.geometric}
\usetikzlibrary{decorations.pathreplacing}

\newcommand{\ie}{i.e.\@\xspace}
\newcommand{\eg}{e.g.\@\xspace}
\newcommand{\wrt}{w.\,r.\,t.\@\xspace}

\newcommand{\wloss}{w.l.o.g.\@\xspace}
\newcommand{\Wloss}{W.l.o.g.\@\xspace}

\newcommand{\logicOpFont}[1]{\mathsf{#1}}         
\newcommand{\mathCommandFont}[1]{\mathrm{#1}}     

\newcommand{\size}[1]{{\protect\ensuremath{\vert\nobreak#1\nobreak\vert}}}
\newcommand{\pow}[1]{{\protect\ensuremath{\mathord{\mathfrak{P}(\nobreak#1\nobreak)}}}}
\newcommand{\arity}[1]{{\protect\ensuremath{\mathCommandFont{ar}(#1)}}}
\newcommand{\Mod}[1]{{\protect\ensuremath{\mathsf{Mod}(#1)}}}

\newcommand{\negg}{{\sim}}
\newcommand{\depop}{{=\!\!(\cdot,\cdot)}}
\newcommand{\dep}[1]{{=\!\!(#1)}}
\newcommand{\depopp}{\depop}

\newcommand{\logic}[1]{\ensuremath{\mathsf{#1}}\xspace}

\newcommand{\PS}{\logic{Prop}}
\newcommand{\D}{\logic{D}}
\newcommand{\PL}{\logic{PL}}
\newcommand{\QPL}{\logic{QPL}}
\newcommand{\QPPL}{\logic{(Q)PL}}
\newcommand{\ML}{\logic{ML}}
\newcommand{\FO}{\logic{FO}}
\newcommand{\SO}{\logic{SO}}
\newcommand{\IF}{\logic{IF}}

\newcommand{\PTL}{\logic{PTL}}
\newcommand{\QPDL}{\logic{QPDL}}
\newcommand{\PDL}{\logic{PDL}}
\newcommand{\QPTL}{\logic{QPTL}}
\newcommand{\QPPTL}{\logic{(Q)PTL}}
\newcommand{\MTL}{\logic{MTL}}
\newcommand{\TL}{\logic{TL}}
\newcommand{\MDL}{\logic{MDL}}
\newcommand{\EMDL}{\logic{EMDL}}
\newcommand{\EEMDL}{\logic{(E)MDL}}

\newcommand{\MINC}{\logic{MInc}}
\newcommand{\MINCEX}{\logic{MIncEx}}
\newcommand{\QPINC}{\logic{QPInc}}
\newcommand{\QPINCEX}{\logic{QPIncEx}}
\newcommand{\PINCEX}{\logic{PIncEx}}
\newcommand{\PINC}{\logic{PInc}}
\newcommand{\QPIND}{\logic{QPIL}}
\newcommand{\QPIL}{\logic{QPIL}}

\newcommand{\PIL}{\logic{PIL}}
\newcommand{\MIL}{\logic{MIL}}

\newcommand{\NE}{\text{\textsc{ne}}}
\newcommand{\E}{\protect\ensuremath\logicOpFont{E}}

\newcommand{\Var}{\mathrm{Var}}

\newcommand{\frakA}{\protect\ensuremath{\mathfrak{A}}}

\newcommand{\calA}{\protect\ensuremath{\mathcal{A}}}
\newcommand{\calB}{\protect\ensuremath{\mathcal{B}}}

\newcommand{\calF}{\protect\ensuremath{\mathcal{F}}}

\newcommand{\calL}{\protect\ensuremath{\mathcal{L}}}
\newcommand{\calK}{\protect\ensuremath{\mathcal{K}}}

\newcommand{\calP}{\protect\ensuremath{\mathcal{P}}}

\newcommand{\calT}{\protect\ensuremath{\mathcal{T}}}

\newcommand{\sfD}{\protect\ensuremath{\mathsf{D}}}
\newcommand{\sfS}{\protect\ensuremath{\mathsf{S}}}
\newcommand{\sfH}{\protect\ensuremath{\mathsf{H}}}

\newcommand{\sfL}{\protect\ensuremath{\mathsf{L}}}
\newcommand{\sfM}{\protect\ensuremath{\mathsf{M}}}
\newcommand{\sfQ}{\protect\ensuremath{\mathsf{Q}}}
\newcommand{\sfU}{\protect\ensuremath{\mathsf{U}}}

\providecommand{\dfn}{\mathrel{\mathop:}=}
\providecommand{\ddfn}{\mathrel{\mathop{{\mathop:}{\mathop:}}}=}

\newcommand{\imp}{\rightarrow}
\newcommand{\limp}{\multimap}
\newcommand{\timp}{\rightarrowtriangle}
\newcommand{\ltimp}{\leftarrowtriangle}
\newcommand{\tequiv}{\leftrightarrowtriangle}
\newcommand{\equ}{\leftrightarrow}
\newcommand{\eqpr}{\dashv\vdash}
\newcommand{\tens}{\otimes}

\newcommand{\Deriv}[1]{{\normalfont\textsf{#1}}}
\newcommand{\Apply}[1]{\textsf{\scriptsize{}#1}}
\newenvironment{bprooftree}{\leavevmode\hbox\bgroup}{\DisplayProof\egroup}

\newcommand{\rel}{\mathsf{rel}}

\newcommand{\oland}{\owedge}

\DeclareMathOperator*{\bigovee}{\scalerel*{\ovee}{\sum}}
\DeclareMathOperator*{\bigoland}{\scalerel*{\owedge}{\sum}}
\DeclareMathOperator*{\bigowedge}{\scalerel*{\owedge}{\sum}}
\DeclareMathOperator*{\bigtens}{\bigotimes}
\DeclareMathOperator*{\bigdiamond}{\scalerel*{\Diamond}{\sum}}

\DeclareMathOperator{\shriek}{!}
\DeclareMathOperator{\indep}{\perp}

\DeclareMathOperator{\exclusion}{\mid}

\EnableBpAbbreviations

\newcommand{\falsum}{\mathrlap{\,\bot}\bot\,}
 
\theoremstyle{plain}
\newtheorem{theorem}{Theorem}[section]

\newtheorem{proposition}[theorem]{Proposition}
\newtheorem{lemma}[theorem]{Lemma}
\newtheorem{corollary}[theorem]{Corollary}
\theoremstyle{definition}
\newtheorem{definition}[theorem]{Definition}
\newtheorem{example}[theorem]{Example}

\makeatletter
\newtheorem*{rep@theorem}{\rep@title}
\newcommand{\newreptheorem}[2]{
\newenvironment{rep#1}[1]{
 \def\rep@title{#2 \ref{##1}}
 \begin{rep@theorem}}
 {\end{rep@theorem}}}
\makeatother
\newreptheorem{proposition}{Proposition}
\newreptheorem{lemma}{Lemma}
\newreptheorem{theorem}{Theorem}

\newcommand{\thm}{\text{\scriptsize\; (thm)}}

\usepackage[backend=biber,style=numeric,sortcites,url=false,maxbibnames=99,maxcitenames=2,citetracker=true]{biblatex}

\DeclareFieldFormat[article]{title}{\mkbibemph{#1}}
\DeclareFieldFormat[article]{journal}{#1}
\DeclareFieldFormat[article]{journaltitle}{#1}
\DeclareFieldFormat[article]{volume}{\textbf{#1}}
\DeclareFieldFormat[article]{number}{no.~#1}

\DeclareFieldFormat[incollection]{title}{\mkbibemph{#1}}
\DeclareFieldFormat[incollection]{booktitle}{#1}
\DeclareFieldFormat[incollection]{volume}{\textbf{#1}}
\DeclareFieldFormat[incollection]{number}{(#1)}

\DeclareFieldFormat[inproceedings]{title}{\mkbibemph{#1}}
\DeclareFieldFormat[inproceedings]{booktitle}{#1}
\DeclareFieldFormat[inproceedings]{volume}{\textbf{#1}}
\DeclareFieldFormat[inproceedings]{number}{(#1)}

\DeclareFieldFormat[unpublished]{title}{\mkbibemph{#1}}

\DeclareFieldFormat[thesis]{title}{\mkbibemph{#1}}

\renewbibmacro*{volume+number+eid}{
  \printfield{volume}\space
  \printtext[parens]{\usebibmacro{date}}
  \setunit{\addcomma\space}
  \printfield{number}
  \setunit{\addcomma\space}
  \printfield{eid}
  }

\renewbibmacro*{issue+date}{
  \iffieldundef{issue}
      {}
      {\printfield{issue}}
  \newunit}

\renewbibmacro{in:}{}
 
\makeatletter
\DeclareOldFontCommand{\rm}{\normalfont\rmfamily}{\mathrm}
\DeclareOldFontCommand{\sf}{\normalfont\sffamily}{\mathsf}
\DeclareOldFontCommand{\tt}{\normalfont\ttfamily}{\mathtt}
\DeclareOldFontCommand{\bf}{\normalfont\bfseries}{\mathbf}
\DeclareOldFontCommand{\it}{\normalfont\itshape}{\mathit}
\DeclareOldFontCommand{\sl}{\normalfont\slshape}{\@nomath\sl}
\DeclareOldFontCommand{\sc}{\normalfont\scshape}{\@nomath\sc}
\makeatother

\title{Axiomatizations of Team Logics}
\author{Martin Lück}
\date{\vspace{-5ex}}
\affil{\small{}Leibniz Universität Hannover, Institut für Theoretische Informatik,\\
Appelstraße 4, 30167 Hannover, Germany\\\texttt{lueck@thi.uni-hannover.de}}

\bibliography{axiom}

\begin{document}
\maketitle

\begin{abstract}
\textbf{Abstract.}
In a modular approach, we lift Hilbert-style proof systems for propositional, modal and first-order logic to generalized systems for their respective team-based extensions.
We obtain sound and complete axiomatizations for the dependence-free fragment FO($\negg$) of Väänänen's first-order team logic TL, for propositional team logic PTL, quantified propositional team logic QPTL, modal team logic MTL, and for the corresponding logics of dependence, independence, inclusion and exclusion.

As a crucial step in the completeness proof, we show that the above logics admit, in a particular sense, a semantics-preserving elimination of modalities and quantifiers from formulas.
\end{abstract}

\noindent\textit{Keywords and phrases:} axiomatization, dependence logic, propositional team logic, modal team logic, team logic

\noindent\textit{1998 ACM Subject Classification:} F.4.1 Mathematical Logic

\section{Introduction}

While their history goes back to ancient philosophers, propositional and modal logics have assumed an outstanding role in the age of modern computer science, with plentiful applications in software verification, modeling, artificial intelligence, and protocol design.
An important property of a logical framework is \emph{completeness}, \ie, that the act of mechanical reasoning can effectively be done by an algorithm.
The question of completeness of first-order logic, which is the foundation of today's mathematics, was not settled until the famous result of Gödel in 1929.
Until today, the area of proof theory has achieved tremendous progress and is still a growing field, especially with regard to many variants of propositional and modal logics as well as non-classical logics (see \eg Fitting \cite{fitting_proof_1983}).

\smallskip

A recent extension of classical logics is so-called \emph{team logic}.
It originated from the concept of quantifier dependence and independence. The following question has been long-known in linguistics: how can the statement
\begin{quote}
\emph{For every $x$ there is $y$, and for every $u$ there is $v$ such that P(x,y,u,v).}
\end{quote} be formalized in first-order logic such that $y$ and $v$ are chosen independently?
Some suggestions were Henkin's \emph{branching quantifiers} \cite{Henkin1961} as well as \emph{in\-de\-pend\-ence-friendly logic} $\IF$ by Hintikka and Sandu~\cite{hintikka_informational_1989}.
The idea of the latter is to assert dependence and independence between quantifiers syntactically, implemented semantically by a game of imperfect information.
Hodges \cite{Hodges1997} proved that $\IF$ also admits a compositional semantics if formulas were evaluated on \emph{teams}, which are sets of assignments, instead of single assignments.
In this vein, Väänänen~\cite{vaananen_dependence_2007} introduced \emph{dependence logic} $\D$.
Here, the fundamental idea is that dependencies are not stated alongside the quantifiers, but instead are expressed as logical \emph{dependence atoms}, written $\dep{x,y}$, which means "$x$ functionally determines $y$."

\smallskip

Beside Väänänen's dependence atom, a variety of atomic formulas solely for reasoning in teams were introduced. Galliani~\cite{galliani_inclusion_2012} as well as Grädel and Väänänen~\cite{gradel2013dependence} pointed out connections to database theory; they formalized common constraints like \emph{independence} $\perp$, \emph{inclusion} $\subseteq$ and \emph{exclusion} $|$ as atoms in the framework of team semantics.
Beside first-order logic, all these atoms have also been adapted for modal logic $\ML$ \cite{vaananen_modal_2008}, and (quantified) propositional logic $\PL$ resp.\ $\QPL$ \cite{sano_et_al,yang_propositional_2016,gandalf}.

\smallskip

As for any logic, the question of axiomatizability arises for these logics with team semantics, in particular for the extensions of first-order logic.
However, dependence logic $\D$ is as expressive as existential second-order logic $\SO(\exists)$ \cite{vaananen_dependence_2007}, while its extension $\TL$, obtained from $\D$ by adding a semantical negation $\negg$, is equivalent to full second-order logic $\SO$ \cite{hutchison_team_2009}.
Accordingly, both are non-axiomatizable.
Later, Kontinen and Väänänen~\cite{kontinen_axiomatizing_2013} gave a partial axiomatization in the sense that $\FO$-consequences of $\D$-formulas are derivable, and recently a system that can derive all so-called \emph{negatable} $\D$-formulas was presented by Yang~\cite{Yang2016}.

\smallskip

For certain fragments of propositional and modal team logic, axiomatizations exist.
Hannula~\cite{Hannula16} presented natural deduction systems for propositional dependence logic $\PDL$, quantified propositional dependence logic $\QPDL$ and extended modal dependence logic $\EMDL$.
By contrast, Sano and Virtema~\cite{sano_et_al} gave Hilbert-style axiomatizations and labeled tableau calculi for propositional dependence logic $\PDL$ and (extended) modal dependence logic $\EEMDL$.
Independently, Yang~\cite{Yang17} presented both Hilbert-style axiomatizations and natural deduction systems for a family of so-called \emph{downward-closed} modal logics with team semantics, which includes $\EMDL$ as well.

However, a fundamental restriction of these solutions is that they all rely on the absence of Boolean negation.
As a consequence, team logics with negation, most notably propositional team logic ($\PTL$), modal team logic ($\MTL$) and $\FO(\negg)$, require a different approach.

\smallskip

\subsubsection*{Contribution}

In this paper, we present complete axiomatizations for several team logics including the $\depopp$-free fragment of $\TL$, coined $\FO(\negg)$ by Gallani~\cite{Galliani14}.
Here, we consider it under \emph{lax semantics}~\cite{galliani_inclusion_2012}.

By showing that $\FO(\negg)$ is axiomatizable, we identify the dependence atom $\depopp$, and not team semantics itself, as the source of incompleteness of $\D$ and $\TL$.
One interpretation is that reasoning about teams can be axiomatized; but only if we cannot talk about the internal dependencies between the elements of the team.

A crucial step in the completeness proof is the perhaps surprising fact that $\TL$ without $\depopp$ collapses to $\calB(\FO)$, the Boolean closure of classical first-order logic $\FO$ under team semantics.
The latter has the so-called \emph{flatness property}, which implies that any classical proof system of $\FO$ is also adequate for team semantics.
From there, an axiomatization of $\calB(\FO)$ is easily found in a similar way as for propositional logic.

Whether logics not collapsing to $\calB(\FO)$ have axiomatizations is beyond the scope of this paper.
Our approach, however, also yields results for (quantified) propositional and modal team logics.
They can define their atoms of dependence, independence and inclusion in terms of other connectives, whereas this is not possible in the first-order setting.
For this reason, the logics $\QPTL$, $\PTL$ and $\MTL$ collapse to $\calB(\QPL)$, $\calB(\PL)$ and $\calB(\ML)$ in a similar fashion as $\FO(\negg)$ to $\calB(\FO)$, and we obtain complete axiomatizations as a byproduct.
Figure~\ref{fig:overview} illustrates this.

\smallskip

The article is organized as follows.
Let us remark that the axiomatizations as a whole can be found in Figure~\ref{fig:all-axioms}.
In each section of the paper, one subsystem is introduced.
First, the system $\sfL$ is presented in Section~\ref{sec:boolean} as a complete proof system for the Boolean closure $\calB(\calL)$ under team semantics, where $\calL\in \{\PL,\QPL,\ML,\FO\}$.
In Section~\ref{sec:splitting}, the system $\sfS$ is added which permits to eliminate the \emph{splitting disjunction} $\tens$ in a semantics-preserving manner.
By means of this elimination, $\PTL$ collapses to $\calB(\PL)$ in the proof system.
Likewise, in Section~\ref{sec:modal} it is shown that modalities can be eliminated in a system we call $\sfM$, and that the problem of the axiomatization of $\MTL$ is thereby reduced to that of $\calB(\ML)$ as well.
The results on $\PTL$ and $\MTL$ were presented earlier \cite{csl16axiom}, and are now extended to logics containing quantifiers, namely quantified Boolean formulas and first-order logic.
Section~\ref{sec:qbf} introduces the system $\sfQ$ which allows to axiomatize and hence eliminate quantifiers in a similar fashion as the modalities.

\begin{figure}\centering
\begin{tikzpicture}[thick,scale=0.8]
\draw (0,-.3) ellipse (1.5cm and 2.5cm);
\node (MTL) at (0,1.8cm){$\PTL$};
\node[color=red] at (0,1.2cm) {$\sfH^0\sfL\sfS$};
\draw (0,-0.99cm) ellipse (1.27cm and 1.81cm);
\node (BML) at (0,0.2cm){$\calB(\PL)$};
\node[color=red] at (0,-.4cm) {$\sfH^0\sfL$};
\draw (0,-1.8cm) ellipse (1cm and 1cm);
\node (ML) at (0,-1.8cm){$\PL$};
\node[color=red] at (0.1cm,-2.4cm) {$\sfH^0$};
\draw[red,<-] (-1.8cm,0.2cm) to[in=180,out=180] node[left,black]{\scriptsize(eliminate $\limp,\Box,\triangle,\forall,\shriek$)} (-1.8cm,1.5cm);
\draw[red,->] (-1.8cm,-1.7cm) to[in=180,out=180] node[left,black]{\scriptsize(add propositional logic)} (-1.8cm,-.2cm);
\end{tikzpicture}
\begin{tikzpicture}[thick,scale=0.8]
\draw (0,-.3) ellipse (1.5cm and 2.5cm);
\node (MTL) at (0,1.8cm){$\MTL$};
\node[color=red] at (0,1.2cm) {$\sfH^\Box\sfL\sfS\sfM$};
\draw (0,-0.99cm) ellipse (1.27cm and 1.81cm);
\node (BML) at (0,0.2cm){$\calB(\ML)$};
\node[color=red] at (0,-.4cm) {$\sfH^\Box\sfL$};
\draw (0,-1.8cm) ellipse (1cm and 1cm);
\node (ML) at (0,-1.8cm){$\ML$};
\node[color=red] at (0.1cm,-2.4cm) {$\sfH^\Box$};
\end{tikzpicture}
\begin{tikzpicture}[thick,scale=0.8]
\draw (0,-.3) ellipse (1.5cm and 2.5cm);
\node (MTL) at (0,1.75cm){$\FO(\negg)$};
\node[color=red] at (0,1.2cm) {$\sfH\sfU\sfL\sfS\sfQ$};
\draw (0,-0.99cm) ellipse (1.27cm and 1.81cm);
\node (BML) at (0,0.2cm){$\calB(\FO)$};
\node[color=red] at (0,-.4cm) {$\sfH\sfU\sfL$};
\draw (0,-1.8cm) ellipse (1cm and 1cm);
\node (ML) at (0,-1.8cm){$\FO$};
\node[color=red] at (0cm,-2.4cm) {$\sfH$};
\end{tikzpicture}
\caption{Lower arrow: Axiomatization of $\calB(\cdot)$ by adding the propositional axioms $\sfL$ to $\sfH^0,\sfH^\Box$ and $\sfH$ (Section~\ref{sec:boolean} and \ref{sec:fo}).
Upper arrow: Axiomatization of $\QPPTL, \MTL$ and $\FO(\negg)$ by reduction to the $\calB(\cdot)$ fragment (Section~\ref{sec:splitting}, \ref{sec:modal} and \ref{sec:qbf}).\label{fig:overview}}
\end{figure}

However, a crucial difference between between first-order logic and propositional or modal logic is the existence of \emph{sentences}.
They complicate the task of finding a complete proof system for $\calB(\FO)$.
In Section~\ref{sec:fo} this obstacle is overcome by adding the so-called \emph{unanimity axiom} $\sfU$.

Another corollary of the operator elimination is a purely syntactical proof of Galliani's theorem \cite{Galliani14} that (on non-empty teams) every $\FO(\negg)$-sentence is equivalent to an $\FO$-sentence.
Finally, in Section~\ref{sec:fragments}, we consider axiomatizable sublogics in the propositional and modal settings, in particular, logics of dependence, independence, inclusion, and exclusion.

\section{Preliminaries}

\label{sec:prelim}

We associate with every logic $\calL$ a triple $(\Phi_\calL, \frakA_\calL, \vDash_\calL)$, where $\Phi_\calL$ is the set of \emph{formulas} of $\calL$, $\frakA_\calL$ is the class of \emph{valuations}, and $\vDash_\calL$ is the \emph{satisfaction relation} between $\frakA_\calL$ and $\Phi_\calL$.
In what follows, we often omit $\calL$ in the components if the meaning is clear.

\medskip

Let $\Psi, \Theta \subseteq \Phi$ be sets of formulas, $\varphi,\psi \in \Phi$ formulas and $A \in \frakA$ a valuation.
$A \vDash \Psi$ means $A \vDash \psi$ for all $\psi \in \Psi$.
We say that $\Psi$ \emph{entails} $\Theta$, in symbols $\Psi \vDash \Theta$, if $A \vDash \Psi$ implies $A \vDash \Theta$ for all $A \in \frakA$.

We usually omit the set braces and simply write, \eg, $\varphi \vDash \psi$ instead of $\{\varphi\} \vDash \{\psi\}$.
Moreover, we write $\vDash \varphi$ and $\vDash \Psi$ instead of $\emptyset \vDash \varphi$ and $\emptyset \vDash \Psi$.
Finally, $\Psi \equiv \Theta$ means $\Psi \vDash \Theta$ and $\Theta \vDash \Psi$.
The class of valuations satisfying a formula $\varphi$, called \emph{models} of $\varphi$, is $\Mod{\varphi} \dfn \Set{A  \in \frakA | A \vDash \varphi}$.
The models $\Mod{\Psi}$ of a set $\Psi$ are defined similarly.
A formula $\varphi$ that has a model, \ie, $\Mod{\varphi} \neq \emptyset$, is called \emph{satisfiable}.
Dually, if $\Mod{\varphi} = \frakA$, then $\varphi$ is called \emph{valid} or \emph{tautology}.

A logic $\calL$ is \emph{compact} if for all $\Psi \subseteq \Phi$ it holds that $\Psi$ has a model if and only if every finite subset of $\Psi$ has a model.

For brevity, we will also write $\varphi \in \calL$ instead of $\varphi \in \Phi_\calL$ and $\Psi \subseteq \calL$ instead of $\Psi \subseteq \Phi_\calL$.
If two logics $\calL, \calL'$ share the same valuations, then $\calL \leq \calL'$ means that for every $\varphi \in \calL$ there is a $\varphi' \in \calL'$ such that $\Mod{\varphi} = \Mod{\varphi'}$. $\calL \equiv \calL'$ means $\calL \leq \calL'$ and $\calL'\leq \calL$.
By contrast, $\calL \subseteq \calL'$ means that $\Phi_\calL \subseteq \Phi_{\calL'}$, but the valuations and truth on the common formulas are identical. Then $\calL$ is a \emph{fragment} of $\calL'$.

\subsection{Propositional team logic}

Classical \emph{propositional logic} $\PL$ is built upon a countably infinite set $\PS$ of \emph{atomic propositions}, denoted by Latin letters $\{ a, b, c, \ldots \}$.
The syntax of \emph{quantified propositional logic} $\QPL$ is given as
\[
\alpha \ddfn \neg \alpha \mid (\alpha \imp \alpha) \mid \forall x \,\alpha \mid x \qquad (x \in \PS)\text{.}
\]
\emph{Propositional logic} $\PL$ is then the quantifier-free fragment of $\QPL$.
The valuations of $\QPL$ are Boolean \emph{assignments} $s \colon \PS \to \{ 1, 0 \}$ with the usual semantics.

\medskip

We extend $\QPL$ to \emph{quantified propositional team logic} $\QPTL$ \cite{hannula_complexity_2015,gandalf} and $\PL$ to \emph{propositional team logic} $\PTL$ \cite{yang_propositional_2017} as follows.
For clarity, in the following we reserve the letters $\alpha, \beta,\gamma,\ldots$ for classical formulas;
and we use $\varphi, \psi, \vartheta,\ldots$ for general formulas.

In team semantics, the valuations of $\QPL$ are sets $T$, called \emph{teams}, of propositional assignments.
For $\QPL$-formulas $\alpha$, a team $T$ satisfies $\alpha$ if and only if all its members satisfy it, \ie, $T \vDash \alpha$ if $\forall s \in T : s \vDash \alpha$.
In particular, the empty team satisfies every classical formula by definition.
$\QPTL$ now extends the syntax of $\QPL$ by
\[
\varphi \ddfn \alpha \mid \negg \varphi \mid (\varphi \timp \varphi) \mid (\varphi \limp \varphi) \mid \forall x \,\varphi \mid \shriek x \,\varphi \qquad (x \in \PS\text{, }\alpha \in \QPL)\text{.}
\]
$\PTL$ is then simply the quantifier-free fragment of $\QPTL$.

Note that, while $\QPL$ offers Boolean operators such as $\imp$ and $\neg$ on the level of single assignments, their meaning changes when switching to team semantics.
In particular, they do not longer correspond to the Boolean implication and negation.
For this reason, the \emph{strong negation} $\negg$ and the \emph{material implication} $\timp$ are introduced as Boolean operators on the level of teams.

On the other hand, the binary operator $\limp$ is another team-semantical generalization of the implication.
Unlike $\timp$, it is not a truth-functional connective, but quantifies over all possible \emph{divisions} of a team into two subteams.

For the semantics of team-wide quantifiers, we need the concept of \emph{supplementing functions}.
Given a team $T$, a supplementing function is a function $f \colon T \to \{\{0\},\{1\},\{0,1\}\}$.
The team $T^x_f \dfn \Set{ s^x_a | s \in T, a \in f(s)}$ is called \emph{supplementing team},
where $s^x_a(x) \dfn a$ and $s^x_a(y) \dfn s(y)$ for $y \in \PS \setminus \{ x\}$.
If $f(s) = A$ is constant, then we simply write $T^x_A$ instead of $T^x_f$.
\label{p:semantics}
The semantics of $\QPTL$ is then
\begin{alignat*}{2}
    &T \vDash \alpha&& \Leftrightarrow \forall s \in T : s \vDash_{\QPL} \alpha\text{,}\\
\intertext{for formulas $\alpha \in \QPL$, and otherwise}
    &T \vDash \negg\varphi&&\Leftrightarrow\; T \nvDash \varphi\text{,}\\
    &T \vDash \varphi \timp \psi&&\Leftrightarrow\; T \vDash \psi \text{ or }T\nvDash \varphi\text{,}\\
    &T \vDash \varphi \limp \psi\;&&\Leftrightarrow\; \text{for all } S, U \subseteq T : \text{ if }S \cup U = T \text{ and }S \vDash \varphi \text{, then }U \vDash \psi\text{,}\\
    &T \vDash \forall x \,\varphi &&\Leftrightarrow\; T^x_{\{0,1\}} \vDash \varphi \text{,}\\
    &T \vDash \shriek x \, \varphi &&\Leftrightarrow\; T^x_{f} \vDash \varphi \text{ for all  }f : T \to \{\{0\},\{1\},\{0,1\}\}\text{.}
\end{alignat*}

\subsection{Modal team logic}

Classical \emph{modal logic} $\ML$ extends the formulas of $\PL$ by the $\Box$-modality:
\[
\alpha \ddfn \neg \alpha \mid (\alpha \imp \alpha) \mid \Box \alpha \mid x \qquad (x \in \PS)
\]

Valuations of modal formulas are \emph{Kripke structures}, which are essentially labeled transition systems.
A \emph{frame} is a directed graph $(W, R)$ where $W$ is the set of \emph{worlds} or \emph{points} and $R \subseteq W \times W$ is a binary edge relation.
A Kripke structure $K = (W,R,V)$ then consists of a frame $(W,R)$ together with a labeling function $V \colon \PS \to \pow{W}$, where $\pow{\cdot}$ denotes the power set operation.
A \emph{pointed Kripke structure} is a pair $(K,w)$ where $K$ is a Kripke structure $(W,R,V)$ and $w \in W$ is its \emph{initial world} or \emph{initial state}.
$\ML$ then has the class of all pointed Kripke structures as evaluations and the usual Kripke semantics.

\emph{Modal team logic} $\MTL$ extends modal logic and was introduced by Müller~\cite{mtlthesis} and studied, \eg, by Kontinen et al.\ \cite{kontinen_van_2014}.
It is evaluated on pairs $(K,T)$, where $K = (W,R,V)$ is a Kripke structure, but with $T \subseteq W$ being a \emph{set} of worlds, called \emph{team (in $K$)}.
$\MTL$ extends $\ML$ by
\[
\varphi \ddfn \alpha \mid \negg \varphi \mid (\varphi \timp \varphi) \mid (\varphi \limp \varphi) \mid \Box \varphi \mid \triangle \varphi\qquad \text{($\alpha \in \ML$).}
\]
Let us turn to its semantics.
If $(W,R,V)$ is a Kripke structure and $T \subseteq W$ a team, then we define its \emph{image} $RT \dfn \Set{ w \in W | \exists v \in T : Rvw }$ and \emph{pre-image} $R^{-1}T\dfn\Set{ w \in W | \exists v \in T : Rwv }$.
A \emph{successor team} $T'$ of $T$ is a team such that $T' \subseteq RT$ and $T \subseteq R^{-1}T$.
Intuitively, a successor team of $T$ can be obtained by picking at least one successor of every element in $T$.
If $\alpha \in \ML$, $K = (W,R,V)$ and $T \subseteq W$, then
\begin{alignat*}{3}
    &(K,T) \vDash \alpha&& \Leftrightarrow\; &&\forall w \in T : (K,w) \vDash_{\ML} \alpha\text{,}\\
\intertext{for formulas $\alpha \in \ML$, and otherwise}
    &(K,T) \vDash \negg\varphi&&\Leftrightarrow &&(K,T) \nvDash \varphi\text{,}\\
    &(K,T) \vDash \varphi \timp \psi&&\Leftrightarrow &&(K,T) \vDash \psi \text{ or }(K,T)\nvDash \varphi\text{,}\\
    &(K,T) \vDash \varphi \limp \psi\;&&\Leftrightarrow &&\text{for all } S, U \subseteq T : \text{ if }S \cup U = T \text{ and }(K,S) \vDash \varphi \text{,}\\
    & && && \text{then }(K,U) \vDash \psi\text{,}\\
    &(K,T) \vDash \Box \varphi &&\Leftrightarrow &&(K,RT) \vDash \varphi \text{,}\\
    &(K,T) \vDash \triangle \varphi &&\Leftrightarrow && (K,S) \vDash \varphi \text{ for all successor teams }S\text{ of }T\text{.}
\end{alignat*}

\subsection{First-order team logic}

Classical \emph{first-order logic} $\FO$ consists of \emph{terms} and \emph{formulas} over some vocabulary $\tau$
of relation symbols $R_i$ and function symbols $f_i$, each with their respective arity $\arity{R_i} \geq 0$, $\arity{f_i}\geq 0$.
A function symbol of arity zero is a \emph{constant symbol} and usually denoted by $c$.

Let $\Var \dfn \{ x_1, x_2, \ldots \}$ be a countably infinite set of first-order variables.
We define the syntax of $\FO$ by
\begin{align*}
\alpha \ddfn \neg \alpha \mid (\alpha \imp \alpha) \mid \forall x \,\alpha \mid t = t \mid R(t_1,\ldots,t_{\arity{R}})\quad\text{(for $x \in \Var$, $\tau$-terms $t_i$)}
\end{align*}

A formula without free variables is called \emph{closed}.

Formulas are evaluated in the classical Tarski semantics.
We require \emph{first-order structures} $\calA = (A, \tau^\calA)$, consisting of a non-empty \emph{domain} $A$ (also denoted by $|\calA|$) and interpretations $\tau^\calA$ for the vocabulary.
A \emph{first-order interpretation} is then a pair $(\calA, s)$ of a structure $\calA = (A, \tau^\calA)$ and an assignment $s \colon \Var \to A$.
Given a term $t$, it evaluates to an element of $A$, denoted by $t^{(\calA,s)}$.

A first-order formula without free variables is called \emph{sentence}.
The set of all sentences is $\FO^0$.
If $\alpha \in \FO^0$, then we sometimes write $\calA \vDash \alpha$ instead of $(\calA,s) \vDash \alpha$.

\smallskip

\emph{First-order team logic} $\TL$ was introduced by Väänänen~\cite{vaananen_dependence_2007}.
We define it via the following syntax, where $\alpha \in \FO$, $t_1,\ldots,t_n$ are terms, $n \geq 1$, and $x \in \Var$.
\[
\varphi \ddfn \alpha \mid \negg \varphi \mid (\varphi \timp \varphi) \mid (\varphi \limp \varphi) \mid \forall x \,\varphi \mid \shriek x \,\varphi \mid \dep{t_1,\ldots,t_n}\text{.}
\]
A $\TL$-valuation is a pair $(\calA,T)$, where $\calA$ is a first-order structure, and $T$ a \emph{team}, \ie, a set of assignments $s : \Var \to \size{\calA}$.
Note that $\dep{t_1,\ldots,t_n}$ is an atomic formula, called \emph{dependence atom} \cite{vaananen_dependence_2007}.
Intuitively it states that in the team $t_n$ is functionally determined by $t_1,\ldots,t_{n-1}$.

For the semantics of the quantifiers, we require the first-order analog of supplementing functions $f$ for a team $T$, formally $f \colon T \to \pow{\size{\cal{A}}}\setminus \{\emptyset\}$.
Then the supplementing team $T^x_f$ and the duplicating team $T^x_{\size{\calA}}$ are defined as in the propositional case.
The semantics of $\TL$, where $\alpha \in \FO$, $x \in \Var$ and $t_1,\ldots,t_n$ are terms, are
\begin{alignat*}{3}
    &(\calA,T) \vDash \alpha&& \Leftrightarrow \; &&\forall s \in T : (\calA,s) \vDash_{\FO} \alpha\text{,}\\
    \intertext{for $\alpha \in \FO$, and otherwise}
    &(\calA,T) \vDash \dep{t_1,\ldots,t_n} && \Leftrightarrow \; &&\forall s,s' \in T : \text{if }t_i^{(\calA,s)} = t_i^{(\calA,s')}\text{ for all }1 \leq i < n\text{,}\\
    & && &&\text{then also }t_n^{(\calA,s)} = t_n^{(\calA,s')}\text{,}\\
    &(\calA,T) \vDash \negg\varphi&&\Leftrightarrow&& (\calA,T) \nvDash \varphi\text{,}\\
    &(\calA,T) \vDash \varphi \timp \psi&&\Leftrightarrow&& (\calA,T) \vDash \psi \text{ or }(\calA,T)\nvDash \varphi\text{,}\\
    &(\calA,T) \vDash \varphi \limp \psi\;&&\Leftrightarrow&& \text{for all } S, U \subseteq T : \text{ if }S \cup U = T \\
    & && && \text{ and }(\calA,S) \vDash \varphi\text{, then }(\calA,U) \vDash \psi\text{,}\\
    &(\calA,T) \vDash \forall x \,\varphi &&\Leftrightarrow &&T^x_{\size{\calA}} \vDash \varphi \text{,}\\
    &(\calA,T) \vDash \shriek x \, \varphi &&\Leftrightarrow&& T^x_{f} \vDash \varphi \text{ for all  }f : T \to \pow{\size{\calA}} \setminus \{\emptyset\} \text{.}
\end{alignat*}

\smallskip

If $\varphi$ is a formula and $\sigma,\sigma'$ are formulas or terms, then $\varphi[\sigma/\sigma']$ denotes the formula obtained from $\varphi$ by substituting in parallel every occurrence of $\sigma$ by $\sigma'$.

\medskip

We require some abbreviations.
For the truth-functional constants in propositional, modal logic, and first-order logic, we use the abbreviations $\top \dfn (x \imp x)$ resp.\ $\top \dfn \exists x x = x$ and $\bot \dfn \neg \top$, where $x$ is a fixed proposition resp.\ variable.
Moreover, we write $\alpha \lor \beta \dfn \neg \alpha \imp \beta$ for the disjunction, $\alpha \land \beta \dfn \neg(\alpha \imp \neg \beta)$ for the conjunction, and $\alpha \equ \beta \dfn (\alpha \imp \beta) \land (\beta \imp \alpha)$ for the equivalence.

Note that the above connectives are not truth-functional under team-semantics:
for example, $p \lor \neg p$ does not imply that either $p$ or $\neg p$ holds in the team.
Moreover, $\bot$ is true in the empty team.

For this reason, we define the \emph{strong falsum} $\falsum \dfn \negg \top$ which is false in all teams.
Likewise, we define proper Boolean connectives based on $\negg$ and $\timp$, \ie, $(\varphi \ovee \psi) \dfn (\negg\varphi\timp\psi)$ (disjunction), $(\varphi \owedge \psi) \dfn \negg(\varphi \timp \negg\psi)$ (conjunction) and $(\varphi \tequiv \psi) \dfn (\varphi \timp \psi) \owedge (\psi \timp \varphi)$ (equivalence).
Dually to the team-semantical interpretation of classical formulas, $\E\alpha \dfn \negg \neg \alpha$ expresses that at least one element of the team satisfies $\alpha$.

Moreover, we abbreviate $\varphi \tens \psi \dfn \negg(\varphi \limp \negg\psi)$.
The meaning of $\varphi \tens \psi$ is that the current team permits \emph{some} division into subteams $S$ satisfying $\varphi$ and $U$ satisfying $\psi$ (cf.\ \cite{vaananen_dependence_2007,yang_propositional_2017}).

We assume $\imp$, $\timp$ and $\limp$ as right-associative and $\land, \owedge, \lor, \ovee, \tens$ as left-associative.

Furthermore, the dual of the $\triangle$-modality is defined as $\Diamond\varphi\dfn\negg \triangle \negg \varphi$, which is true if \emph{some} successor team satisfies $\varphi$.
Likewise, the dual of $\shriek$ is $\exists$, \ie, $\exists x \, \varphi \dfn \negg \shriek x \negg \varphi$, which is true if there is \emph{some} supplementing function $f$ such that $T^x_f$ satisfies $\varphi$.

\smallskip

Väänänen's \emph{dependence logic} $\D$ can then be defined as the fragment of $\TL$ that is the closure of $\FO$ and $\depop$ under only $\tens,\oland,\exists$ and $\forall$ \cite{vaananen_dependence_2007}.
The logic $\FO(\negg)$ is then simply the $\depop$-free fragment of $\TL$.

\smallskip

Note that propositional and modal team logic possess a dependence atom as well, written $\dep{\alpha_1,\ldots,\alpha_n}$ for \emph{formulas} $\alpha_1,\ldots,\alpha_n$ instead of terms \cite{yang_propositional_2016,emdl,vaananen_modal_2008}.
It has the meaning that every subteam uniform in the truth of $\alpha_1,\ldots,\alpha_n$ is also uniform in the truth of $\beta$, in other words, that the truth of $\beta$ is a function of that of $\alpha_1,\ldots,\alpha_n$.
However, as this atom can expressed as
\[
\top \limp \left[\bigoland_{i=1}^n(\alpha_i \ovee \neg \alpha_i) \timp (\beta \ovee \neg \beta)\right]\text{,}
\]
we do not regard it as part of the syntax of $\QPPTL$ or $\MTL$ (see also Section~\ref{sec:fragments}).

\smallskip

In the first-order setting, the dependence atom cannot be defined by the other operators.
As $\D$ and $\TL$ are not axiomatizable \cite{vaananen_dependence_2007}, we will focus on the $\depop$-free fragment $\FO(\negg)$.

\smallskip

Note that propositional and modal team logics are often defined in the literature using only literals $p, \neg p$ as atoms.
The classical operators $\lor,\land,\Box,\Diamond,\forall$ and $\exists$ are then the primitive connectives (see \eg Väänänen, Sano, and Virtema \cite{sano_et_al,vaananen_modal_2008,vaananen_dependence_2007}).

The rationale behind deviating from this convention is twofold.
First, embedding the classical logics, not necessarily in negation normal form, as a "layer" in team logic allows to comfortably build onto existing proof systems.
Second, we make extensive use of introduction rules such as $\varphi \vdash \Box \varphi$.
Such rules would be unsound for $\Diamond,\tens$ and $\exists$.
For these reasons, we prefer $\triangle$, $\limp$ and $\shriek$ as primitive connectives.

\subsection{Proof systems}

A proof system is a tuple $\Omega = (\Xi, \Psi, I)$ where $\Xi$ is a set of \emph{judgments} (usually $\calL$-formulas), $\Psi \subseteq \Xi$ is a set of \emph{axioms}, and $I \subseteq \pow{\Xi} \times \Xi$ is a set of \emph{inference rules}.
Throughout this paper, $\Xi$, $\Psi$ and $I$ are all assumed countable and efficiently decidable.
The component-wise union of two systems $\Omega$, $\Omega'$ is written $\Omega\Omega'$.

An \emph{$\Omega$-proof $\calP$} of $\varphi \in \Xi$ from $\Phi \subseteq \Xi$ is a finite sequence $\calP = (P_0, \ldots, P_n)$ of finite sets $P_i \subseteq \Xi$ such that $\varphi \in P_n$, and $\xi \in P_i$ implies that either $\xi \in P_{i-1} \cup \Psi \cup \Phi$, or $(P_{i-1}', \xi) \in I$ for some $P_{i-1}' \subseteq P_{i-1}$.
We write $\Phi \vdash_\Omega \varphi$ if there is some $\Omega$-proof of $\varphi$ from $\Phi$, and usually omit $\Omega$ if it is clear.

If two formulas $\varphi$ and $\varphi'$ prove each other, \ie, $\{\varphi\} \vdash \varphi'$ and $\{\varphi'\}\vdash \varphi$, then we write $\varphi \dashv\vdash \varphi'$.
For sets we write $\Phi \dashv\vdash \Phi'$ if for every $\varphi \in \Phi$ it holds $\Phi' \vdash \varphi$, and for every $\varphi \in \Phi'$ it holds $\Phi \vdash \varphi$.
Instead of $\emptyset \vdash \varphi$, we also write $\vdash \varphi$.

A system $\Omega$ is \emph{sound} for a logic $\calL$ if for all $\Phi \subseteq \calL$ and $\varphi \in \calL$ it holds that $\Phi \vdash_\Omega \varphi$ implies $\Phi \vDash_\calL \varphi$.
It is \emph{complete} if $\Phi \vDash_\calL \varphi$ implies $\Phi \vdash_\Omega \varphi$.
Note that every logic $\calL$ with a sound and complete proof system is compact.

\medskip

The proof systems in this paper are based on the common Hilbert-style axiomatizations of propositional, modal and first-order logic (Figure~\ref{fig:hilbert-base}).
We use the system $\sfH^0$ (\Deriv{(A1)}--\Deriv{(A3)} and \Deriv{(E$\imp$)} for $\PL$, the system $\sfH^1$ ($\sfH^0$ and \Deriv{(A4)}) for $\QPL$, the system $\sfH^\Box$ ($\sfH^0$, \Deriv{(K)} and \Deriv{(Nec)}) for $\ML$, and the system $\sfH$ ($\sfH^0$, \Deriv{(A5)}--\Deriv{(A8}) and \Deriv{(UG$\forall$)}) for $\FO$.
More precisely, \Deriv{(A1)} stands for all substitution instances of the schema \Deriv{(A1)}, and so on.

Let us explain the inference rules \Deriv{(Nec)} and \Deriv{(UG$\forall$)}.
In contrast to \Deriv{(E$\imp$)}, they cannot be applied to arbitrary derived formulas.
Instead, they can only be applied to \emph{theorems} $\alpha$, meaning that $\alpha$ was derived from the axioms of the system only, not using any premises.
This ensures that the \emph{deduction theorem} holds.\footnote{For instance, if $\alpha$ is a tautology, then so is $\Box\alpha$, however, $\alpha \nvDash \Box\alpha$.
This phenomenon is discussed by Fitting and Hakli and Negri \cite{deduction_fail,patrick_blackburn_2_2007} as the "failure of the deduction theorem" in modal logic, and one way to remedy it is exactly the restriction of $\alpha$ to be a theorem.}

\begin{figure}[t]
\centering
\begin{tabular}{ll}\toprule
\Deriv{(A1)}&$\alpha \imp (\beta \imp \alpha)$\\
\Deriv{(A2)}&$(\alpha \imp (\beta \imp \gamma)) \imp (\alpha \imp \beta) \imp (\alpha \imp \gamma)$\\
\Deriv{(A3)}&$(\neg\alpha \imp \neg\beta)\imp(\beta \imp \alpha)$\\
\midrule
\Deriv{(A4)}&$\forall x \, \alpha \leftrightarrow (\alpha[x/\top] \land \alpha[x/\bot])$\\
\midrule
\Deriv{(A5)}&$\forall x \, \alpha \imp \alpha[x/t]$, $t$ term\\
\Deriv{(A6)}&$\forall x \, (\alpha \imp \beta) \imp (\alpha \imp \forall \, x \beta$), $x$ not free in $\alpha$\\
\Deriv{(A7)}&$x = x$\\
\Deriv{(A8)}&$x = y \imp (\alpha \imp \alpha[x/y])$\\
\midrule
\Deriv{(K)}&$\Box(\alpha\imp\beta) \imp (\Box\alpha \imp \Box\beta)$\\
\midrule
\Deriv{(E$\imp$)}&\begin{bprooftree}
\AxiomC{$\alpha$}
\AxiomC{$\alpha \imp \beta$}
\BinaryInfC{$\beta$}
\end{bprooftree}\vspace{8pt}\\
\Deriv{(Nec)}&\begin{bprooftree}
\AxiomC{$\alpha$}
\RightLabel{\small{}($\alpha$ theorem)}
\UnaryInfC{$\Box \alpha$}
\end{bprooftree}\vspace{8pt}\\
\Deriv{(UG$\forall$)}&\begin{bprooftree}
\AxiomC{$\alpha$}
\RightLabel{\small{}($\alpha$ theorem, $t$ term)}
\UnaryInfC{$\forall x \, \alpha[t/x]$}
\end{bprooftree}\vspace{8pt}\\
\bottomrule
\end{tabular}
\caption{Hilbert-style axiomatizations of $\QPPL$, $\ML$ and $\FO$}\label{fig:hilbert-base}
\end{figure}

\begin{proposition}\label{prop:base-completeness}
$\sfH^0$ is sound and complete for $\PL$. $\sfH^1$ is sound and complete for $\QPL$. $\sfH^\Box$ is sound and complete for $\ML$. $\sfH$ is sound and complete for $\FO$.
\end{proposition}

We defined classical logics to have the \emph{flatness} property, \ie, a classical formula is satisfied by a team $T$ under team semantics exactly when all of $T$'s members satisfy it in classical semantics.
In the following, we emphasize this by referring to flat logics as $\calF$.

Flatness has one particularly useful consequence regarding proof systems:

\begin{proposition}\label{prop:equal-semantics}
Let $\calF \in \{ \PL, \QPL, \ML, \FO\}$, $\Gamma \subseteq \calF, \alpha \in \calF$.
Then $\Gamma \vDash \alpha$ holds in classical semantics if and only if it holds under team semantics.
\end{proposition}
\begin{proof}
We prove only the $\FO$ case.
The other cases are proven similarly.
For "$\Rightarrow$", let $\Gamma \vDash \alpha$ in classical semantics.
Suppose that an arbitrary valuation $(\calA, T)$ satisfies $\Gamma$.
Then $(\calA, s)\vDash \Gamma$ for all $s\in T$.
By assumption, $(\calA, s) \vDash \alpha$ in for all $s\in T$.
Consequently, $(\calA, T)\vDash \alpha$.

Next, we prove "$\Leftarrow$" by contraposition.
If $\Gamma \nvDash \alpha$ in classical semantics, then there is a valuation $(\calA,s)$ such that $(\calA,s)\vDash \Gamma$ and $(\calA,s)\nvDash \alpha$.
But then also $(\calA,\{s\}) \vDash \Gamma$ and $(\calA,\{s\}) \nvDash \alpha$.
Consequently, $\Gamma \nvDash \alpha$ under team semantics.
\end{proof}

\begin{corollary}\label{cor:base-completeness-team}
Under team semantics, the systems $\sfH^0$, $\sfH^1$, $\sfH^\Box$ and $\sfH$ are sound and complete for $\PL$, $\QPL$, $\ML$ and $\FO$, respectively.
\end{corollary}

Accordingly, we will not distinguish between the two entailment relations for the rest of the paper. Other immediate consequences of flatness are the following:

\begin{proposition}\label{prop:downward-closure}
The logics $\QPL$, $\ML$ and $\FO$ are \emph{downward closed}:

If $\alpha \in \QPL$ and $T \vDash \alpha$, then $T' \vDash \alpha$ for all $T' \subseteq T$.

If $\alpha \in \ML$ and $(K,T) \vDash \alpha$, then $(K,T') \vDash \alpha$ for all $T' \subseteq T$.

If $\alpha \in \FO$ and $(\calA,T) \vDash \alpha$, then $(\calA,T') \vDash \alpha$ for all $T' \subseteq T$.
\end{proposition}

\begin{proposition}\label{prop:union-closure}
The logics $\QPL$, $\ML$ and $\FO$ are \emph{union closed}:

Let $\calT$ be a set of teams.

If $\alpha \in \QPL$ and $T \vDash \alpha$ for all $T \in \calT$, then $\bigcup\calT \vDash \alpha$.

If $\alpha \in \ML$ and $(K,T) \vDash \alpha$ for all $T \in \calT$, then $(K,\bigcup\calT) \vDash \alpha$.

If $\alpha \in \FO$ and $(\calA,T) \vDash \alpha$ for all $T \in \calT$, then $(\calA,\bigcup\calT) \vDash \alpha$.
\end{proposition}

\begin{proposition}\label{prop:flatness-tens}
Let $\calF \in \{\QPL, \ML, \FO\}$ and $\alpha, \beta \in \calF$.
Then $\alpha \lor \beta \equiv \alpha \tens \beta$.
\end{proposition}

\section{Axioms of the Boolean closure}\label{sec:boolean}

We begin the development of a proof system for team logic with the operators $\timp$ and $\negg$, \ie, for the Boolean closure of classical logic under team semantics.

\begin{definition}
If $\calF$ is a logic, then $\calB(\calF)$ is the \emph{Boolean closure} of $\calF$, defined by the grammar
$\varphi \ddfn \alpha \mid \negg\varphi \mid \varphi \timp \varphi$,
where $\alpha \in \calF$, and with the semantics
\begin{alignat*}{2}
    &A \vDash \alpha&& \Leftrightarrow A \vDash_{\calF} \alpha\;\text{ for $\alpha \in \calF$,}\\
    &A \vDash \negg\varphi&&\Leftrightarrow A \nvDash \varphi\text{,}\\
    &A \vDash \varphi \timp \psi&&\Leftrightarrow A \vDash \psi \text{ or }A\nvDash \varphi\text{.}
\end{alignat*}
\end{definition}

\begin{figure}
\centering
\begin{tabular}{ll}\toprule
\Deriv{(L1)}&$\varphi \timp (\psi \timp \varphi)$\\
\Deriv{(L2)}&$(\varphi \timp (\psi \timp \vartheta)) \timp (\varphi \timp \psi) \timp (\varphi \timp \vartheta)$\\
\Deriv{(L3)}&$(\negg\varphi \timp\negg\psi)\timp(\psi \timp\varphi)$\\
\midrule
\Deriv{(L4)}&$(\alpha\imp\beta)\timp (\alpha\timp\beta) \qquad\qquad \alpha,\beta \in \calF$\\
\midrule
\Deriv{(E$\timp$)}&\begin{bprooftree}
\AxiomC{$\varphi$}
\AxiomC{$\varphi \timp \psi$}
\BinaryInfC{$\psi$}
\end{bprooftree}\\
\bottomrule
\end{tabular}
\caption{The system $\sfL$\label{fig:boolean-axioms}}
\end{figure}

In this section, we develop a "template" proof system for $\calB(\cdot)$, viz.\ the system $\sfL$ \emph{(lifted propositional axioms)} depicted in Figure~\ref{fig:boolean-axioms}.
The axioms of $\sfL$ mostly resemble their classical counterparts in $\sfH^0$.
One exception is \Deriv{(L4)}, which relates the propositional and the team-semantical implication.
We demonstrate that a complete proof system for a logic $\calF$ can be augmented with $\sfL$ to obtain a complete system for $\calB(\calF)$.
This procedure, however, can only be a first step to full axiomatizations of $\QPTL$, $\MTL$ and $\FO(\negg)$, since clearly $\calB(\QPL) \subsetneq \QPTL$, $\calB(\ML) \subsetneq \MTL$, and $\calB(\FO) \subsetneq \FO(\negg)$.

While the systems $\sfH^0$, $\sfH^1$, $\sfH^\Box$ and $\sfH$ can only be applied to classical formulas $\alpha,\beta,\gamma,\ldots$, the axioms and rules in $\sfL$ are permitted for general team-logical formulas $\varphi,\psi,\vartheta,\ldots$.

\begin{figure}[b]
\centering
\fitchprf{
\pline[A ]{\xi \timp \alpha} \\
\pline[B ]{\xi \timp (\alpha \imp \beta)}
}
{
\pline[1 ]{(\alpha \imp \beta) \timp (\alpha \timp \beta)}[\Deriv{(L4)}]\\
\pline[2 ]{\xi \timp ((\alpha \imp \beta) \timp (\alpha \timp \beta))}[\Deriv{(L1)}]\\
\pline[3 ]{\xi \timp (\alpha \timp \beta)}[\Deriv{(L2)}, B, 2]\\
\pline[\slider]{\xi \timp \beta}[\Deriv{(L2)}, A, 3]
}
\caption{Example derivation in $\sfL$\label{fig:example-deriv}}
\end{figure}

Derivations are written down as in the example below (Figure~\ref{fig:example-deriv}).
The premises have the special line numbers A, B, \textellipsis, whereas \slider marks the conclusion.
The right column of each proof shows the applied rules with the line numbers of the arguments.
The format is
\[
(\mathsf{rule}_1), \ldots, (\mathsf{rule}_n), \; \mathsf{argument}_1,\ldots,\mathsf{argument}_n
\]
where the line numbers of the arguments are omitted if only the preceding lines are used.
The "rule" $\sfL$ means that several axioms and rules of the system $\sfL$ are used without stating the exact steps ($\sfL$ proves all Boolean tautologies, as shown later in Theorem~\ref{thm:completeness-bool}).
For the sake of readability, we omit most applications of modus ponens \Deriv{(E$\timp$)} in $\sfL$.

\subsection{The deduction theorem for team logics}

The first step to prove $\sfL$ complete is to establish a variant of the deduction theorem, \ie, that $\Phi \vdash (\varphi \timp \psi)$ if and only if $\Phi \cup \{\varphi\} \vdash \psi$.
The crucial point is that the deduction theorem implies \emph{Lindenbaum's lemma}, which permits the construction of maximal consistent sets required for the completeness proof.
We begin by identifying a family of proof systems that guarantee a deduction theorem, based on the ideas of Hakli and Negri \cite{deduction_fail}.

\begin{definition}
Let $\Omega = (\Xi, \Psi, I)$ be a proof system. A rule $(\{\xi_1,\ldots,\xi_k\},\psi) \in I$ \emph{has conditionalization} if $\Set{(\varphi \timp \xi_i) | 1 \leq i \leq k} \vdash (\varphi \timp \psi)$ for all $\varphi \in \Xi$.
\end{definition}

In other words, the rule can also be applied under arbitrary premises $\varphi$.
We say that a system $\Omega$ has conditionalization if all inference rules have it.

\begin{lemma}\label{lem:deduction-theorem}
If $\Omega$ is a proof system and $\Omega\sfL$ has conditionalization, then the deduction theorem holds for $\Omega\sfL$, \ie, $\Phi \vdash_{\Omega\sfL} (\varphi \timp \psi)$ if and only if $\Phi \cup \{\varphi\} \vdash_{\Omega\sfL} \psi$.
\end{lemma}

\begin{proof}
"$\Rightarrow$" is clear, as $\sfL$ has \Deriv{(E$\timp$)}.
We prove "$\Leftarrow$" by induction on the length $n$ of a shortest proof of $\psi$.
If $\psi \in \Phi$, $\psi = \varphi$, or if $\psi$ is an axiom, then $\Phi \vdash (\varphi \timp \psi)$ by \Deriv{(L1)} and \Deriv{(E$\timp$)}.
For $n = 1$ these are the only cases.
If $n > 1$, then $\psi$ could also be obtained by application of some inference rule $(\{\xi_1, \ldots, \xi_k\},\psi)$.
But then $\xi_1, \ldots, \xi_k$ each have a proof of length $\leq n-1$ from $\Phi \cup \{\varphi\}$, so by induction hypothesis, $\Phi \vdash \varphi \timp \xi_i$ for $1 \leq i \leq k$.
As $\Omega\sfL$ has conditionalization by assumption, $\Phi \vdash \varphi \timp \psi$.
\end{proof}

\begin{definition}
Let $\Omega$ and $\Omega'$ be proof systems.
$\Omega'$ is a \emph{conservative extension of} $\Omega$, in symbols $\Omega' \succeq \Omega$, if $\Omega'$ contains all judgments, rules, and axioms of $\Omega$, but all rules of $\Omega'$ that are not in $\Omega$ produce only theorems.
\end{definition}

For instance, $\sfH \succeq \sfH^0$ and $\sfH^\Box \succeq \sfH^0$, as the only rule possibly producing non-theorems, \Deriv{(E$\imp$)}, is already present in $\sfH^0$.
Note that $\succeq$ is a partial ordering.

\begin{theorem}\label{thm:ext-deduction}
Every conservative extension $\Omega$ of $\sfL$ or $\sfH^0\sfL$ has the deduction theorem.
\end{theorem}
\begin{proof}
By Lemma~\ref{lem:deduction-theorem}, it suffices to show that all inference rules of $\Omega$ have conditionalization.
There are three cases to distinguish:
the rule \Deriv{(E$\imp$)} in $\sfH^0$ (if $\Omega\succeq \sfH^0\sfL$), the rule \Deriv{(E$\timp$)} in $\sfL$, and an arbitrary rule that produces only theorems.
The latter case is clear, as for every theorem $\psi$, by \Deriv{(L1)} and \Deriv{(E$\timp$)} we can prove $\xi \timp \psi$ for arbitrary $\xi$.

Next, consider the rule \Deriv{(E$\timp$)} $= (\{\varphi, \varphi \timp \psi\}, \psi)$.
To demonstrate that it has conditionalization, we assume the premises $\xi \timp (\varphi \timp \psi)$ and $\xi \timp \varphi$, where $\xi$ is arbitrary.
By \Deriv{(L2)} and \Deriv{(E$\timp$)}, it is straightforward to derive $\xi \timp \psi$.
Finally, for \Deriv{(E$\imp$)}, conditionalization is proven in Figure~\ref{fig:example-deriv}.
\end{proof}

\subsection{Completeness of the Boolean closure}

As standard completeness proofs often use Lindenbaum's lemma to construct a \emph{maximal consistent set}, let us introduce the notion of consistency.

\begin{definition}
Let $\Omega = (\Xi, \Psi, I)$ be a proof system.
A set $\Phi$ is $\Omega$-\emph{inconsistent} if $\Phi \vdash \Xi$. $\Phi$ is $\Omega$-\emph{consistent} if it is not $\Omega$-inconsistent.
Moreover, $\Phi \subseteq \Xi$ is \emph{maximal $\Omega$-consistent} if it is $\Omega$-consistent and contains $\xi$ or $\negg\xi$ for every $\xi \in \Xi$.
\end{definition}

As before, we usually omit $\Omega$.
The following lemmas are standard, with their proofs also found in the appendix.

\begin{lemma}\label{lem:only-one-consistent}
Let $\Omega \succeq \sfL$ and let $\Phi$ be consistent.
Then $\Phi \nvdash \varphi$ implies that $\Phi \cup \{ \negg \varphi \}$ is consistent, and $\Phi \vdash \varphi$ implies that $\Phi \cup \{\varphi\}$ is consistent.
\end{lemma}

\begin{lemma}[Lindenbaum's Lemma]\label{lem:lindenbaum}
	If $\Omega \succeq \sfL$, then every $\Omega$-consistent set has a maximal $\Omega$-consistent superset.
\end{lemma}

The next step in standard completeness proofs is to construct an explicit model for any maximal consistent set.
The application of Lindenbaum's lemma is usually as follows:
if $\Phi$ is maximal consistent, then there is a model $M$ fulfilling all its atomic formulas.
By the maximality of $\Phi$, then also all Boolean combinations of atomic formulas, if they are in $\Phi$, are true in $M$.
The latter "inductive step" works as well for $\timp, \negg$.
However, more work is required for the induction basis---to construct the model $M$ that satisfies the atomic formulas.
The reason is that in our context an "atom" is, in fact, any formula of the underlying classical logic, such as $\QPL$, $\ML$ or $\FO$.
Due to this complication, we require the next property.

\begin{definition}
A logic $\calF$ admits \emph{counter-model merging} if, for arbitrary sets $\Gamma,\Delta \subseteq \calF$ the following holds:
Suppose that for every $\delta \in \Delta$ there is a model $M$ of $\Gamma$ such that $M \nvDash \delta$.
Then $\Gamma$ also has a model $M$ that falsifies every formula in $\Delta$.
\end{definition}

\begin{proposition}\label{prop:counter-models-ml-pl}
$\PL$, $\QPL$ and $\ML$, under team semantics, admit counter-model merging.
\end{proposition}
\begin{proof}
We prove the $\ML$ case.
Let $\Gamma, \Delta \subseteq \ML$, and for each $\delta \in \Delta$, let $(K_\delta, T_\delta)$ be a model of $\Gamma \cup \{\negg \delta\}$.
\Wloss the structures $K_\delta$ are pairwise disjoint; let $\calK$ denote their union.
The truth of $\ML$-formulas is invariant under disjoint union of structures \cite{Goranko2007249}; hence $(K_\delta,w) \vDash \alpha$ if and only if $(\calK,w) \vDash \alpha$ for all formulas $\alpha \in \ML$ and $w \in T_\delta$.
From the flatness property of $\ML$ it follows $(\calK,T_\delta) \vDash \Gamma$ and $(\calK,T_\delta) \nvDash \delta$.
Finally, consider the team $\calT \dfn \bigcup_{\delta \in \Delta} T_\delta$.
As $\ML$ is union closed (Proposition~\ref{prop:union-closure}), $(\calK,\calT)$ satisfies $\Gamma$, and as it is downwards closed (Proposition~\ref{prop:downward-closure}), $(\calK,\calT)$ falsifies each $\delta \in \Delta$.
\end{proof}

Let $\negg \calF$ denote the fragment of $\calB(\calF)$ that is restricted to the formulas in $\Set{ \negg \alpha | \alpha \in \calF }$.
Likewise, $\calF \cup \negg \calF$ denotes the fragment restricted to formulas in $\Set{\alpha,\negg\alpha | \alpha \in \calF}$.
Intuitively, $\calF \cup \negg \calF$ is the set of "literals."

\begin{definition}
A proof system $\Omega$ is \emph{refutation complete} for $\calL$ if for every unsatisfiable $\Phi \subseteq \calL$ there is a formula $\varphi$ such that $\Phi \vdash \{\varphi, \negg\varphi\}$.
\end{definition}

\begin{lemma}\label{lem:model-existence1}
If $\calF$ admits counter-model merging and $\Omega$ is complete for $\calF$, then $\Omega$ is refutation complete for $\calF \cup \negg \calF$.
\end{lemma}
\begin{proof}
Let $\Phi \subseteq \calF \cup \negg \calF$ be unsatisfiable.
Let $\Gamma \dfn \Phi \cap \calF$ and $\Delta \dfn \Phi \cap \negg\calF$.
There exists $\negg \delta \in \Delta$ such that $\Gamma \cup \{\negg\delta\}$ is unsatisfiable, since otherwise $\Phi$ would be satisfiable by counter-model merging.
But then $\Gamma \vDash \delta$, which implies $\Gamma \vdash \delta$ by completeness of $\Omega$ for $\calF$.
Consequently, $\Phi \vdash \{\delta, \negg\delta\}$.
\end{proof}

Note that $\FO$ does not admit counter-model merging.
In Section~\ref{sec:fo}, we give a counter-example.
However, it is still possible to construct a proof system that is refutation complete for $\FO \cup \negg \FO$ by introducing an additional axiom.

\medskip

Let us emphasize again the difference to classical logics such as $\PL$.
The atoms of $\PL$ are the set $\PS$; the analogously defined fragment $\PS \cup \neg \PS$ of literals is immediately refutation complete, as a set $\Gamma \subseteq \Set{ p, \neg p | p \in \PS }$ is inconsistent only if contains $p, \neg p$ for some proposition $p$.
Since team logic constitutes another "layer" on top of classical logic, the issue of refutation completeness becomes non-trivial.

\medskip

With the atoms handled correctly by the proof system (by refutation completeness of $\calF \cup \negg \calF$), the induction step goes through as for classical logic:

\begin{theorem}[Completeness of $\sfL$]\label{thm:completeness-of-L}
If $\Omega \succeq \sfL$ is refutation complete for $\calF \cup \negg \calF$, then it is complete for $\calB(\calF)$.
\end{theorem}
\begin{proof}
Let $\Phi' \subseteq \calB(\calF)$ and $\varphi \in \calB(\calF)$.
For completeness, we have to show that $\Phi' \nvdash \varphi$ implies $\Phi' \nvDash \varphi$, or in other words, that $\Phi \dfn \Phi' \cup \{\negg\varphi\}$ has a model.
First note that, if $\Phi' \nvdash \varphi$, then $\Phi'$ is consistent, and by Lemma~\ref{lem:only-one-consistent}, $\Phi$ is consistent, too.

Any consistent set $\Phi$ has a maximal consistent superset $\Phi^*$ by Lemma~\ref{lem:lindenbaum}.
Clearly, $\Phi^* \cap (\calF \cup \negg \calF)$ is then consistent as well.
By refutation completeness of $\Omega$ for $\calF \cup \negg \calF$, it has a model $A$.
We show that $\psi \in \Phi^* \Leftrightarrow A\vDash \psi$ for all $\psi \in \calB(\calF)$.
In particular, $\Phi$ is then satisfiable, which proves the theorem.
That $\psi \in \Phi^* \Leftrightarrow A\vDash \psi$ holds for $\psi \notin (\calF \cup \negg \calF)$ can be proven by induction on the length of $\psi$ (see the appendix).
\end{proof}

Conversely, we show that $\sfL$ also preserves the soundness of existing systems:

\begin{lemma}
Suppose that $\calF$ does not use $\negg$ or $\timp$.
If $\Omega$ and \Deriv{(E$\imp$)} are sound for $\calF$, then $\Omega\sfL$ is sound for $\calB(\calF)$.
\end{lemma}
\begin{proof}
We show that all axioms and inference rules of $\Omega\sfL$ are sound.
Then the soundness of $\Omega\sfL$ is easily shown by induction on the length of proofs.
The axioms and rules of $\Omega$ apply only to $\calF$, and for this reason are sound by assumption.
As \Deriv{(E$\imp$)} is sound, $\{\alpha, \alpha\imp\beta\} \vDash \beta$ for all $\alpha,\beta \in \calF$.
Equivalently, $\{\alpha\imp\beta\}\vDash (\alpha \timp\beta)$, hence \Deriv{($\sfL4$)} is sound.

The other axioms and the rules of $\sfL$ are sound by definition of $\negg$ and $\timp$.
\end{proof}

\begin{corollary}\label{cor:completeness-bpl-bml}
	$\sfH^0\sfL$ is sound and complete for $\calB(\PL)$.
	$\sfH^1\sfL$ is sound and complete for $\calB(\QPL)$.
	$\sfH^\Box\sfL$ is sound and complete for $\calB(\ML)$.
\end{corollary}
\begin{proof}
The soundness follows from the previous lemma and Corollary~\ref{cor:base-completeness-team}.
The completeness follows from Proposition~\ref{prop:counter-models-ml-pl}, Lemma \ref{lem:model-existence1} and Theorem~\ref{thm:completeness-of-L}.
\end{proof}

Next, we show that all Boolean tautologies over $\negg,\timp$ are provable in $\sfL$.
As a consequence, we can derive distributive laws, De Morgan's laws, double negation elimination and so on.

A logic $\calF$ is called \emph{free} if the set $\negg\Phi \cup (\calF\setminus \Phi)$ of $\calB(\calF)$-formulas is satisfiable for all $\Phi \subseteq \calF$.
An example of a free logic is the fragment $\PS$ of $\calB(\PL)$ that contains only propositions and no connectives.

\begin{theorem}\label{thm:completeness-of-l-2}
Let $\calF$ be free.
Then $\sfL$ is complete for $\calB(\calF)$.
\end{theorem}
\begin{proof}
We apply Theorem~\ref{thm:completeness-of-L}, since $\sfL$ is trivially refutation complete for $\calF \cup \negg \calF$: if a set $\Phi \subseteq \calF \cup\negg\calF$ is unsatisfiable, then $\alpha,\negg\alpha \in \Phi$ for some $\alpha \in \calF$, as $\calF$ is free.
\end{proof}

Let $\varphi \in \calB(\PS)$.
A formula $\varphi'$ is a \emph{substitution instance} of $\varphi$ if there are $\psi_1,\ldots,\psi_n$ such that $\varphi' = \varphi[p_1/\psi_1]\cdots[p_n/\psi_n]$ for propositions $p_1,\ldots,p_n$.

\begin{theorem}\label{thm:completeness-bool}
If $\vDash_{\calB(\PS)} \varphi$, then $\vdash_\sfL \varphi'$ for any substitution instance $\varphi'$ of $\varphi$.
\end{theorem}

\begin{example}
The distributive law $a \oland (b \ovee c) \tequiv (a \oland b) \ovee (a \oland c)$ is a tautology in $\calB(\PS)$.
Therefore it gives rise to the instances $\varphi \oland (\psi \ovee \vartheta) \tequiv (\varphi \oland \psi) \ovee (\varphi \oland \psi)$ being provable in $\sfL$ for all $\varphi,\psi,\vartheta$.
\end{example}

\begin{proof}[Proof of Theorem~\ref{thm:completeness-bool}]
Let $\varphi \in \calB(\PS)$ such that $\vDash_{\calB(\PS)} \varphi$.
Suppose that $\varphi'$ is a substitution instance of $\varphi$, \ie, $\varphi' = \varphi[p_1/\psi_1]\cdots[p_n/\psi_n]$.
For arbitrary formulas $\vartheta$, let $\vartheta' \dfn \vartheta[p_1/\psi_1]\cdots[p_n/\psi_n]$ denote the same substitution applied to $\vartheta$.

As $\PS$ is free, $\vdash_\sfL \varphi$ by Theorem~\ref{thm:completeness-of-l-2}.
We proceed with showing $\vdash_\sfL \varphi'$ by induction on the length of a shortest proof of $\varphi$ in $\sfL$.
If $\varphi$ is an instance of \Deriv{($\sfL1$)}, \Deriv{($\sfL2$)}, or \Deriv{($\sfL3$)}, then the same in the case for $\varphi'$.
(Being a $\calB(\PS)$ formula, $\varphi$ cannot be an instance of \Deriv{(L4)}.)

If $\varphi$ was derived from $\psi \timp \varphi$ and $\psi$ via \Deriv{(E$\timp$)}, then $\vdash_\sfL (\psi\timp\varphi)'$ and $\vdash_\sfL \psi'$ by induction hypothesis.
As $(\psi\timp\varphi)' = \psi' \timp\varphi'$, we can apply \Deriv{(E$\timp$)} to obtain $\varphi'$.
\end{proof}

\subsection{A remark on (para-)consistency in team logics}

Due to the two-layered nature of team logics, proof-theoretical subtleties can arise.
We use the term \emph{inconsistent} to describe that a set $\Phi \subseteq \calB(\calF)$ can derive all $\calB(\calF)$ formulas, including $\falsum$.
The ability to derive all formulas in a given system was coined \emph{absolute inconsistency} by Hilbert.

Mossakowski and Schröder \cite{mossakowski2015inconsistency} discussed the difference between absolute inconsistency and so-called \emph{$\bot$-inconsistency}, meaning that $\bot$ can be derived.
Furthermore, they call a set \emph{Aristotle inconsistent} for a given negation symbol $\neg$ if $\alpha$ and $\neg \alpha$ can be derived.
$\neg$ is called \emph{proof-theoretic negation} if $\neg\alpha$ is derivable from $\alpha \vdash \bot$.
Likewise, $\bot$ is called \emph{proof-theoretic falsum} if any formula can be proven from it.
They have to be distinguished from a \emph{semantic} falsum and negation.
Here, $\negg$ and $\falsum$ are both a semantic and proof-theoretic falsum resp.\ negation.

\smallskip

In classical logic using $\bot$ and $\neg$, all above notions of inconsistency coincide with unsatisfiability.
Under team semantics, all above notions of inconsistency still coincide; however, every set of formulas is true in the empty team.
Consequently, classical logics with team semantics have falsum and negation in the proof-theoretic sense, but not in the semantical sense.

A possible workaround is to exclude the empty team from the class of valuations.
If $\calF_+$ is the restriction of the logic $\calF$ (under team semantics) to valuations with non-empty team, then clearly $\Gamma \vDash_\calF \alpha$ implies $\Gamma \vDash_{\calF_+} \alpha$ for all $\Gamma \subseteq \calF$ and $\alpha \in \calF$.
Since valuations with empty team trivially satisfy $\alpha$, the converse is also true.
As a consequence, the consistent sets are then again exactly the satisfiable sets:

\begin{proposition}
Let $\calF \in \{\QPL,\ML,\FO\}$ and $\Gamma \subseteq \calF$.
The following are equivalent:
\begin{itemize}
\item $\Gamma \vdash \calF$
\item $\Gamma \vdash \bot$
\item $\Gamma \vdash \alpha,\neg\alpha$ for some $\alpha \in \calF$
\item $\Gamma$ is unsatisfiable is classical semantics.
\item $\Gamma$ has no team-semantical model with a non-empty team.
\end{itemize}
\end{proposition}

However, while $\bot$ is a semantical falsum when excluding the empty team, clearly $\neg$ is still no semantical negation.
In particular, the law of excluded middle fails, \ie, there are formulas $\alpha$ and valuations satisfying neither $\alpha$ nor $\neg \alpha$ under team semantics.
As an example, consider the propositional team $T = \{ p \mapsto 0, p \mapsto 1 \}$ and the formulas $p$ and $\neg p$.

\begin{corollary}
A proof system is sound and complete for $\calF  \in \{\QPL,\ML,\FO\}$ if and only if it is sound and complete for $\calF_+$.
\end{corollary}

The above corollary is explained by the fact that there simply is no formula of $\QPL, \ML$ or $\FO$ expressing non-emptiness of teams (cf.\ Proposition~\ref{prop:downward-closure}).

With the Boolean closure $\calB(\calF)$ however, the picture changes:
The operators $\negg$ and $\falsum$ assume the role of semantic and proof-theoretic negation and falsum.
\begin{proposition}
Let $\calF \in \{\QPL,\ML \}$ and $\Phi \subseteq \calB(\calF)$.
The following are equivalent:
\begin{itemize}
\item $\Phi \vdash \calB(\calF)$
\item $\Phi \vdash \falsum$
\item $\Phi \vdash \varphi,\negg\varphi$ for some $\varphi \in \calB(\calF)$
\item $\Phi$ is unsatisfiable under team semantics.
\end{itemize}
\end{proposition}

Here, the empty team is again permitted as a valuation.
The connectives $\neg$ and $\bot$ behave interestingly:
Despite clearly being Aristotle inconsistent, $\{\alpha,\neg\alpha\}$ is not 
absolutely inconsistent anymore, \ie, $\{\alpha,\neg\alpha\} \vdash \calF$, but $\{\alpha,\neg\alpha\} \nvdash \calB(\calF)$.
Mossakowski and Schröder call such an operator $\neg$  \emph{paraconsistent negation}.
Similarly, $\bot \nvdash \calB(\calF)$.

The term "paraconsistent" is a little inappropriate for team semantics, as any model of $\{\alpha,\neg\alpha\}$ or $\bot$ can only have an empty team and thus is not very meaningful.
In fact, removing the empty team as a valuation establishes $\bot \vDash \calB(\calF)$ and avoids paraconsistency,
with the formula $\NE \dfn \negg \bot$ (which expresses non-emptiness of teams) added as an axiom.
On the other hand, the empty team cannot be easily excluded from, say, $\MTL$, unless the successor relation is total and provides all teams in all Kripke structures with a non-empty image.
Likewise, the semantics of the splitting operator $\limp$ would have to be changed in order to avoid empty teams.
This also implies unwanted side-effects such as $\bot \tens \top$ being equivalent to $\falsum \tens \top$, and hence being contradictory instead of valid, thereby violating Proposition~\ref{prop:flatness-tens}.

As a consequence, for the rest of this paper, we permit the empty team and tolerate proof systems that are paraconsistent in the above sense.

\section{Axioms of splitting}

\allowdisplaybreaks

\label{sec:splitting}

\begin{figure}[b]
	\centering
\scalebox{.95}{
	\begin{tabular}{lll}
		\toprule
		\Deriv{(F$\tens$)}&$(\alpha\tens\beta) \tequiv (\alpha\lor\beta)$&Flatness of $\tens$\\
		\Deriv{(F$\limp$)}&$\alpha \timp (\varphi \limp \alpha)$&Downwards closure\\
		\Deriv{(Lax)}&$\varphi \timp (\varphi \limp \psi) \timp (\vartheta \limp \psi)$&Lax semantics\\
		\Deriv{(Ex$\limp$)}&$(\varphi \limp \psi \limp \vartheta) \timp (\psi \limp \varphi \limp \vartheta)$&Exchange of hypotheses\\
		\Deriv{(C$\limp$)}&$(\varphi \limp \negg\psi) \timp (\psi \limp\negg \varphi)$&Contraposition\\
		\Deriv{(Dis$\limp$)}&$(\varphi \limp (\psi \timp \vartheta)) \timp (\varphi \limp \psi) \timp (\varphi \limp \vartheta)$&Distribution axiom\\
		\midrule
		\Deriv{(Nec$\limp$)}&\begin{bprooftree}
\AxiomC{$\varphi$}
\RightLabel{\small{}($\varphi$ theorem)}
\UnaryInfC{$\psi \limp \varphi$}
\end{bprooftree}
&Necessitation\\
		\bottomrule
	\end{tabular}}
	\caption{The system $\sfS$}\label{fig:splitting}
\end{figure}

In the previous section, we added team-semantical Boolean connectives to classical logics.
Hodges \cite{Hodges1997} and Väänänen \cite{vaananen_dependence_2007} introduced the \emph{splitting disjunction} $\tens$, also called \emph{splitjunction} or \emph{tensor}.
Formally, $T \vDash \varphi \tens \psi$ if $T$ can be divided into (possibly overlapping) subteams $S, U$ such that $S \vDash \varphi$ and $U \vDash \psi$.
Intuitively, $\tens$ is a "member-wise disjunction": each element in $T$ chooses $\varphi$ or $\psi$ or both (cf.\ Proposition~\ref{prop:flatness-tens}).
Galliani~\cite{galliani_inclusion_2012} referred to this semantics of $\tens$ as \emph{lax semantics}.
By contrast, in the so-called \emph{strict} semantics the division must form a partition; hence the strict $\tens$ rather is a member-wise "exclusive or".

\smallskip

This section is devoted to axiomatizing $\limp$ and hence $\tens$, as $\varphi \tens \psi \dfn \negg(\varphi \limp \negg \psi)$.
\label{p:yang}
In our approach, we interpret $\limp$ as countably many unary modalities of the form "$\psi \limp$" instead of a disjunction-like operator.
This permits a natural axiomatization by the system $\sfS$ (see Figure~\ref{fig:splitting}).

\medskip

With a model-theoretic argument, Yang \cite[Theorem 4.6.4.]{yang_extensions_2014} showed that every $\PTL$ formula can be brought into a normal form over Boolean conjunctions ($\oland$), disjunctions ($\ovee$), splitting ($\tens$), and non-emptiness atoms ($\NE \dfn \negg \bot$).
She argued that the axiomatization of this fragment is easier than for full $\PTL$, as it avoids arbitrary negation.
On the other hand, this fragment demands a rather complicated set of rules for many special cases, in particular to handle $\NE$.

\bigskip

\begin{theorem}\label{thm:ptl-soundness}
The proof system $\sfH^0\sfL\sfS$ is sound for $\PTL$.
\end{theorem}
\begin{proof}
The proof is straightforward and can be found in the appendix.
\end{proof}

The idea for proving completeness is to reduce the problem to the completeness for better-behaved fragment.
More precisely, every $\PTL$-formula will be broken down into a $\calB(\PL)$-formula (cf.\ Figure~\ref{fig:overview}).
This is formally stated in the next theorem, and the remaining parts of this section will culminate in a proof.

\begin{theorem}\label{thm:ptl-to-bpl-equiv}
Let $\varphi \in \PTL$.
Then there is $\psi \in \calB(\PL)$ such that $\varphi \eqpr_{\sfH^0\sfL\sfS} \psi$.
\end{theorem}

The following lemma shows that such a translation in principle is sufficient for showing completeness, provided the system is also sound.

\begin{lemma}\label{lem:translate-completeness}
Let $\calL,\calL'$ be logics such that $\calL' \subseteq \calL$.
Let $\Omega$ be a proof system that is sound for $\calL$ and complete for $\calL'$, and such that every $\calL$-formula is provably equivalent to an $\calL'$-formula in $\Omega$.
Then $\Omega$ is also complete for $\calL$.
\end{lemma}
\begin{proof}
Assume $\Phi \subseteq \calL$ and $\varphi \in \calL$.
For completeness we have to show that $\Phi \vDash \varphi$ implies $\Phi \vdash \varphi$.
By assumption, every $\calL$-formula is provably equivalent to an $\calL'$-formula, hence $\Phi \eqpr \Phi'$ for some set $\Phi' \subseteq \calL'$.
Likewise, $\varphi \eqpr \varphi'$ for some $\varphi' \in \calL'$.
Since these equivalences are proven between (sets of) $\calL$-formulas, soundness for $\calL$ implies $\Phi \equiv \Phi'$ and $\varphi \equiv \varphi'$.
Consequently, $\Phi' \vDash \varphi'$.
By completeness of $\Omega$ for $\calL'$, we obtain $\Phi' \vdash \varphi'$.
Altogether, then $\Phi \vdash \Phi' \vdash \varphi' \vdash \varphi$.
As $\vdash$ is transitive, the lemma follows.
\end{proof}

Due to the above lemma, Theorem~\ref{thm:ptl-soundness} and \ref{thm:ptl-to-bpl-equiv}, and Corollary~\ref{cor:completeness-bpl-bml}, we obtain an axiomatization of $\PTL$:

\begin{corollary}\label{cor:ptl-completeness}
The proof system $\sfH^0\sfL\sfS$ is sound and complete for $\PTL$.
As a consequence, $\PTL$ is axiomatizable and compact.
\end{corollary}

The remainder of this section is devoted to proving the required Theorem~\ref{thm:ptl-to-bpl-equiv}.

However, we will restrict ourselves to lax semantics instead of strict semantics.
One reason is that the former enjoys several natural properties such as the \emph{locality property}: If two teams agree on their assignments \wrt some variables $p_1, \ldots, p_n$, then they satisfy the same formulas over these variables (see also Yang and Väänänen~\cite{yang_propositional_2017}).

For any propositional team $T$ and propositions $p_1,\ldots,p_n \in \PS$, we define
\[
\rel(T, (p_1,\ldots,p_n)) \dfn \{ (s(p_1),\ldots,s(p_n)) \mid s \in T \}\text{.}
\]
Then we can state the locality property as follows:
\begin{proposition}[Locality]\label{prop:locality}
Let $T,T'$ be propositional teams and $\varphi \in \PTL$ such that $\varphi$ contains the propositions $p_1, \ldots, p_n$.
Then $\rel(T,(p_1,\ldots,p_n)) = \rel(T',(p_1,\ldots,p_n))$ implies $T \vDash \varphi \Leftrightarrow T' \vDash \varphi$ in lax semantics.
\end{proposition}

A proof is found in the appendix.
Under strict semantics, locality is not true in general:

\begin{example}
Under strict semantics, $\psi\dfn \negg p \tens \negg p$ states that the team contains at least two assignments $s,s'$ with $s(p) = s'(p) = 0$.

Now, for an assignment $s$ with $s(p) = 0$, consider the teams $\{s\}$ and $\{ s^{q}_0, s^{q}_1 \}$, where $q \neq p$.
Clearly, $\rel(\{ s \},(p)) = \{ (0) \} = \rel(\{ s^{q}_0, s^{q}_1 \}, (p))$.
However, $\{ s \} \nvDash \psi$ and $\{ s^{q}_0, s^{q}_1 \} \vDash \psi$, violating locality.
\end{example}

Note that lax and strict semantics coincide for $\calB(\PL)$.

\label{p:count}

\begin{corollary}\label{cor:no-counting}
In strict semantics, $\negg p \tens \negg p$ is not equivalent to any $\calB(\PL)$-formula.
\end{corollary}

Observe that the axiom \Deriv{(Lax)} is not provable from the remaining axioms:
The system $\sfH^0\sfL\sfS$ except \Deriv{(Lax)} is easily proven sound for strict semantics, and consequently cannot prove \Deriv{(Lax)}, as the latter is not a theorem in strict semantics.
For this reason, an explicit axiom for lax semantics must necessarily be added.

\subsection{Splitting elimination}

As a specific instance of Lemma~\ref{lem:translate-completeness}, we introduce \emph{$\mathfrak{f}$-elimination}:

\begin{definition}\label{def:f-elim}
	Let $\calL$ be a logic and $\Omega$ a proof system. Let $\mathfrak{f}$ be an $n$-ary connective.
	We say that $\calL$ has \emph{$\mathfrak{f}$-elimination in $\Omega$} if for all formulas $\xi_1,\ldots,\xi_n \in \calL$ there exists some $\varphi \in \calL$ such that $\mathfrak{f}(\xi_1, \ldots, \xi_n) \eqpr_\Omega \varphi$.
\end{definition}
In other words, if $\xi_1, \ldots, \xi_n$ are $\calL$-formulas, then $\mathfrak{f}(\xi_1, \ldots, \xi_n)$ is as well equivalent to an $\calL$-formula.

In this subsection, we aim at proving that $\calB(\PL)$ has $\limp$-elimination in order to prove Theorem~\ref{thm:ptl-to-bpl-equiv}.

As we let the elimination start at the innermost subformulas, we additionally require the next definition.

\begin{definition}
Let $\mathfrak{g}$ be an $n$-ary connective.
Say that a proof system $\Omega$ has \emph{substitution in }$\mathfrak{g}$ if for all $\varphi_1,\psi_1,\ldots,\varphi_n,\psi_n$ it holds that $\varphi_1 \eqpr \psi_1, \ldots, \varphi_n \eqpr \psi_n$ implies $\mathfrak{g}(\varphi_1, \ldots, \varphi_n) \eqpr \mathfrak{g}(\psi_1, \ldots, \psi_n)$.
\end{definition}

\label{pg:metarules}

In order to prove $\limp$-elimination and substitution, we require several auxiliary results, such as in the following lemma.
Note that the deduction theorem is available for any system $\Omega \succeq \sfL\sfS$.
By means of the latter and the system $\sfS$, the proof of the following meta-rules is straightforward and can be found in the appendix.

\begin{figure}[b]\centering
\begin{tabular}{lll}
		\toprule
		\Deriv{(Com$\tens$)}&$(\varphi \tens \psi) \tequiv (\psi \tens \varphi)$&Commutative law for $\tens$\\
		\Deriv{(Ass$\tens$)}&$((\varphi \tens \psi)\tens\vartheta) \tequiv (\varphi \tens(\psi\tens \vartheta))$&Associative law for $\tens$\\
		\Deriv{(D$\oland\tens$)}&$\alpha \oland (\varphi \tens \psi) \tequiv (\alpha \oland \varphi) \tens (\alpha \oland \psi)$&Distr. law for $\tens$ and $\oland$\\
		\Deriv{(D$\ovee\tens$)}&$\varphi \tens (\psi \ovee \vartheta) \tequiv (\varphi \tens \psi) \ovee (\varphi \tens \vartheta)$&Distr. law for $\tens$ and $\ovee$\\
		\Deriv{(Aug$\tens$)}&$(\varphi \tens \psi) \oland (\varphi \limp \vartheta) \timp (\varphi \tens (\psi \oland \vartheta))$&Augment splitting\\
		\Deriv{(Abs$\tens$)}&$(\E \alpha \tens \varphi) \timp \E\alpha$&Absorption law of $\tens$\\
		\Deriv{(Join$\E$)}&$(\alpha \oland \E\beta) \timp \E(\alpha \land \beta)$&\\
		\Deriv{(Isolate$\E$)}&$(\varphi \tens (\alpha \oland \E\beta)) \tequiv (\varphi \tens \alpha) \oland \E(\alpha\land\beta)$&\\
		\bottomrule
	\end{tabular}
	\caption{Provable laws $\sfS'$ of $\limp$ and $\tens$}\label{fig:splitting2}
\end{figure}

\begin{lemma}\label{lem:meta-ptl}
Let $\Omega \succeq \sfL\sfS$ be a proof system.
Them $\Omega$ has substitution in $\negg$, $\timp$ and $\limp$.
Furthermore, $\Omega$ admits the following meta-rules:
\begin{itemize}
	\item Reductio ad absurdum \Deriv{(RAA)}:
	If $\Phi \cup \{\varphi\} \vdash \{\psi, \negg \psi\}$, then $\Phi \vdash\negg\varphi$.
	If $\Phi \cup \{\negg\varphi\} \vdash \{\psi, \negg \psi\}$, then $\Phi \vdash \varphi$.
	\item Modus ponens in $\limp$ \Deriv{(MP$\limp$)}:
	If $\;\vdash \varphi \timp \psi$ and $\Phi \vdash \vartheta \limp \varphi$, then $\Phi \vdash \vartheta \limp \psi$.
	\item Modus ponens in $\tens$ \Deriv{(MP$\tens$)}:
	If $\;\vdash \varphi \timp \psi$ and $\Phi \vdash \vartheta \tens \varphi$, then $\Phi \vdash \vartheta \tens \psi$.
\end{itemize}
\end{lemma}

Moreover, the axioms $\sfS$ allow to derive basic laws regarding $\limp$, its dual $\tens$, and the remaining connectives, with the derivations again found in the appendix:

\begin{lemma}\label{lem:ptl-laws}
Let $\Omega \succeq \sfH^0\sfL\sfS$. Then all instances of the axioms in $\sfS'$ are theorems of $\Omega$.
\end{lemma}

In the rest of the section, we prove that the system admits $\limp$-elimination.
The proof spans over several lemmas.
We implicitly apply Lemma~\ref{lem:meta-ptl} when using substitution in $\timp, \negg$ and $\limp$ and make use of the laws in Lemma~\ref{lem:meta-ptl} and the system $\sfS'$.
Roughly speaking, we pull $\tens$ inside any Boolean connectives.
The first step is the \emph{and/or lemma}.

\begin{lemma}[And/Or lemma]\label{lem:vee-swap}
If $\Omega \succeq \sfH^0\sfL\sfS$, then
\[
\bigowedge_{i=1}^n \E\beta_i \tequiv \bigtens_{i=1}^n \E\beta_i
\]
is a theorem of $\Omega$ for all $\beta_1,\ldots,\beta_n \in \calF$.
\end{lemma}
\begin{proof}
Using the deduction theorem, we show $\bigowedge_{i=1}^n \E\beta_i \eqpr \bigtens_{i=1}^n \E\beta_i$.
We begin with the direction "$\vdash$", and proceed by induction on $n$, where $n = 1$ is trivial.
For $n > 1$, by induction hypothesis and substitution in $\owedge$, it suffices to prove $(\bigtens_{i=1}^{n-1} \E\beta_i) \oland \E\beta_{n} \vdash \bigtens_{i=1}^n \E\beta_i$.

In $\sfL$, we can decompose the conjunction.
Then, assuming $\bigtens_{i=1}^{n-1}\E\beta_i$ and $\E\beta_n$ as premises, we prove $\bigtens_{i=1}^n \E\beta_i$ by \Deriv{(RAA)}.
From its negation, viz.\ $\bigotimes_{i=1}^{n-1}\E\beta_i \limp \negg\E\beta_n$,
we derive $\top \limp \negg\E\beta_n$ with \Deriv{(Lax)}.
By \Deriv{(C$\limp$)}, then $\E\beta_n \limp \negg \top$ follows.
Finally, again by \Deriv{(Lax)}, we obtain $\top \limp \negg \top$.
However, $\top \lor \top$, and hence $\top \tens \top = \negg(\top \limp \negg \top)$, is a theorem of $\sfH^0\sfS$ as well.
By \Deriv{(RAA)}, we conclude $\bigtens_{i=1}^n \E\beta_i$.

The other direction "$\ltimp$" is shown by a separate derivation of each conjunct with \Deriv{(Abs$\tens$)}, \Deriv{(Ass$\tens$)} and \Deriv{(Com$\tens$)}, which in $\sfL$ then yields the conjunction.
\end{proof}

\begin{lemma}[Generalized distributive law]\label{lem:distribute-alpha}
If $\Omega\succeq \sfH^0\sfL\sfS$, then
\[
\alpha \oland \left( \bigoland_{i=1}^n \E \beta_i\right) \tequiv \bigotimes_{i=1}^n (\alpha \oland \E\beta_i)
\]
is a theorem of $\Omega$ for all $\alpha,\beta_1,\ldots,\beta_n \in \calF$.
\end{lemma}
\begin{proof}
First we apply the previous lemma to replace the large conjunction by a large splitting disjunction.
Then we distribute $\alpha$ with repeated application of \Deriv{(D$\oland\tens$)}, \Deriv{(Ass$\tens$)} and \Deriv{(Com$\tens$)}.
\end{proof}

\begin{lemma}[$\E$ isolation]\label{lem:isolate-alpha}
If $\Omega\succeq \sfH^0\sfL\sfS$, then
\[
\bigtens_{i=1}^n \left( \alpha_i \oland \E \beta_i\right) \tequiv \left(\bigtens_{i=1}^n \alpha_i\right) \oland \bigoland_{i=1}^n \E (\alpha_i \land \beta_i)
\] is a theorem of $\Omega$ for all $\alpha_1,\ldots,\alpha_n,\beta_1,\ldots,\beta_n \in \calF$.
\end{lemma}
\begin{proof}
For "$\vdash$", we obtain $\bigtens_{i=1}^n \alpha_i$ from $\bigtens_{i=1}^n \left( \alpha_i \oland \E \beta_i\right)$ by the application of \Deriv{(Ass$\tens$)}, \Deriv{(Com$\tens$)} and \Deriv{(MP$\tens$)}, as $(\alpha_i\oland\E\beta_i) \vdash\alpha_i$ for all $i \in \{1, \ldots, n\}$.

Next, we apply \Deriv{(Join$\E$)} to similarly derive $\bigtens_{i=1}^n \E(\alpha_i \land \beta_i)$, which by Lemma~\ref{lem:vee-swap} yields $\bigoland_{i=1}^n \E(\alpha_i \land \beta_i)$.
For "$\dashv$", we repeatedly apply the theorem \Deriv{(Isolate$\E$)} of $\sfS'$, $(\varphi \tens \alpha) \oland \E(\alpha \land \beta) \timp \varphi \tens (\alpha \oland \E \beta)$, as follows:
Assume that the formula has the following form after $k$ applications.
\[
\left(\bigtens_{i=1}^k(\alpha_i \oland \E\beta_i) \tens \bigtens_{i=k+1}^{n} \alpha_i\right)\oland \bigoland_{i=k+1}^{n}\E(\alpha_i \land \beta_i)\text{.}
\]
For $k = 0$, this is obvious. With commutative and associative laws we isolate a single subformula on each side:
\[
\left[\left(\bigtens_{i=1}^k(\alpha_i \oland \E\beta_i)  \tens \bigtens_{i=k+2}^{n} \alpha_i\right) \tens \alpha_{k+1}\right]\oland  \E(\alpha_{k+1}\land \beta_{k+1})\oland\bigoland_{i=k+2}^{n}\E(\alpha_i \land \beta_i)
\]
Then we apply \Deriv{(Isolate$\E$)}, resulting in
\[
\left[\left(\bigtens_{i=1}^k(\alpha_i \oland \E\beta_i)  \tens \bigtens_{i=k+2}^{n} \alpha_i\right)\tens  (\alpha_{k+1}\oland\E \beta_{k+1})\right]\oland\bigoland_{i=k+2}^{n}\E(\alpha_i \land \beta_i)\text{,}
\]
and again with commutative and associative laws in
\[
\left(\bigtens_{i=1}^{k+1}(\alpha_i \oland \E\beta_i)  \tens \bigtens_{i=k+2}^{n} \alpha_i\right)\oland\bigoland_{i=k+2}^{n}\E(\alpha_i \land \beta_i)\text{,}
\]
where we can repeat the above steps until $k = n$.
\end{proof}

\begin{lemma}[Flatness of $\tens$]\label{lem:flatness-transform}
If $\Omega \succeq \sfH^0\sfL\sfS$, then
\[
\bigtens_{i=1}^n\alpha_i \tequiv \bigvee_{i=1}^n \alpha_i
\]
is a theorem of $\Omega$ for all $\alpha_1, \ldots,\alpha_n \in \calF$.
\end{lemma}
\begin{proof}
By induction on $n$, where $n = 1$ is trivial.
For $n > 1$, first let $\varphi \dfn \bigtens_{i=1}^{n-1}\alpha_i$ and $\gamma\dfn \bigvee_{i=1}^{n-1} \alpha_i$.
Then, by induction hypothesis, $\varphi \eqpr \gamma$.
By \Deriv{(Com$\tens$)} and \Deriv{(MP$\tens$)}, $\varphi\tens\alpha_n \eqpr \gamma\tens \alpha_n$ follows.
Finally, by \Deriv{(F$\tens$)} we obtain $\varphi \tens \alpha_n \eqpr \gamma\tens \alpha_n \eqpr \gamma\lor \alpha_n$.
\end{proof}

\begin{lemma}[Flatness of $\owedge$]\label{lem:flatness-transform2}
If $\Omega \succeq \sfH^0\sfL\sfS$, then
\[
\bigoland_{i=1}^n\alpha_i \tequiv \bigwedge_{i=1}^n\alpha_i
\]
is a theorem of $\Omega$ for all $\alpha_1, \ldots,\alpha_n \in \calF$.
\end{lemma}
\begin{proof}
The proof is again by induction on $n$.
Analogously as before, let $\varphi \dfn \bigoland_{i=1}^{n-1}\alpha_i$ and $\gamma\dfn \bigwedge_{i=1}^{n-1} \alpha_i$, where $\varphi \eqpr \gamma$.
Then $\varphi\oland\alpha_n \eqpr \gamma\oland \alpha_n$ in $\sfL$ by substitution.
Next, to prove the lemma, we show $\gamma\land\alpha_n \eqpr \gamma\oland\alpha_n$.

Clearly from $\gamma\land \alpha_n$ we can derive $\gamma$ and $\alpha_n$ in $\sfH^0$, and then $\gamma \oland\alpha_n$ in $\sfL$.
For the other direction, \ie, to prove $\gamma\land\alpha_n$ from $\gamma\oland\alpha_n$, we use \Deriv{(RAA)} and assume the premises $\gamma\oland\alpha_n$ and $\negg(\gamma\land\alpha_n) = \negg \neg(\gamma\imp \neg\alpha_n) = \E(\gamma\imp\neg\alpha_n)$.

By $\sfL$, we have $\gamma \oland \alpha_n \vdash \gamma$ and $\gamma \oland \alpha_n \vdash \alpha_n$.
Two applications of \Deriv{(Join$\E$)} then produce $\E(\gamma \land \alpha_n \land (\gamma\imp\neg\alpha_n))$.
Clearly, this yields $\E\bot = \negg \neg \bot$ in $\sfH^0\sfL$.
At the same time, $\top$ and consequently $\neg \bot$ is derivable in $\sfH^0$.
By \Deriv{(RAA)}, we conclude $\gamma\land\alpha_n$ from $\neg \bot$ and $\negg \neg \bot$.
\end{proof}

With the above lemmas, we are finally ready to prove the $\limp$-elimination.

\begin{lemma}[$\limp$-elimination]\label{lem:tensor-elim}
Let $\calF$ be a logic closed under $\neg, \lor, \land$.
Let $\Omega\succeq \sfH^0\sfL\sfS$.
Then $\calB(\calF)$ has $\limp$-elimination in $\Omega$.
\end{lemma}
\begin{proof}
To prove $\limp$-elimination, suppose that $\varphi = \psi \limp \vartheta$ is a formula where $\psi,\vartheta \in \calB(\calF)$.
By substitution, $\varphi \eqpr \psi \limp \negg \negg \vartheta$, and by Theorem~\ref{thm:completeness-bool}, we can apply De Morgan's laws and distributive laws on both $\psi$ and $\negg\vartheta$.
This allows to replace $\psi$ and $\negg\vartheta$ by formulas $\psi'$, $\vartheta'$ in disjunctive normal form (DNF) over $\oland,\ovee$.

We arrive at the following provably equivalent form of $\varphi$,
\[
\negg\left[\bigovee_{i=1}^n\; \left(\bigoland_{j=1}^{o_i}\alpha_{i,j} \oland \bigoland_{j=1}^{m_i} \E\beta_{i,j} \right) \;\tens \;\bigovee_{i=1}^{n'}\; \left( \bigoland_{j=1}^{o'_i}\alpha'_{i,j} \oland \bigoland_{j=1}^{m'_i} \E\beta'_{i,j} \right)\right]\text{,}
\]
with the negative literals represented with $\E$, since $\sfH^0$ can introduce $\neg\neg$ if necessary.
For suitable $\alpha_i$ and $\alpha'_i$ in $\calF$, we derive in $\sfH^0\sfL\sfS$:
	\begin{align*}
	\Apply{(Lemma~\ref{lem:flatness-transform2})}\eqpr \quad &\negg\left[\bigovee_{i=1}^n \left(\alpha_i \oland \bigoland_{j=1}^{m_i} \E\beta_{i,j} \right) \;\tens \;\bigovee_{i=1}^{n'}\; \left(\alpha'_i \oland \bigoland_{j=1}^{m'_i} \E\beta'_{i,j} \right)\right]\\
	\Apply{(D$\ovee\tens$)}\eqpr \quad &\negg\!\!\!\bigovee_{\substack{1 \leq i \leq n\\1\leq i' \leq n'}}\!\! \left[ \left( \alpha_{i} \oland \bigoland_{j=1}^{m_i} \E \beta_{i,j} \right) \tens \left( \alpha'_{i'} \oland \bigoland_{j=1}^{m'_{i'}} \E \beta'_{i',j} \right) \right]\\
	\Apply{(Lemma~\ref{lem:distribute-alpha})}\eqpr \quad &\negg\!\!\!\bigovee_{\substack{1 \leq i \leq n\\1\leq i' \leq n'}} \!\! \left( \bigtens_{j=1}^{m_i} \bigg( \alpha_{i} \oland \E \beta_{i,j} \bigg) \tens \bigtens_{j=1}^{m'_{i'}} \bigg( \alpha'_{i'} \oland \E \beta'_{i',j} \bigg)\right)\\
	\text{\scriptsize(Renaming)}\; = \quad &\negg \bigovee_{i = 1}^{\ell}\;  \bigtens_{j=1}^{k_i} \left(\gamma_{i,j} \oland \E \delta_{i,j} \right)\\
		\Apply{(Lemma~\ref{lem:isolate-alpha})}\eqpr \quad &\negg \bigovee_{i = 1}^{\ell}  \left(\bigtens_{j=1}^{k_i} \gamma_{i,j} \oland \bigoland_{j=1}^{k_i} \E \left(\gamma_{i,j} \land \delta_{i,j} \right)\right)\\
		\Apply{(Lemma~\ref{lem:flatness-transform})}\eqpr \quad &\negg \bigovee_{i = 1}^{\ell} \left(\bigvee_{j=1}^{k_i} \gamma_{i,j} \oland \bigoland_{j=1}^{k_i} \E \left(\gamma_{i,j} \land \delta_{i,j} \right)\right) \quad =: \varphi' \in \calB(\calF)\text{.}
	\end{align*}
\end{proof}

We are now ready to prove the main theorem of this section:

\begin{reptheorem}{thm:ptl-to-bpl-equiv}
Let $\varphi \in \PTL$.
Then there is $\psi \in \calB(\PL)$ such that $\varphi \eqpr_{\sfH^0\sfL\sfS} \psi$.
\end{reptheorem}
\begin{proof}
Given $\varphi \in \PTL$, we construct $\psi \in \calB(\PL)$ by induction on $\varphi$.
If $\varphi = \varphi_1 \timp \varphi_2$ or $\varphi = \negg\varphi_1$, then by induction hypothesis, $\varphi_1$ and/or $\varphi_2$ are provably equivalent to $\calB(\PL)$ formulas $\psi_1$ and $\psi_2$.
By Lemma~\ref{lem:meta-ptl}, we obtain a provably equivalent formula $\psi \in \calB(\PL)$ by substitution in $\timp$ and $\negg$.

The remaining case is $\varphi = \varphi_1 \limp \varphi_2$.
By induction hypothesis, $\varphi_1 \eqpr \psi_1$ and $\varphi_2\eqpr \psi_2$ for some $\psi_1,\psi_2 \in\calB(\calF)$.
Here, the theorem follows by $\limp$-elimination (Lemma~\ref{lem:tensor-elim}).
\end{proof}

\subsection{Examples in propositional team logic}

Constraints such as dependence, independence or inclusion on teams are definable in $\PTL$.
As a consequence, laws such as Armstrong's axioms for functional dependence can be proved in our system.

\begin{example}
The dependency atom $\dep{x,y}$ ("$y$ is a function of $x$") can be written as $\top \limp (\dep{x} \timp \dep{y})$, where $\dep{\alpha} \dfn \alpha \ovee \neg \alpha$.
Figure~\ref{fig:example-dep} depicts a proof of one of Armstrong's axioms of dependence \cite{armstrong} in the system, namely the axiom of transitivity.
It states that from $\dep{x,y}$ and $\dep{y,z}$ we can infer $\dep{x,z}$.
\end{example}

\begin{figure}[b!]
\centering
{
\setlength{\fitchprfwidth}{3.63in}

\fitchprf{
\pline[A ]{\dep{x,y}} \\
\pline[B ]{\dep{y,z}}
}
{
	\pline[1 ]{\top \limp (\dep{x} \timp \dep{y})}[Def., A]\\
	\pline[2 ]{\top \limp (\dep{y} \timp \dep{z})}[Def., B]\\
	\subproof{}
	{
		\pline[3 ]{
			\brokenform{
				(\dep{x} \timp \dep{y})
			}{
				\formula{ \timp ((\dep{y}\timp\dep{z})\timp (\dep{x} \timp \dep{z}))}
			}
		}[$\sfL$]\\
		\pline[4 ]{
			\brokenform{
				\top \limp \big(((\dep{x} \timp \dep{y})
			}{
			\formula{ \timp ((\dep{y}\timp\dep{z}) \timp (\dep{x} \timp \dep{z}))\big)}
			}
		}[\Deriv{(Nec$\limp$)}]\\
		\pline[5 ]{
			\brokenform{
				\big(\top \limp (\dep{x} \timp \dep{y}) \big)
			}{
				\formula{\timp \big(\top \limp ((\dep{y}\timp\dep{z})\timp (\dep{x} \timp \dep{z})) \big)}
			}
		}[\Deriv{(Dis$\limp$)}]
	}
	\\
\pline[6 ]{\top \limp \big( (\dep{y}\timp\dep{z}) \timp (\dep{x} \timp \dep{z})\big)}[\Deriv{(E$\timp$)}, 1, 5]\\
\pline[7 ]{\big(\top \limp (\dep{y}\timp\dep{z}) \big) \timp\big( \top \limp (\dep{x} \timp \dep{z})\big)}[\Deriv{(Dis$\limp$)}]\\
\pline[8 ]{\top \limp (\dep{x} \timp \dep{z}))}[\Deriv{(E$\timp$)}, 2, 7]\\
\pline[\slider]{\dep{x,z}}[Def.]
}
}\caption{Example derivation of the transitivity of dependence\label{fig:example-dep}}
\end{figure}

\begin{example}
For $\alpha,\beta \in \PL$, the formula $(\alpha \limp \beta) \timp \beta$ is a theorem of $\PTL$.
It is easy to see that it is valid:
$\alpha$ is satisfied by the empty team, and as every team $T$ has the trivial division into $\emptyset$ and $T$, having $T \vDash \alpha \limp \beta$ implies $T \vDash \beta$.

We sketch a proof in the system $\sfH^0\sfL\sfS$.
First, clearly $\vdash_{\sfH^0} \bot \imp \alpha$.
This implies $\vdash_{\sfH^0\sfL} \negg\alpha \timp \negg\bot$ by contraposition.
As this formula is a theorem, \Deriv{(MP$\limp$)} is applicable.
Moreover, $\alpha \limp \negg\negg\beta$ follows form $\alpha \limp\beta$ by substitution in $\limp$, which by \Deriv{(C$\limp$)} yields $\negg\beta \limp \negg\alpha$.
By \Deriv{(MP$\limp$)}, we obtain $\negg \beta \limp \negg \bot$.

Finally, $\beta$ is proved with \Deriv{(RAA)} by assuming $\negg \beta$.
From $\negg \beta$ and $\negg \beta \limp \negg \bot$ we obtain $\top \limp \negg \bot$ by \Deriv{(Lax)}.
However, this contradicts $\top \tens \bot \dfn \negg(\top \limp \negg \bot)$, which itself follows from $\vdash_{\sfH^0} \top \lor \bot$ and \Deriv{(F$\tens$)}.
\end{example}

\section{Modal team logic}

\label{sec:modal}

Modal team logic generalizes the modal operators $\Diamond$ (here defined via $\triangle$) and $\Box$ to act on teams.
Analogously to $\limp$ in Theorem~\ref{thm:ptl-to-bpl-equiv}, we axiomatize the modalities $\triangle$ and $\Box$ in order to eliminate them from formulas.

By a model-theoretic argument, Kontinen, Müller, Schnoor and Vollmer~\cite{kontinen_van_2014} showed $\MTL \equiv \calB(\ML)$, \ie, that every $\MTL$-formula is equivalent to a $\calB(\ML)$-formula.
Their idea is that every $\MTL$-formula $\varphi$ can be written as a Boolean combination (over $\negg,\timp$) of finitely many so-called \emph{Hintikka formulas} of the bisimulation types of the models of $\varphi$.
These formulas essentially characterize a Kripke structure up to bounded bisimulation (see also Goranko and Otto~\cite{Goranko2007249}).
As Hintikka formulas are $\ML$-formulas, Kontinen et al.\ conclude that every $\MTL$-formula has an equivalent $\calB(\ML)$-formula.

Analogously as for $\PTL$, we give a purely syntactical proof of $\MTL \equiv \calB(\ML)$.
This translation utilizes the system $\sfM$, depicted in Figure~\ref{fig:modal}.

\begin{figure}
	\centering
	\begin{tabular}{lll}
		\toprule
		\Deriv{(Lin$\Box$)}&$\Box \negg \varphi \tequiv \negg\Box\varphi $&The image team is unique.\\
		\Deriv{(F$\Diamond$)}&$\Diamond\alpha \tequiv \neg \Box\neg\alpha$&Flatness of $\Diamond$\\
		\Deriv{(D$\Diamond\tens$)}&$\Diamond (\varphi \tens \psi) \tequiv \Diamond\varphi \tens \Diamond \psi$&$\Diamond$ distributes over splitting.\\
		\Deriv{(E$\Box$)}&$\Box\alpha \timp \triangle \alpha$&Successor teams are subteams \\
        & & of the image.\\
		\Deriv{(I$\Box$)}&$\Diamond \varphi \timp (\triangle\psi \timp \Box \psi)$&If there is some successor team,\\
        & & then the image is a successor team.\\
		\Deriv{(Dis$\Box$)}&$\Box (\varphi \timp \psi) \timp (\Box \varphi \timp \Box\psi)$&Distribution axiom\\
		\Deriv{(Dis$\triangle$)}&$\triangle (\varphi \timp \psi) \timp (\triangle \varphi \timp \triangle\psi)$&Distribution axiom\\
		\midrule
		\Deriv{(Nec$\Box$)}
		&\begin{bprooftree}
\AxiomC{$\varphi$}
\RightLabel{\small{}($\varphi$ theorem)}
\UnaryInfC{$\Box \varphi$}
\end{bprooftree}\vspace{15pt}
&Necessitation\\
\Deriv{(Nec$\triangle$)}
&\begin{bprooftree}
\AxiomC{$\varphi$}
\RightLabel{\small{}($\varphi$ theorem)}
\UnaryInfC{$\triangle \varphi$}
\end{bprooftree}
&Necessitation\\
		\bottomrule
	\end{tabular}
	\caption{The system $\sfM$}\label{fig:modal}
\end{figure}

\begin{theorem}\label{thm:mtl-completeness}
	The proof system $\sfH^\Box\sfL\sfS\sfM$ is sound for $\MTL$.
\end{theorem}
\begin{proof}
As for $\PTL$, a soundness proof is not difficult and can be found in the appendix.
\end{proof}

Let us state the main theorem of this section.
As before, its proof then extends over several lemmas.

\begin{theorem}\label{thm:mtl-is-equiv-to-bml}
Let $\varphi \in \MTL$.
Then there is $\psi \in \calB(\ML)$ such that $\varphi \eqpr_{\sfH^\Box\sfL\sfS\sfM} \psi$.
\end{theorem}

Analogously as for $\PTL$, with Corollary~\ref{cor:completeness-bpl-bml} and Lemma~\ref{lem:translate-completeness} this then yields a complete axiomatization for $\MTL$, settling an open question of Kontinen et~al.~\cite{kontinen_van_2014}.

\begin{corollary}\label{cor:mtl-completeness}
The proof system $\sfH^\Box\sfL\sfS\sfM$ is sound and complete for $\MTL$.
As a consequence, $\MTL$ is axiomatizable and compact.
\end{corollary}

\subsection{Proving the modality elimination}

Note that the term \emph{$\Box$-elimination} resp.\ \emph{$\triangle$-elimination} should not be taken literally; the idea rather is to "push inside" modal operators into classical $\ML$-subformulas.
In fact, it is not hard to prove that the total nesting depth of modalities cannot in general decrease in any semantics preserving translation from $\MTL$ to $\calB(\ML)$.

Like $\limp$, the modal operators of $\MTL$ admit several provable meta-rules.
The proof of the following lemma can be found in the appendix.

\begin{lemma}\label{lem:meta-mtl}
Let $\Omega\succeq \sfL\sfS\sfM$ be a proof system.
Then $\Omega$ has substitution in $\timp, \negg, \limp, \Box$ and $\triangle$.
Furthermore, $\Omega$ admits the following meta-rules:
\begin{itemize}
	\item Modus ponens in $\Box$ \Deriv{(MP\,$\Box$)}:
	If \,$\vdash \varphi \timp \psi$ and $\Phi \vdash \Box \varphi$, then $\Phi \vdash \Box \psi$.
	\item Modus ponens in $\triangle$ \Deriv{(MP$\triangle$)}:
	If \,$\vdash \varphi \timp \psi$ and $\Phi \vdash \triangle \varphi$, then $\Phi \vdash \triangle \psi$.
	\item Modus ponens in $\Diamond$ \Deriv{(MP$\Diamond$)}:
	If \,$\vdash \varphi \timp \psi$ and $\Phi \vdash \Diamond \varphi$, then $\Phi \vdash \Diamond \psi$.
\end{itemize}
\end{lemma}

We proceed with proving that every $\MTL$-formula can be translated to a $\calB(\ML)$-formula.
Note that the $\limp$-elimination shown in Lemma~\ref{lem:tensor-elim} also applies to $\MTL$, since $\sfH^\Box\sfL\sfS\sfM$ is a conservative extension of $\sfH^0\sfL\sfS$.
It remains to establish the corresponding elimination lemmas for $\Box$ and $\triangle$.

The axioms of the system $\sfM$, depicted in Figure~\ref{fig:modal}, characterize the modal operators $\Box$ and $\triangle$ and their relationship with the other team-logical connectives.
As in the previous section, we require several auxiliary laws.
They are gathered in the system $\sfM'$ which is depicted in Figure~\ref{fig:modal2}.

\begin{lemma}\label{lem:mtl-laws}
Let $\Omega \succeq \sfH^\Box\sfL\sfS\sfM$. Then all instances of the axioms in $\sfM'$ are theorems of $\Omega$.
\end{lemma}
\begin{proof}
The "$\timp$" part of \Deriv{(D$\Box$$\timp$)} is \Deriv{(Dis$\Box$)}.
See the appendix for the other derivations.
\end{proof}

\begin{figure}[t]\centering
\scalebox{.95}{
\begin{tabular}{lcl}
		\toprule
        \Deriv{(D$\Box$$\timp$)}&$\Box(\varphi \timp \psi) \tequiv (\Box \varphi \timp \Box \psi)$&Distributive law for $\Box$ and $\timp$\\
		\Deriv{(D$\Diamond\ovee$)}&$\Diamond(\varphi \ovee \psi) \tequiv (\Diamond \varphi \ovee \Diamond \psi)$&Distributive law for $\Diamond$ and $\ovee$\\
		\Deriv{($\Diamond$Isolate$\E$)}&$\Diamond(\alpha \oland \E\beta) \tequiv \Diamond \alpha \oland \E\neg\Box\neg(\alpha \land \beta)$&\\
		\bottomrule
	\end{tabular}}
	\caption{Provable laws $\sfM'$ of $\Box$, $\triangle$ and $\Diamond$}\label{fig:modal2}
\end{figure}

\begin{lemma}\label{lem:step-box-mtl}
Let $\Omega \succeq \sfL\sfS\sfM$. Then $\calB(\ML)$ has $\Box$-elimination in $\Omega$.
\end{lemma}
\begin{proof}
Suppose $\varphi \in \calB(\ML)$.
To prove the lemma, we have to show that $\Box\varphi \eqpr \psi$ for some $\calB(\ML)$.
We repeatedly apply \Deriv{(D$\Box$$\timp$)} and \Deriv{(Lin$\Box$)} to $\Box\psi$ in order to push $\Box$ inside any $\timp$ and $\negg$ operators.
By Lemma~\ref{lem:meta-mtl}, this is also possible inside subformulas.
Since afterwards $\Box$ only occurs in classical subformulas, and since the above laws are symmetric, we conclude that $\Box\varphi$ is provably equivalent to a $\calB(\ML)$-formula.
\end{proof}

\begin{lemma}\label{lem:step-diamond-mtl}
Let $\Omega \succeq \sfH^\Box\sfL\sfS\sfM$. Then $\calB(\ML)$ has $\triangle$-elimination in $\Omega$.
\end{lemma}
\begin{proof}
Suppose $\varphi \in \calB(\ML)$.
We prove that $\triangle\varphi \eqpr \psi$ for some $\psi \in \calB(\ML)$.
By Lemma~\ref{lem:meta-mtl}, we can again perform substitution.

With $\sfL$ and \Deriv{(MP$\triangle$)}, we can show $\triangle\varphi\eqpr \negg\negg\triangle\negg\negg\varphi = \negg\Diamond\negg\varphi$.
By Theorem~\ref{thm:completeness-bool}, we can prove $\negg\varphi$ equivalent to a formula in disjunctive normal form, analogously to the proof of Lemma~\ref{lem:tensor-elim}:
\begin{align*}
	\quad&\phantom{\negg\bigdiamond} \bigovee_{i = 1}^{n}\;  \left(\bigoland_{j=1}^{o_i}\alpha_{i,j} \oland \bigoland_{j = 1}^{k_i}\E \beta_{i,j} \right)\\
	\intertext{Then $\triangle\varphi$ itself is provably equivalent to:}
	\quad&\negg\bigdiamond  \bigovee_{i = 1}^{n}\;  \left(\bigoland_{j=1}^{o_i}\alpha_{i,j} \oland \bigoland_{j = 1}^{k_i}\E \beta_{i,j} \right)\\
\intertext{For suitable $\alpha_i,\mu_{i,j},\nu_{i,j}\in \ML$:}
	\Apply{(Lemma~\ref{lem:flatness-transform2})}\eqpr \quad&\negg\bigdiamond  \bigovee_{i = 1}^{n}\;  \left(\alpha_{i} \oland \bigoland_{j = 1}^{k_i}\E \beta_{i,j} \right)\\
	\Apply{(Lemma~\ref{lem:distribute-alpha})}\eqpr \quad &\negg\bigdiamond \bigovee_{i = 1}^{n} \bigtens_{j=1}^{k_i} \bigg( \alpha_{i} \oland \E \beta_{i,j} \bigg)\\
	\Apply{(D$\Diamond\ovee$)}\eqpr \quad&\negg \bigovee_{i = 1}^{n}\; \bigdiamond \bigtens_{j=1}^{k_i} \left(\alpha_{i} \oland \E \beta_{i,j} \right)\\
	\Apply{(D$\Diamond\tens$)}\eqpr \quad&\negg \bigovee_{i = 1}^{n}\; \bigtens_{j=1}^{k_i} \Diamond \left(\alpha_{i} \oland \E \beta_{i,j} \right)\\
	\Apply{(Lemma~\ref{lem:mtl-laws})} \eqpr \quad &\negg\bigovee_{i = 1}^{n}\;  \bigtens_{j=1}^{k_i} (\Diamond\alpha_i \oland \E\neg\Box\neg(\alpha_i\land\beta_{i,j}))\\
	\Apply{(F$\Diamond$)} \eqpr \quad &\negg\bigovee_{i = 1}^{n}\;  \bigtens_{j=1}^{k_i} (\neg\Box\neg\alpha_i \oland \E\neg\Box\neg(\alpha_i\land\beta_{i,j}))\\
	\Apply{(Renaming)} = \quad &\negg\bigovee_{i = 1}^{n}\;  \bigtens_{j=1}^{k_i}(\mu_{i,j} \oland \E\nu_{i,j}) \\
	\Apply{(Lemma~\ref{lem:isolate-alpha})}\eqpr \quad &\negg\bigovee_{i = 1}^{\ell} \; \left(\bigtens_{j=1}^{k_i} \mu_{i,j} \oland \bigoland_{j=1}^{k_i} \E \left(\mu_{i,j} \land \nu_{i,j} \right)\right)\\
	\Apply{(Lemma~\ref{lem:flatness-transform})}\eqpr \quad &\negg\bigovee_{i = 1}^{\ell} \; \left(\bigvee_{j=1}^{k_i} \mu_{i,j} \oland \bigoland_{j=1}^{k_i} \E \left(\mu_{i,j} \land \nu_{i,j} \right)\right) \quad \in \calB(\ML)\text{.}\qedhere
\end{align*}
\end{proof}

We are now ready to prove the main theorem of this section:

\begin{reptheorem}{thm:mtl-is-equiv-to-bml}
Let $\varphi \in \MTL$.
Then there is $\psi \in \calB(\ML)$ such that $\varphi \eqpr_{\sfH^\Box\sfL\sfS\sfM} \psi$.
\end{reptheorem}
\begin{proof}
By induction on $\varphi$.
Suppose $\varphi \notin \calB(\ML)$.
If $\varphi$ is of the form $\varphi_1 \timp \varphi_2$ resp.\ $\negg\varphi_1$, then by induction hypothesis, $\varphi_1 \eqpr \psi_1$ for some $\psi_1 \in \calB(\ML)$ (and likewise $\varphi_2 \eqpr \psi_2$ for some $\psi_2 \in \calB(\ML)$).
By substitution, then $\varphi \eqpr \psi_1 \timp \psi_2$ resp.\ $\varphi \eqpr \negg \psi_1$.

If $\varphi$ is of the form $\Box\varphi_1$, $\triangle\varphi_1$ or $\varphi_1 \limp \varphi_2$, then again we can assume $\psi_1$ (resp.\ $\psi_1$ and $\psi_2$) as above.
By substitution, $\varphi$ is then again provably equivalent to $\Box\psi_1$, $\triangle\psi_1$, or $\psi_1 \limp \psi_2$, respectively.
By Lemma~\ref{lem:tensor-elim}, \ref{lem:step-box-mtl} and \ref{lem:step-diamond-mtl},
 $\calB(\ML)$ has elimination of $\limp$, $\Box$ and $\triangle$.
Consequently, $\varphi$ has a provably equivalent $\calB(\ML)$-formula.
\end{proof}
 
\section{First-order logic}

\label{sec:fo}

First-order logic $\FO$ does not enjoy the counter-model merging property (cf.\ Proposition~\ref{prop:counter-models-ml-pl}).
Consider, for instance, the sentences $R(c)$ and $\neg R(c)$, where $c$ is a constant.
Clearly, either of them can be falsified by an appropriate interpretation in team semantics, but to falsify both in the same structure is impossible regardless of the assigned teams.
The crucial point is that $R(c)$ and $\neg R(c)$ are contradicting \emph{sentences}.

In this section, we show that sentences are in fact the \emph{only} obstacle for axiomatizing $\calB(\FO)$ in the spirit of Section~\ref{sec:boolean}.
The problem can be remedied by the introduction of an additional axiom, the \emph{unanimity axiom}.
Then, we can prove a contradiction from formulas which, roughly speaking, already contradict on the level of sentences.

\medskip

\begin{center}
\begin{tabular}{lcl}
\toprule
\Deriv{(U)}&$\negg\alpha \timp \neg \alpha$ \quad {\small{}($\alpha$ sentence)}&\\
\bottomrule
\end{tabular}
\end{center}

We will refer to the above system simply as $\sfU$.
Similar to classical first-order logic, the truth of a sentence depends only on the underlying structure itself and not on the assignments in a given team:

\begin{lemma}\label{lem:sentences}
For any $\alpha \in \FO^0$ and structure $\calA$, the following are equivalent:
\begin{enumerate}
	\item $(\calA,T) \vDash \alpha$ for some non-empty team $T$.
	\item $(\calA,T) \vDash \alpha$ for all teams $T$.
	\item $(\calA,s) \vDash \alpha$ for some $s : \Var \to \size{\calA}$.
	\item $(\calA,s) \vDash \alpha$ for all $s : \Var \to \size{\calA}$.
\end{enumerate}
\end{lemma}
\begin{proof}
By definition of team semantics on classical formulas, 2.\ is equivalent to 4., and 1.\ is equivalent to 3.
Furthermore, 4.\ implies 3.
For this reason, it remains to show that 3.\ implies 4.
Suppose that $\calA$ is a structure, $\alpha$ is a sentence, and $s$ is as assignment such that $(\calA, s) \vDash \alpha$.
In classical semantics, it is well-known that $(\calA,s)$ and $(\calA,s')$ satisfy the same sentences for arbitrary assignments $s,s'$.
Consequently, 4.\ follows.
\end{proof}

The above lemma allows to prove \Deriv{U} sound.

\begin{lemma}[Soundness of \Deriv{U}]\label{lem:u-soundness}
Let $\alpha \in \FO^0$, and let $\calA$ be a structure.
Then $(\calA,T) \vDash \negg\alpha$ implies $(\calA,T)\vDash \neg \alpha$ for all teams $T$.
Moreover, for all non-empty teams $T$, we have $(\calA,T)\vDash \negg \alpha$ if and only if $(\calA,T) \vDash \neg \alpha$.
\end{lemma}
\begin{proof}
Assume $\alpha \in \FO^0$ and $\calA$ as above.
For the first part of the lemma, suppose $(\calA, T)\vDash \negg\alpha$.
By definition, then $(\calA,s)\vDash \neg\alpha$ for some $s \in T$.
Since $\alpha$ is a sentence, so is $\neg\alpha$, and by Lemma~\ref{lem:sentences} and the non-emptiness of $T$, $(\calA,T)\vDash \neg\alpha$.

For the second part of the lemma, we also prove the reverse direction.
For this reason, assume $T \neq \emptyset$ due to some $s \in T$, and let $(\calA, T)\vDash \neg \alpha$.
Then in particular $(\calA, s)\vDash \neg \alpha$, which implies $(\calA,T)\vDash \negg\alpha$.
\end{proof}

We proceed by investigating the fragment $\negg \FO=\{ \negg \alpha \mid \alpha \in \FO\}$.
The next proposition and the subsequent lemma show that the system $\sfH\sfU$ is not only sound, but also "complete" for $\FO$-entailments from sets of $\negg\FO$-formulas:

\begin{proposition}
Let $\Delta \subseteq \negg \FO$ be non-empty, and let $\Delta \vDash \alpha$ for some $\alpha\in \FO$.
Then there is a sentence $\varepsilon$ such that $\Delta \vDash \negg \varepsilon$, $\negg \varepsilon \vDash \neg \varepsilon$ and $\neg \varepsilon \vDash \alpha$.
\end{proposition}
\begin{proof}
Define $\varepsilon \dfn \exists x_1 \cdots \exists x_n \neg \alpha$, where $x_1,\ldots,x_n$ are the free variables of $\alpha$.
Clearly, $\neg \varepsilon \equiv \forall x_1 \cdots \forall x_n \alpha$.
In particular, $\neg \varepsilon \vDash \alpha$.
Moreover, $\negg \varepsilon \vDash \neg \varepsilon$ by the previous lemma.

It remains to prove $\Delta \vDash \negg \varepsilon$.
Suppose $(\calA, T)\vDash \Delta$ for some team $T$ and first-order structure $\calA$.
Let $V=\Set{ s | s : \Var \to \size{\calA} }$ be the team of \emph{all} assignments.
Then $T \subseteq V$, and $(\calA, V)\vDash \Delta$ by Proposition~\ref{prop:downward-closure}.
By assumption, also $(\calA, V)\vDash \alpha$.

The next step is to show that $\calA \vDash \neg\varepsilon$: Since $V$ contains all assignments, it also contains a non-empty duplicating team, \ie, a team of the form $U^{x_1}_{\size{\calA}}\ldots^{x_n}_{\size{\calA}}$ for non-empty $U$, that then satisfies $\alpha$ as well by downward closure.
By definition of $\forall$ in team semantics, $(\calA,U) \vDash \forall x_1 \cdots \forall x_n \, \alpha$, implying $(\calA,U) \vDash \neg \varepsilon$.
Note that $T \neq \emptyset$, as $T$ satisfies at least one $\negg\FO$-formula.
By Lemma~\ref{lem:sentences},  $(\calA, U)\vDash \neg \varepsilon$ implies $(\calA,T) \vDash \neg \varepsilon$, and by Lemma~\ref{lem:u-soundness}, we conclude $(\calA, T)\vDash \negg \varepsilon$.
\end{proof}

The above proposition exhibits an important property of $\negg\FO$: If a subset $\Delta  \subseteq \negg \FO$ is not satisfiable, then it already entails contradicting \emph{sentences}.
This fact is exploited in the next lemma.
It is the first step to prove the refutation completeness of the fragment $\FO \cup \negg \FO$, which is required in order to utilize Theorem~\ref{thm:completeness-of-L} for completeness of $\calB(\FO)$.

\begin{lemma}\label{lem:negg-fo-completeness}
$\sfH\sfU\sfL$ is refutation complete for $\negg\FO$.
\end{lemma}
\begin{proof}
Let $\Delta \subseteq \negg\FO$ be unsatisfiable.
Note that $\negg\delta \vdash_{\sfH\sfL} \negg \bot$ for all $\delta \in \FO$.
As $\Delta$ necessarily contains at least one formula, which is of the form $\negg\delta$, demonstrating $\Delta \vdash \bot$ then shows its inconsistency.

For the rest of the proof, we write $\delta(x_1, \ldots, x_n)$ to indicate that $\delta$ has the free variables $x_1, \ldots, x_n$.
Then we define a set $\Gamma \subseteq \FO^0$ by
\[
\Gamma \dfn \Set{\exists x_1 \cdots \exists x_n \neg\delta(x_1,\ldots,x_n) | \negg\delta(x_1,\ldots,x_n) \in \Delta}\text{.}
\]
The remaining proof of $\Delta \vdash \bot$ is split into showing $\Delta \vdash_{\sfH\sfU\sfL} \Gamma$ and $\Gamma \vdash_{\sfH} \bot$.
For the first part, note that $\forall x_1\ldots\forall x_n\, \delta(x_1,\ldots,x_n) \vdash_\sfH \delta(x_1,\ldots,x_n)$ for all $\delta(x_1,\ldots,x_n) \in \FO$ by Proposition~\ref{prop:base-completeness}.
Consequently, for all $\exists x_1 \cdots \exists x_n \neg\delta(x_1, \ldots, x_n) \in \Gamma$,
\begin{alignat*}{3}
 &	\Delta && \vdash &&\;\negg\delta(x_1,\ldots,x_n)\\
 &	&& \vdash_{\sfH\sfL} &&\;\negg \forall x_1\cdots \forall x_n\, \delta(x_1,\ldots,x_n)\\
&	 && \vdash_\sfU &&\;\neg \forall x_1 \cdots \forall x_n\, \delta(x_1,\ldots,x_n)\\
& && \vdash_\sfH &&\; \exists x_1 \cdots \exists x_n \neg\delta(x_1, \ldots, x_n)\text{.}
\end{alignat*}

It remains to prove $\Gamma \vdash \bot$, \ie, that $\Gamma$ is unsatisfiable under classical semantics.
For the sake of contradiction, assume that $\Gamma$ has a model $(\calA,s)$.
For every formula $\exists x_1 \cdots \exists x_n \neg\delta(x_1,\ldots,x_n) \in \Gamma$, let $S_\delta \dfn \Set{ s \colon \Var \to \size{\calA} | (\calA, s) \vDash \neg\delta(x_1,\ldots,x_n)}$.
By assumption, each such $S_\delta$ is non-empty, which implies $(\calA, S_\delta) \vDash \negg \delta(x_1,\ldots,x_n)$.
By Proposition~\ref{prop:downward-closure}, $(\calA,\bigcup_{\negg\delta \in \Delta}S_\delta) \vDash \Delta$, which contradicts the assumption that $\Delta$ is unsatisfiable.
\end{proof}

\subsection{From compactness to completeness}

In our approach to establish refutation completeness for $\FO\cup\negg\FO$ instead of only $\negg \FO$, we require that $\FO \cup \negg\FO$ is compact.
This is achieved by translating this fragment to classical first-order logic in a specific way.
The idea is to replace free variables in the formulas by fresh constants.
Since $\Phi$ may contain all constants in the vocabulary, we first show that an infinite set of constants can be excluded from $\Phi$.

If the vocabulary contains constants $c_0,c_1,c_2,\ldots$, and $\varphi$ is a formula, then we define $\varphi^\mathsf{even}$ as the formula where every occurrence of a constant $c_i$ is replaced by $c_{2i}$.
Analogously, $\Phi^\mathsf{even} \dfn \{ \varphi^\mathsf{even} \mid \varphi \in \Phi \}$ for sets $\Phi \subseteq \FO\cup \negg\FO$.

\begin{lemma}
If $\Phi^\mathsf{even}$ has a finite unsatisfiable subset, then $\Phi$ has a finite unsatisfiable subset of the same size.
\end{lemma}
\begin{proof}
Let $\Phi' \subseteq \Phi^\mathsf{even}$ be finite and unsatisfiable.
Then $\Phi'$ is of the form $(\Phi'')^\mathsf{even}$ for some $\Phi'' \subseteq \Phi$, since every constant in $\Phi'$ is of the form $c_{2i}$.
Moreover, $\size{\Phi'} = \size{\Phi''}$.
For this reason, it suffices to prove for arbitrary finite sets $\Phi \subseteq \FO\cup \negg\FO$ that $\Phi$ is unsatisfiable if $\Phi^\mathsf{even}$ is unsatisfiable.
But this is straightforward by contraposition, since any model $\calA$ of $\Phi^\mathsf{even}$ can be transformed into a model of $\Phi$ by reassigning the constants accordingly.
\end{proof}

\smallskip

Let now $\Phi \subseteq \FO \cup \negg\FO$.
In order to prove that $\Phi$ is either satisfiable or has a finite unsatisfiable subset, we perform a translation of $\Phi$ to a set $\Phi_f \subseteq \FO^0$.
The idea is to encode the assignments of a given team directly into the model using new constant symbols $c^x_\delta$ as explained below.
By the above lemma, we can assume that $\Phi$ excludes an infinite set of constants.
Furthermore, \wloss no variable occurs both bound and free in a formula of $\Phi$.
Then
\begin{alignat*}{4}
&\Phi_f \dfn&& &&\;\{ \gamma(c^{x_1}_\delta, \ldots, c^{x_n}_\delta) &&\;|\; \gamma(x_1, \ldots, x_n) \in \Phi\cap\FO, \negg\delta \in \Phi\cap\negg\FO\}\\
& &&\;\cup \; &&\; \{ \neg\delta(c^{x_1}_\delta, \ldots, c^{x_n}_\delta) &&\;| \;\negg\delta(x_1, \ldots, x_n) \in \Phi\cap\negg\FO\}\text{,}
\end{alignat*}
where the $c^x_\delta$ are pairwise distinct constant symbols not occurring in $\Phi$, and where $\alpha(t_1, \ldots, t_n) \dfn \alpha[x_1/t_1]\cdots[x_n/t_n]$.

\begin{lemma}
Let $\Phi \subseteq \FO \cup \negg\FO$. Then $\Phi$ is satisfiable in team semantics if and only if $\Phi_f$ is satisfiable in classical semantics.
\end{lemma}
\begin{proof}

Let $\Gamma \dfn \Phi \cap \FO$ and $\Delta \dfn \Phi \cap \negg\FO$.
\Wloss for all $\gamma \in \Gamma$ and $\negg \delta \in \Delta$ the formulas $\gamma$ and $\delta$ are distinct.

"$\Rightarrow$":
Suppose $(\calA,T)\vDash \Phi$ for some first-order structure $\calA$ and team $T$.
Extend $\calA$ to a structure $\calA'$ by interpreting the new constants as follows.
For each $\negg\delta \in \Delta$, there is a non-empty set $S_\delta \dfn \Set{ s \in T | (\calA,s)\nvDash \delta}$.
Using the axiom of choice, we let $s_\delta \in S_\delta$ be fixed and assign $(c^y_\delta)^{\calA'} \dfn s_\delta(y)$ for all $y \in \Var$.
Then $\calA' \vDash \Phi_f$ by the following argument:
\begin{align*}
	&\gamma(c^{x_1}_\delta, \ldots, c^{x_n}_\delta)  \in \Phi_f\\
	\Rightarrow \;&\gamma(x_1, \ldots, x_n) \in \Gamma\\
	\Rightarrow \;&\forall s \in T \,:\, (\calA, s) \vDash \gamma(x_1, \ldots, x_n)\\
	\Rightarrow \;&(\calA, s_\delta) \vDash \gamma(x_1, \ldots, x_n)\\
	\Rightarrow \;&\calA' \vDash \gamma(c^{x_1}_\delta, \ldots, c^{x_n}_\delta)
\end{align*}
and
\begin{align*}
	&\neg\delta(c^{x_1}_\delta, \ldots, c^{x_n}_\delta)  \in \Phi_f\\
	\Rightarrow \;&\negg\delta(x_1, \ldots, x_n) \in \Delta\\
	\Rightarrow \;&(\calA, s_\delta) \nvDash \delta(x_1, \ldots,x_n)\\
	\Rightarrow \;&(\calA, s_\delta) \vDash \neg\delta(x_1, \ldots,x_n)\\
	\Rightarrow \;&\calA' \vDash \neg\delta(c^{x_1}_\delta, \ldots, c^{x_n}_\delta)\text{.}
\end{align*}

\medskip

"$\Leftarrow$": Suppose $\calA \vDash \Phi_f$ for a first-order structure $\calA$.
For every $\negg\delta \in \Delta$, we define an assignments $s_\delta$ by $s_\delta(x) \dfn (c^{x}_\delta)^\calA$.
Then
\begin{align*}
	&\gamma(x_1, \ldots, x_n) \in \Gamma\\
	\Rightarrow\;&\forall\, \negg\delta \in \Delta \,:\; \gamma(c^{x_1}_\delta, \ldots, c^{x_n}_\delta)  \in \Phi_f\\
	\Rightarrow\;&\forall\, \negg\delta \in \Delta \,:\; \calA \vDash \gamma(c^{x_1}_\delta, \ldots, c^{x_n}_\delta)\\
	\Rightarrow\;&\forall\, \negg\delta \in \Delta \,:\; (\calA,\{s_\delta\}) \vDash \gamma(x_1,\ldots,x_n)
\end{align*}
and
\begin{align*}
	&\negg \delta(x_1, \ldots, x_n) \in \Delta\\
	\Rightarrow\;&\neg\delta(c^{x_1}_\delta, \ldots, c^{x_n}_\delta)  \in \Phi_f\\
	\Rightarrow\;&\calA \vDash \neg\delta(c^{x_1}_\delta, \ldots, c^{x_n}_\delta)\\
	\Rightarrow\;&(\calA, s_\delta) \vDash \neg\delta(x_1, \ldots, x_n)\\
	\Rightarrow\;&(\calA, \{s_\delta\}) \vDash \negg\delta(x_1, \ldots, x_n)\text{.}
\end{align*}

Define $T \dfn \Set{ s_\delta | \negg\delta \in \Delta}$.
By contraposition of Proposition~\ref{prop:downward-closure} (downward closure), $(\calA,T) \vDash \negg\delta(x_1,\ldots,x_n)$ for all $\negg\delta \in \Delta$.
By Proposition~\ref{prop:union-closure} (union closure), $(\calA, T)\vDash \gamma(x_1,\ldots,x_n)$ for all $\gamma\in \Gamma$.
Consequently, $\Phi = \Gamma \cup \Delta$ is satisfiable.
\end{proof}

From the above construction we also obtain a generalization of the empty team property (which itself states that every $\Phi \subseteq \FO$ is satisfied by the empty team).

\begin{corollary}
Every satisfiable set $\Phi\subseteq\FO \cup \negg\FO$ is satisfied in a structure with a team of cardinality $\size{\Phi\cap\negg\FO}$.
\end{corollary}

In particular the team can always be chosen countable.

\begin{lemma}[Compactness of $\FO \cup \negg \FO$]
	If a set $\Phi \subseteq \FO \cup \negg \FO$ is unsatisfiable, then it has a finite unsatisfiable subset.
\end{lemma}
\begin{proof}
Let $\Phi$ be unsatisfiable and infinite.
By the previous lemma, $\Phi_f$ is unsatisfiable as well.
By the compactness property of classical first-order logic, there exists a finite unsatisfiable subset $\Phi' \subseteq \Phi_f$.
It is now easy to show that there exists a finite set $\Phi^* \subseteq \Phi$ such that $\Phi' \subseteq (\Phi^*)_f$.
As $\Phi^*$ is then unsatisfiable, this proves the lemma.
\end{proof}

We are now ready to prove that $\calB(\FO)$ has a sound and complete proof system.

\begin{lemma}
	The system $\sfH\sfU\sfL$ is refutation complete for $\FO \cup \negg \FO$.
\end{lemma}
\begin{proof}
We have to show that any unsatisfiable $\Phi \subseteq \FO \cup \negg \FO$ is inconsistent.
However, if such $\Phi$ is unsatisfiable, then by the previous lemma, already some finite $\Phi' \subseteq \Phi$ is unsatisfiable.
Let $\Gamma \dfn \Phi' \cap \FO$ and $\Delta \dfn \Phi' \cap \negg \FO$.
As $\Gamma$ is finite, by completeness of $\sfH$, \wloss $\Gamma = \{\gamma\}$.
The following set $\Delta^\gamma \subseteq \negg \FO$ "adjoins" $\gamma$ to all formulas in $\Delta$:
\[
\Delta^\gamma \dfn \Set{ \negg(\neg \gamma \lor \delta) | \negg\delta \in \Delta} \equiv \Set{ \E(\gamma \land \neg \delta) | \negg \delta \in \Delta }
\]

The remainder of the proof shows that $\{\gamma\}\cup\Delta \vdash \Delta^\gamma$ and that $\Delta^\gamma$ is unsatisfiable.
As $\sfH\sfU\sfL$ is refutation complete for $\negg \FO$ by Lemma~\ref{lem:negg-fo-completeness}, then $\Delta^\gamma$ and consequently $\Phi$ is inconsistent.

As $\{\gamma,\neg\gamma \lor \delta\} \vdash_\sfH \delta$, we have $\Phi \vdash \{\gamma, \negg \delta\} \vdash_{\sfH\sfL} \negg(\neg\gamma\lor\delta)$ for all $\negg(\delta\lor\neg\gamma) \in \Delta^\gamma$.

 Next, assume for the sake of contradiction that $\Delta^\gamma$ is satisfiable, say, in $(\calA, T)$ for a first-order structure $\calA$ and team $T$.
For each $\negg\delta\in\Delta$, there is $s \in T$ such that $(\calA,s) \nvDash \neg\gamma \lor\delta $, \ie, $(\calA, s)\vDash \gamma\land \neg\delta $.
However, if $T' \dfn \Set{ s \in T | (\calA,s) \vDash \gamma }$, then $(\calA,T') \vDash \gamma$ by Proposition~\ref{prop:union-closure} and $(\calA,T')\vDash \Delta$ by Proposition~\ref{prop:downward-closure}, which contradicts the unsatisfiability of $\{\gamma\} \cup \Delta$.
\end{proof}

\begin{theorem}\label{thm:completeness-b-fo}
$\sfH\sfU\sfL$ is sound and complete for $\calB(\FO)$.
\end{theorem}
\begin{proof}
\Deriv{U} is sound by Lemma~\ref{lem:u-soundness}.
The systems $\sfH$ and $\sfL$ are easily checked sound.
By Theorem~\ref{thm:completeness-of-L} and the above lemma, $\sfH\sfU\sfL$ is complete.
\end{proof}

\section{Quantifier elimination}

\label{sec:qbf}

As shown in Figure~\ref{fig:foquant}, quantifiers (both propositional and first-order) are axiomatizable in a similar fashion as the modalities.
$\forall$ behaves like $\Box$, and $\exists$ behaves like $\Diamond$.
For this reason, all proofs in the spirit of Lemma~\ref{lem:meta-mtl} and \ref{lem:mtl-laws} go through for $\sfQ$ as well.
The axioms \Deriv{(I$\Box$)} and \Deriv{(I$\forall$)} differ slightly due to the fact that a duplicating team is always a supplementing team, whereas an image team is not always a successor team.

\begin{figure}[b!]
	\centering
	\begin{tabular}{ccl}
		\toprule
		\Deriv{(Lin$\forall$)}&$\forall x \negg \varphi \tequiv \negg\forall x \varphi $&The duplicating team is unique.\\
		\Deriv{(F$\exists$)}&$\exists x \alpha \tequiv \neg \forall\neg\alpha$&Flatness of $\exists$.\\
		\Deriv{(D$\exists\tens$)}&$\exists x (\varphi \tens \psi) \tequiv \exists x\varphi \tens \exists x \psi$&$\exists$ distributes over splitting.\\
		\Deriv{(E$\forall$)}&$\forall x\alpha \timp \shriek x \alpha$&Supplementing teams are subteams\\
		&& of the duplicating team.\\
		\Deriv{(I$\forall$)}&$\shriek x\psi \timp \forall x \psi$&The duplicating team is a \\
		&& supplementing team.\\
		\Deriv{(Dis$\forall$)}&$\forall x (\varphi \timp \psi) \timp (\forall x \varphi \timp \forall x \psi)$&Distribution axiom\\
		\Deriv{(Dis$\shriek$)}&$\shriek x(\varphi \timp \psi) \timp (\shriek x \varphi \timp \shriek x\psi)$&Distribution axiom\\
		\midrule
		\Deriv{(UG$\shriek$)}
		&\begin{bprooftree}
\AxiomC{$\varphi$}
\RightLabel{\small{}($\varphi$ theorem)}
\UnaryInfC{$\shriek x \varphi$}
\end{bprooftree}
&Universal generalization\\
		\bottomrule
	\end{tabular}
	\caption{The system $\sfQ$}\label{fig:foquant}
\end{figure}

\begin{theorem}\label{thm:soundness-fo}
$\sfH^1\sfL\sfS\sfQ$ is sound for $\QPTL$ and $\sfH\sfU\sfL\sfS\sfQ$ is sound for $\FO(\negg)$.
\end{theorem}
\begin{proof}
The proof is straightforward and can be found in the appendix.
\end{proof}

In this section, we regard $\forall x$ and $\shriek x$ as infinitely many unary connectives and prove elimination in the spirit of Definition~\ref{def:f-elim}.
Moreover, we refer to both propositional and first-order variables simply as "variables".

As for the other logics, we require a substitution lemma for the new logical connectives:

\begin{lemma}\label{lem:fo-meta}
Let $\Omega\succeq \sfL\sfS\sfQ$.
Then $\Omega$ has substitution in $\timp$, $\negg$, $\limp$, $\forall x$ and $\shriek x$.
Furthermore, $\Omega$ admits the following meta-rules:
\begin{itemize}
	\item Modus ponens in $\forall$ \Deriv{(MP$\forall$)}:
	If $\vdash \varphi \timp \psi$ and $\Phi \vdash \forall x\, \varphi$, then $\Phi \vdash \forall x \psi$.
	\item Modus ponens in $\shriek$ \Deriv{(MP$\shriek$)}:
	If $\vdash \varphi \timp \psi$ and $\Phi \vdash \shriek x\, \varphi$, then $\Phi \vdash \shriek x \psi$.
	\item Modus ponens in $\exists$ \Deriv{(MP$\exists$)}:
	If $\vdash \varphi \timp \psi$ and $\Phi \vdash \exists x\, \varphi$, then $\Phi \vdash \exists x \psi$.
\end{itemize}
\end{lemma}
\begin{proof}
Proven similarly to Lemma~\ref{lem:meta-mtl}.
\end{proof}

A logic $\calF$ is \emph{closed under $\forall$} if for every $\varphi \in \calF$ and variable $x$ we have $\forall x\, \varphi \in \calF$.

\begin{lemma}[$\forall x$-elimination]\label{lem:forall-elim}
Let $\calF$ be a logic closed under $\forall$. Let $\Omega \succeq \sfL\sfS\sfQ$.
Then $\calB(\calF)$ has $\forall x$-elimination in $\Omega$.
\end{lemma}
\begin{proof}
Proven as for $\Box$ in Lemma~\ref{lem:step-box-mtl} with the axioms $\sfQ$ instead of $\sfM$;
we simply "push" $\forall x$ inside $\calF$-subformulas through the enclosing $\negg$ and $\timp$.
\end{proof}

Likewise, $\shriek x$-elimination is essentially proven similarly as $\triangle$-elimination using $\sfQ$ instead of $\sfM$ and $\sfH$ resp.\ $\sfH^1$ instead of $\sfH^\Box$.

\begin{lemma}[$\shriek x$-elimination]\label{lem:shriek-elim}
Let $\calF$ be a logic closed under $\neg, \lor, \land$ and $\forall$. If $\Omega\succeq \sfH^1\sfL\sfS\sfQ$ or $\Omega \succeq \sfH\sfL\sfS\sfQ$, then $\calB(\calF)$ has $\shriek x$-elimination in $\Omega$.
\end{lemma}
\begin{proof}
Proven similarly to Lemma~\ref{lem:step-diamond-mtl}.
\end{proof}

\begin{theorem}\label{thm:qfo-to-bfo}
If $\varphi \in \FO(\negg)$, then there is $\psi \in \calB(\FO)$ such that $\varphi \eqpr_{\sfH\sfL\sfS\sfQ} \psi$.
If $\varphi \in \QPTL$, then there is $\psi \in \calB(\QPL)$ such that $\varphi \eqpr_{\sfH^1\sfL\sfS\sfQ} \psi$.
\end{theorem}
\begin{proof}
Proven analogously to Theorem~\ref{thm:mtl-is-equiv-to-bml} by using Lemma~\ref{lem:forall-elim} and \ref{lem:shriek-elim}.
\end{proof}

Similarly as for $\PTL$ and $\MTL$, we lift the completeness results for $\calB(\FO)$ and $\calB(\QPL)$ up to the full logic.

\begin{theorem}\label{thm:qfo-completeness}
$\sfH\sfU\sfL\sfS\sfQ$ is sound and complete for $\FO(\negg)$.
$\sfH^1\sfL\sfS\sfQ$ is sound and complete for $\QPTL$.
As a consequence, $\FO(\negg)$ and $\QPTL$ are axiomatizable and compact.
\end{theorem}
\begin{proof}
We have completeness for $\calB(\QPL)$ resp.\ $\calB(\FO)$ by Corollary~\ref{cor:completeness-bpl-bml} and Theorem~\ref{thm:completeness-b-fo}, and soundness by Theorem~\ref{thm:soundness-fo}.
Combining Theorem~\ref{thm:qfo-to-bfo} and  Lemma~\ref{lem:translate-completeness} proves completeness for $\FO(\negg)$ and $\QPTL$.
\end{proof}

\begin{corollary}
Both the validity problem and the entailment problem of $\FO(\negg)$ are complete for $\Sigma^0_1$, the class of recursively enumerable sets.
\end{corollary}

We also obtain a proof of Galliani's theorem \cite{Galliani14}, which states that closed $\FO(\negg)$-formulas are only as expressive as $\FO$-sentences (on non-empty teams).

\begin{theorem}
If $\varphi \in \FO(\negg)$ is closed, then $\vdash \NE \timp (\varphi \tequiv \alpha)$ for some $\alpha \in \FO^0$.
\end{theorem}
\begin{proof}
First note that the presented proofs of $\limp$-, $\forall x$- and $\shriek x$-elimination of $\calB(\FO)$ do not introduce new free variables.
Consequently, by Theorem~\ref{thm:qfo-to-bfo}, \wloss $\varphi \in \calB(\FO^0)$.
Next, we obtain $\alpha$ from $\varphi$ by replacing every occurrence of $\negg$ with $\neg$ and $\timp$ with $\imp$.
By completeness, it suffices to show that $\varphi$ and $\alpha$ are equivalent on non-empty teams.
This is by induction: on non-empty teams, $\negg\alpha'\equiv \neg\alpha'$ for all $\alpha' \in \FO^0$ by Lemma~\ref{lem:u-soundness}, and similarly, $\alpha' \timp \alpha'' \equiv \negg (\alpha' \owedge \negg \alpha'') \equiv \neg (\alpha' \land \neg \alpha'') \equiv \alpha' \imp \alpha''$.
\end{proof}
 
\section{Dependence, independence, inclusion and exclusion logic}\label{sec:fragments}

In (quantified) propositional and modal team logic, we can express atoms of dependence, independence, inclusion and exclusion in terms of other operators.
For this reason, $\QPTL$ and $\MTL$ in fact subsume a whole family of logics of dependence and independence, each obtained by adding one or more logical atoms to modal logic with team semantics.
In what follows, let $\calF$ denote $\PL$, $\ML$, or $\QPL$, respectively.

The first one is the \emph{propositional/modal dependence atom} $\dep{\vec{\alpha},\beta}$ \cite{yang_propositional_2016,vaananen_modal_2008,emdl}, where $\vec{\alpha} = (\alpha_1,\ldots,\alpha_n)$, $n \geq 0$, and $\alpha_1,\ldots,\alpha_n,\beta \in \calF$.
Next, Kontinen et~al.~\cite{mind} introduced the atom $\vec\alpha \indep_{\vec \beta} \vec \gamma$ as an equivalent to the first-order \emph{independence atom}~\cite{gradel2013dependence}.
Here, $\vec\alpha,\vec\beta,\vec\gamma$ are finite sequences of $\calF$-formulas with $\vec\alpha,\vec\gamma$ non-empty.
Finally, analogously to Galliani~\cite{galliani_inclusion_2012}, the \emph{inclusion atom} and \emph{exclusion atom} are $\vec{\alpha} \subseteq \vec{\beta}$ and $\vec{\alpha} \exclusion \vec{\beta}$~\cite{HellaS15}, where $\size{\vec\alpha} = \size{\vec\beta} > 0$.

\medskip

The semantics of these atoms is defined in terms of truth vectors.
If $\vec \alpha = (\alpha_1,\ldots,\alpha_n)$ and $s$ is an element of a team, then let $s(\vec \alpha) \dfn (s(\alpha_1),\ldots,s(\alpha_n))$, where $s(\alpha_i) \dfn 1$ if $s \vDash \alpha_i$, and $s(\alpha_i) \dfn 0$ otherwise.
For a team $T$ of propositional assignments, then
\begin{alignat*}{2}
    &T \vDash \dep{\vec\alpha,\beta} \quad && \Leftrightarrow \quad \forall s, s' \in T : s(\vec\alpha) = s'(\vec\alpha) \Rightarrow s(\beta) = s'(\beta),\\
    &T\vDash \vec\alpha \perp_{\vec \beta} \vec\gamma \quad  && \Leftrightarrow\quad  \forall s,s' \in T : s(\vec\beta) = s'(\vec\beta) \Rightarrow\\
    & && \qquad\qquad \exists s'' \in T \; s(\vec\alpha\vec\beta) = s''(\vec\alpha\vec\beta) \text{ and } s'(\vec\beta\vec\gamma) = s''(\vec\beta\vec\gamma),\\
    &T \vDash \vec\alpha \subseteq \vec \beta \quad && \Leftrightarrow \quad \forall s \in T \; \exists s' \in T : s(\vec\alpha) = s'(\vec\beta)\text{,}\\
    &T \vDash \vec\alpha \mid \vec \beta \quad && \Leftrightarrow \quad \forall s, s' \in T  : s(\vec\alpha) \neq s'(\vec\beta)\text{.}
\end{alignat*}
For teams in a Kripke structure, the definitions are analogous.

\medskip

Based on the dependence atom, the \emph{modal dependence logic} $\MDL$ \cite{vaananen_modal_2008} has the syntax
\[
\varphi \ddfn \alpha \mid \varphi \tens \varphi \mid \varphi \oland \varphi \mid \Diamond \varphi \mid \Box \varphi \mid \dep{p_1,\ldots,p_n,q}\text{,}
\]
where $n \geq 1$, $p_1,\ldots,p_n,q \in \PS$, and where $\alpha$ is an $\ML$-formula in negation normal form, \ie, with $\neg$ only occurring in front of propositional symbols $p \in \PS$.

If we allow arbitrary $\ML$-formulas in $\depop$~\cite{emdl}, then the above grammar generates \emph{extended modal dependence logic} $\EMDL$ \cite{emdl}, and $\dep{\alpha_1,\ldots,\alpha_n,\beta}$ is then called \emph{extended dependence atom}.

By analogously adding the independence atom $\vec{\alpha} \indep_{\vec{\beta}} \vec{\gamma}$, we obtain \emph{modal independence logic} $\MIL$~\cite{mind}, and with
the inclusion atom $\vec{\alpha} \subseteq \vec{\beta}$, we have modal inclusion logic $\MINC$~\cite{HellaS15}.
Adding both the inclusion and exclusion atom results in \emph{modal inclusion/exclusion logic} $\MINCEX$, the modal analogon to Galliani's I/E-logic~\cite{galliani_inclusion_2012}.

\smallskip

The modality-free fragments of $\EMDL$, $\MIL$, $\MINC$ and $\MINCEX$ are \emph{propositional dependence logic} $\PDL$~\cite{yang_propositional_2016},
\emph{propositional independence logic} $\PIL$~\cite{hannula_complexity_2015},
\emph{propositional inclusion logic} $\PINC$~\cite{hannula_complexity_2015}, and
\emph{propositional inclusion/exclusion logic} $\PINCEX$, respectively.
Adding propositional quantifiers $\exists p$ and $\forall p$ to these fragments yields \emph{quantified propositional dependence logic} $\QPDL$~\cite{gandalf},
\emph{quantified propositional independence logic} $\QPIND$~\cite{gandalf},
\emph{quantified propositional inclusion logic} $\QPINC$~\cite{gandalf} and
\emph{quantified propositional inclusion/exclusion logic} $\QPINCEX$.

Figure~\ref{fig:dep-atoms} depicts the axiom system $\sfD$ that defines the above atoms in terms of propositional/modal team logic.
Let us abbreviate $\mathbf{2}^n \dfn \{0,1\}^n$, \ie, the set of all $n$-ary truth vectors. If $\vec s \in \{0,1\}^n$, then the formula $\vec \alpha = \vec s\,$ is shorthand for
\begin{align*}
\bigwedge^n_{\substack{i=1 \\s_i = \,1}} \alpha_i \land \bigwedge^n_{\substack{i=1\\s_i = \,0}} \neg\alpha_i\text{.}
\end{align*}

\begin{figure}[t]
\centering
\scalebox{0.95}{
\begin{tabular}{rlc}
\toprule
$\dep{\vec \alpha, \beta} \; $&$\tequiv \; \top \limp \big((\bigoland\limits_{i=1}^{\size{\vec\alpha}} \dep{\alpha_i}) \timp \dep{\beta}\big)$ \quad   & \Deriv{(D1)}\\
$\dep{\beta} \; $&$\tequiv \; \beta \ovee \neg \beta$ \quad  & \Deriv{(D2)}\vspace{10pt}\\
$\vec{\alpha} \indep_{\vec{\beta}} \vec{\gamma} \; $&$\tequiv \; \bigtens\limits_{\vec s \in \mathbf{2}^{\size{\vec \beta}}} (\vec\beta=\vec s \oland \vec{\alpha} \indep \vec{\gamma})$ \quad & \Deriv{(Ind1)}\\
$\vec{\alpha} \indep \vec{\beta} \; $&$\tequiv \bigoland\limits_{\substack{\vec{s} \in \mathbf{2}^{\size{\vec\alpha}}\\ \vec{t} \in \mathbf{2}^{\size{\vec\beta}}}} \E \left(\vec{\alpha} = \vec{s}\,\right) \timp \E \left(\vec{\beta} = \vec{t}\, \right) \timp \E \left(\vec{\alpha} = \vec{s} \land \vec{\beta} = \vec{t}\, \right)$&\Deriv{(Ind2)}\vspace{10pt}\\
$\vec{\alpha} \subseteq \vec{\beta} $&$\tequiv \bigoland\limits_{\vec{s} \in \mathbf{2}^{\size{\vec\alpha}}} \E \left(\vec{\alpha} = \vec{s}\right) \timp \E \left(\vec{\beta} = \vec{s}\right)$ \quad & \Deriv{(Inc)}\vspace{10pt}\\
$\vec{\alpha} \,\, | \,\, \vec{\beta} $&$\tequiv \bigoland\limits_{\vec{s} \in \mathbf{2}^{\size{\vec\alpha}}} \E \left(\vec{\alpha} = \vec{s}\,\right) \timp \neg\! \left(\vec{\beta} = \vec{s}\,\right)$ \quad & \Deriv{(Exc)}\vspace{10pt}\\
\bottomrule
\end{tabular}}
\caption{The system $\sfD$\label{fig:dep-atoms}}
\end{figure}

\begin{theorem}
Let $\calL \in \{\PTL, \QPTL, \MTL\}$.
Let $\calL'$ be the extension of $\calL$ by $\depop, \, \indep, \subseteq$ and $\,\mid$.
Then $\calL'$ has a sound and complete proof system.
\end{theorem}
\begin{proof}
We show that there is a proof system $\Omega$ that is sound for $\calL'$, complete for $\calL$, and in which every $\calL'$-formula is provably equivalent to a $\calL$-formula.
This implies completeness by Lemma~\ref{lem:translate-completeness}.

Let $\Omega$ be the system $\sfH^0\sfL\sfS\sfD$, $\sfH^1\sfL\sfS\sfQ\sfD$ or $\sfH^\Box\sfL\sfS\sfM\sfD$, respectively.
$\Omega$ is sound for $\calL'$, and as it is a conservative extension, it admits substitution and has the deduction theorem.
As $\Omega$ can eliminate $\depop,\indep,\subseteq$ and $\exclusion$ in $\sfD$, the theorem follows.
\end{proof}

\begin{corollary}
The logics $\PDL$, $\PIL$, $\PINC$, $\PINCEX$,
$\QPDL$, $\QPIL$, $\QPINC$, $\QPINCEX$,
$\MDL$, $\EMDL$, $\MIL$, $\MINC$ and $\MINCEX$ are axiomatizable and compact.
\end{corollary}

\section{Conclusion}

\tikzset{Log/.style={circle,draw=black,fill=white,minimum width=8mm}}

\begin{figure}[b!]
\centering
\begin{tikzpicture}[->,>=latex]
\node[Log, fill=black!10] (FO) at (0,0) {};
\node[Log, fill=black!10] (QFO) at (-2cm, 2cm) {};
\node[Log, fill=black!40] (D) at (2cm, 2cm) {};
\node[Log, fill=black!70] (TL) at (0cm, 4cm) {};

\node[below of = FO] (LFO) {$\FO$};
\node[left = 0.5cm of QFO] (LQFO) {$\FO(\negg)$};
\node[right = 0.2cm of D] (LD) {$\D$};
\node[above = 0.2cm of TL] (LTL) {$\TL$};

\node[below left = -1.2cm and -0cm of LFO] (AFO) {\begin{tabular}{c}\small{}\emph{$\Sigma^0_1$-complete}\\ \small{}\emph{axiomatizable}\\ \small{}\emph{compact}\end{tabular}};
\node[above left = .2cm and -1cm of LQFO] (AQFO) {\begin{tabular}{c}\small{}\emph{$\Sigma^0_1$-complete}\\ \small{}\emph{axiomatizable} \\ \small{}\emph{compact} \end{tabular}};
\node[above right = .2cm and -1cm of LD] (AD) {\begin{tabular}{c}\small{}\emph{non-arithmetic}\\ \small{}\emph{not axiomatizable}\\ \small{}\emph{compact}\end{tabular}};
\node[above = 0.3 cm of LTL] (ATL) {\begin{tabular}{c}\small{}\emph{non-arithmetic}\\ \small{}\emph{not axiomatizable}\\ \small{}\emph{not compact}\end{tabular}};

\draw[dotted,rounded corners=15pt,black,rotate around={45:(-2cm,2cm)}] (-2.5cm,2.5cm) rectangle (1.32cm,1.5cm);
\draw[dotted,rounded corners=15pt,black,rotate around={45:(0cm,0cm)}] (-0.5cm,-0.5cm) rectangle (3.32cm,0.5cm);

\draw[draw,every node/.style={sloped,anchor=south,auto=false,inner sep=2pt},->]
    (FO) edge[->] node {\small{}$ \mathbf{\negg}$} (QFO)
    (FO) edge node {\small{}$\depop$} node[inner sep=6mm,anchor=north] {{\small{}\emph{downward closed}}} (D)
    (QFO) edge node[anchor=north] {\small{}$\depop$} node[inner sep=5.2mm] {\small{}\emph{negation closed}} (TL)
    (D) edge node[anchor=north] {\small{}$\negg$} (TL);
\end{tikzpicture}
\caption{Fragments of Väänänen's team logic $\TL$.
The arrows indicate that $\negg$ resp.\ $\depop$ is added to the syntax.}\label{fig:conclusion}
\end{figure}

\begin{figure}
\centering
\begin{tabular}{lll}\toprule
$\sfL$&\Deriv{(L1)}&$\varphi \timp (\psi \timp \varphi)$\\
&\Deriv{(L2)}&$(\varphi \timp (\psi \timp \vartheta)) \timp (\varphi \timp \psi) \timp (\varphi \timp \vartheta)$\\
&\Deriv{(L3)}&$(\negg\varphi \timp\negg\psi)\timp(\psi \timp\varphi)$\\
&\Deriv{(L4)}&$(\alpha\imp\beta)\timp (\alpha\timp\beta)$\vspace{5pt}\\
&\Deriv{(E$\timp$)}&\begin{bprooftree}
\AxiomC{$\varphi$}
\AxiomC{$\varphi \timp \psi$}
\BinaryInfC{$\psi$}
\end{bprooftree}\\
\midrule
$\sfS$&\Deriv{(F$\tens$)}&$(\alpha\tens\beta) \tequiv (\alpha\lor\beta)$\\
&\Deriv{(F$\limp$)}&$\alpha \timp (\varphi \limp \alpha)$\\
&\Deriv{(Lax)}&$\varphi \timp (\varphi \limp \psi) \timp (\vartheta \limp \psi)$\\
&\Deriv{(Ex$\limp$)}&$(\varphi \limp \psi \limp \vartheta) \timp (\psi \limp \varphi \limp \vartheta)$\\
&\Deriv{(C$\limp$)}&$(\varphi \limp \negg\psi) \timp (\psi \limp\negg \varphi)$\\
&\Deriv{(Dis$\limp$)}&$(\varphi \limp (\psi \timp \vartheta)) \timp (\varphi \limp \psi) \timp (\varphi \limp \vartheta)$\vspace{5pt}\\
&\Deriv{(Nec$\limp$)}&\begin{bprooftree}
\AxiomC{$\varphi$}
\RightLabel{\small{}($\varphi$ theorem)}
\UnaryInfC{$\psi \limp \varphi$}
\end{bprooftree}\\
\midrule
$\sfM$&\Deriv{(Lin$\Box$)}&$\Box \negg \varphi \tequiv \negg\Box\varphi $\\
&\Deriv{(F$\Diamond$)}&$\Diamond\alpha \tequiv \neg \Box\neg\alpha$\\
&\Deriv{(D$\Diamond\tens$)}&$\Diamond (\varphi \tens \psi) \tequiv \Diamond\varphi \tens \Diamond \psi$\\
&\Deriv{(E$\Box$)}&$\Box\alpha \timp \triangle \alpha$\\
&\Deriv{(I$\Box$)}&$\Diamond \varphi \timp (\triangle\psi \timp \Box \psi)$\\
&\Deriv{(Dis$\Box$)}&$\Box (\varphi \timp \psi) \timp (\Box \varphi \timp \Box\psi)$\\
&\Deriv{(Dis$\triangle$)}&$\triangle (\varphi \timp \psi) \timp (\triangle \varphi \timp \triangle\psi)$\vspace{5pt}\\
&\Deriv{(Nec$\Box$)}
&\begin{bprooftree}
\AxiomC{$\varphi$}
\RightLabel{\small{}($\varphi$ theorem)}
\UnaryInfC{$\Box \varphi$}
\end{bprooftree}\vspace{10pt}\\
&\Deriv{(Nec$\triangle$)}
&\begin{bprooftree}
\AxiomC{$\varphi$}
\RightLabel{\small{}($\varphi$ theorem)}
\UnaryInfC{$\triangle \varphi$}
\end{bprooftree}\\
\midrule
$\sfU$&\Deriv{(U)}&$\negg\alpha \timp \neg \alpha$ \quad {\small{}($\alpha$ sentence)}\\
\midrule
$\sfQ$&\Deriv{(Lin$\forall$)}&$\forall x \negg \varphi \tequiv \negg\forall x \varphi $\\
&\Deriv{(F$\exists$)}&$\exists x \alpha \tequiv \neg \forall\neg\alpha$\\
&\Deriv{(D$\exists\tens$)}&$\exists x (\varphi \tens \psi) \tequiv \exists x\varphi \tens \exists x \psi$\\
&\Deriv{(E$\forall$)}&$\forall x\alpha \timp \shriek x \alpha$\\
&\Deriv{(I$\forall$)}&$\shriek x\psi \timp \forall x \psi$\\
&\Deriv{(Dis$\forall$)}&$\forall x (\varphi \timp \psi) \timp (\forall x \varphi \timp \forall x \psi)$\\
&\Deriv{(Dis$\shriek$)}&$\shriek x(\varphi \timp \psi) \timp (\shriek x \varphi \timp \shriek x\psi)$\vspace{5pt}\\
&\Deriv{(UG$\shriek$)}
&\begin{bprooftree}
\AxiomC{$\varphi$}
\RightLabel{\small{}($\varphi$ theorem)}
\UnaryInfC{$\shriek x \varphi$}
\end{bprooftree}\\
\bottomrule
\end{tabular}
\caption{The systems $\sfL$, $\sfS$, $\sfM$, $\sfU$ and $\sfQ$}\label{fig:all-axioms}
\end{figure}

Figure~\ref{fig:conclusion} visualizes the landscape of fragments of Väänänen's team logic $\TL$ and dependence logic $\D$ \cite{vaananen_dependence_2007}.
We showed that $\FO(\negg)$, \ie, $\TL$ with lax semantics and without dependence atom, collapses to $\calB(\FO)$ and tremendously loses expressive power.

Galliani \cite{Galliani14} called $\FO(\negg)$ a natural "stopping point" of well-behaved first-order logic with team semantics, and argued that together with nothing more than a unary dependence atom, it is already as strong as full second-order logic $\SO$.
The result that $\FO(\negg)$ is axiomatizable, recursively enumerable and compact confirms that it is well-behaved.

The team-semantical extensions of propositional logic $\PL$, quantified propositional logic $\QPL$ and modal logic $\ML$, \ie, $\PTL$, $\QPTL$ and $\MTL$, have been studied as well. They have been shown axiomatizable using the fact that they collapse to the Boolean closures of their classical base logics in a similar fashion as $\FO(\negg)$, \ie, to $\calB(\PL)$ and $\calB(\ML)$.
Figure~\ref{fig:all-axioms} depicts an overview on the involved axioms.

\smallskip

For our results, using lax semantics was crucial.
In strict semantics, team divisions are defined via partitions \cite{galliani_inclusion_2012}; successor teams pick exactly one successor per world \cite{minc}; and supplementing functions have range $A$ instead of $\pow{A} \setminus \{\emptyset\}$ \cite{galliani_inclusion_2012,vaananen_dependence_2007}.
The semantics of $\tens$ would then allow to \emph{count} certain elements in the team (cf.\ p.~\pageref{p:count}).
Since this cannot be finitely expressed in $\calB(\cdot)$ (see Corollary~\ref{cor:no-counting}), no completeness proof based on a similar collapse result can exist for strict semantics.

If we permit both the lax and the strict variants of the above operators simultaneously, then $\FO(\negg)$ lies strictly between $\calB(\FO)$ and $\TL$ in terms of its expressive power.
For this reason, in future work it would be interesting to either confirm or refute whether this logic still has the above "nice properties."
One possible approach could be a proof system for $\calB(\cdot)$ that permits counting, and to extend it towards the full logic.

\section*{Acknowledgements}

The author wishes to thank Anselm Haak and Juha Kontinen, as well as the anonymous referees, for numerous comments and hints and for pointing out helpful references.

\printbibliography

\clearpage

\appendix

\renewcommand*{\thesection}{\Alph{section}}

\section*{Appendix}

The appendix contains several technical or standard proofs omitted from the previous sections of this paper.
Moreover, several derivations in the introduced proof systems are listed.
Consider the following example, viz.\ the derivation of \Deriv{(MP$\tens$)} in the system $\sfL\sfS$.

\medskip

{
\setlength{\fitchprfwidth}{2.2in}
\fitchprf{
\pline[A ]{\varphi \timp \psi \thm} \\
\pline[B ]{\vartheta \tens \varphi}
}
{
\pline[1 ]{\negg \psi \timp \negg \varphi \thm }[$\sfL$, A]\\
\pline[2 ]{(\vartheta\limp\negg\psi)\timp(\vartheta \limp \negg\varphi) \thm}[\Deriv{(Nec$\limp$)}, \Deriv{(Dis$\limp$)}, 1]\\
\pline[3 ]{\negg(\vartheta \limp \negg \varphi)}[Def., B]\\
\pline[4 ]{\negg(\vartheta\limp\negg\psi)}[$\sfL$]\\
\pline[\slider]{\vartheta \tens \psi}[Def.]
}
}

\medskip

Recall that in this paper the premises of the proof have the special line numbers A, B, \textellipsis, and that \slider marks the conclusion.
Here, "Def." means that a non-primitive logical connective such as $\tens$ or $\ovee$ is replaced by its definition.
A judgment that is a theorem (\ie, derived without using premises) is marked with "{\scriptsize(thm)}".

\section{Proof details for Section~\ref{sec:boolean}}

\begin{lemma}\label{lem:derive-contradiction}
Let $\Omega \succeq \sfL$. The following statements are equivalent:
\begin{enumerate}
	\item $\Phi \vdash \varphi$ and $\Phi\vdash \negg\varphi$ for some $\varphi$,
	\item $\Phi$ is inconsistent,
	\item $\Phi\vdash \falsum$.
\end{enumerate}
\end{lemma}
\begin{proof}
For 1.\ $\Rightarrow$ 2., we show $\Phi \vdash \xi$ for arbitrary $\xi$.
As a first step, $\Phi \vdash (\negg\xi \timp \negg\varphi)$ follows from $\Phi \vdash\negg\varphi$ and \Deriv{(L1)}.
Next, $\Phi \vdash (\varphi \timp \xi)$ follows by \Deriv{(L3)}, and by  \Deriv{(E$\timp$)} then $\Phi \vdash \xi$.
2.\ $\Rightarrow$ 3.\ is obvious.
For 3.\ $\Rightarrow$ 1., we derive $\top$ by a standard proof, as $\falsum = \negg\top$.
\end{proof}

\begin{replemma}{lem:only-one-consistent}
Let $\Omega \succeq \sfL$ and let $\Phi$ be consistent.
Then $\Phi \nvdash \varphi$ implies that $\Phi \cup \{ \negg \varphi \}$ is consistent, and $\Phi \vdash \varphi$ implies that $\Phi \cup \{\varphi\}$ is consistent.
\end{replemma}
\begin{proof}
For the first part, suppose for the sake of contradiction that $\Phi\nvdash \varphi$, but $\Phi \cup \{ \negg \varphi \}$ is inconsistent.
Then $\Phi \cup \{\negg\varphi\} \vdash \negg\psi$ for any axiom $\psi$.
Consequently, by Theorem~\ref{thm:ext-deduction}, $\Phi \vdash (\negg \varphi \timp \negg\psi)$.
But by \Deriv{(L3)}, $\Phi \vdash \psi \timp \varphi$, and ultimately $\Phi \vdash \varphi$, since $\psi$ is an axiom.
Contradiction to $\Phi\nvdash \varphi$.
As a result, $\Phi \cup \{ \negg \varphi \}$ is consistent.

The second statement is proven similarly:
Suppose that $\Phi \vdash \varphi$, but $\Phi\cup\{\varphi\}$ is inconsistent.
Then $\Phi\cup\{\varphi\} \vdash \falsum$ by Lemma~\ref{lem:derive-contradiction}, and again by Theorem~\ref{thm:ext-deduction}, $\Phi \vdash \varphi \timp \falsum$.
As a result, $\Phi \vdash \falsum$, contradicting Lemma~\ref{lem:derive-contradiction}, since $\Phi$ is consistent.
\end{proof}

\begin{replemma}{lem:lindenbaum}[Lindenbaum's Lemma]
	If $\Omega \succeq \sfL$, then every $\Omega$-consistent set has a maximal $\Omega$-consistent superset.
\end{replemma}
\begin{proof}
Let $\Phi$ be $\Omega$-consistent, $\Omega = (\Xi,\Psi,I)$, and $\Xi = \{ \xi_1, \xi_2, \ldots\}$.
Define $\Phi_0 \dfn \Phi$, and for each $i \geq 1$,
\[
\Phi_{i} \dfn \begin{cases}\Phi_{i-1} \cup \{ \xi_i \} & \text{if }\Phi_{i-1} \vdash  \xi_i \text{,}\\
\Phi_{i-1} \cup \{ \negg \xi_i \} & \text{otherwise.}\end{cases}
\]

By Lemma~\ref{lem:only-one-consistent}, the $\Omega$-consistency of $\Phi_{i-1}$ implies that of $\Phi_{i}$.
Consequently, $\Phi_i$ is $\Omega$-consistent for all $i$, and hence $\Phi^* \dfn \bigcup_{n \geq 0}\Phi_n$ is $\Omega$-consistent as well.
By construction, $\Phi^*$ is maximal $\Omega$-consistent.
\end{proof}

\begin{reptheorem}{thm:completeness-of-L}
If $\Omega \succeq \sfL$ is refutation complete for $\calF \cup \negg \calF$, then it is complete for $\calB(\calF)$.
\end{reptheorem}
\begin{proof}
Let $\Phi' \subseteq \calB(\calF)$ and $\varphi \in \calB(\calF)$.
For completeness, we have to show that $\Phi' \nvdash \varphi$ implies $\Phi' \nvDash \varphi$, \ie, that $\Phi \dfn \Phi' \cup \{\negg\varphi\}$ has a model.
First note that, if $\Phi' \nvdash \varphi$, then $\Phi'$ is consistent, and by Lemma~\ref{lem:only-one-consistent}, $\Phi$ as well.

$\Phi$ has a maximal consistent superset $\Phi^*$ by Lemma~\ref{lem:lindenbaum}.
Clearly, $\Phi^* \cap (\calF \cup \negg \calF)$ is consistent as well, and by refutation completeness for $\calF \cup \negg \calF$, it has a model $A$.
In what follows, we show that $\psi \in \Phi^* \Leftrightarrow A\vDash \psi$ for all $\psi \in \calB(\calF)$.
In particular, $\Phi$ is then satisfiable, which proves the theorem.

The proof is by induction on $\psi$.
Suppose $\psi \in \calF$.
If $\psi \in\Phi^*$, then $A\vDash\psi$ by definition of $A$.
If $\psi \notin \Phi^*$, then $\negg\psi \in \Phi^*$ by maximality of $\Phi^*$, and $A \nvDash \psi$ by definition of $A$.

For the induction step, let $\psi \notin \calF$.
The case $\psi = \negg \vartheta$ is clear as $\Phi^*$ is maximal consistent.
Next, let $\psi = \psi_1 \timp \psi_2$.
If $\psi\in\Phi^*$, then either $\psi_1\notin\Phi^*$ or $\negg\psi_2\notin\Phi^*$, as otherwise $\Phi^*$ is inconsistent.
But then $A\vDash\psi_1\timp\psi_2$ by induction hypothesis.

If $\psi\notin\Phi^*$, then $\negg\psi\in\Phi^*$.
By consistency, $\Phi^* \nvdash \psi$.
For the sake of contradiction, suppose that $A \vDash \psi$, \ie, $A \vDash \psi_2$ or $A \nvDash \psi_1$.
If $A\vDash\psi_2$, then $\psi_2 \in\Phi^*$ by induction hypothesis.
By \Deriv{(L1)}, we can then derive $\psi$, contradiction.
If $A\nvDash\psi_1$, then $\negg\psi_1 \in \Phi^*$ by induction hypothesis.
Again with \Deriv{(L1)}, we can then infer $\negg\psi_2 \timp \negg\psi_1$, and by \Deriv{(L3)} obtain $\psi$, contradicting $\Phi^* \nvdash \psi$.
\end{proof}

\section{Proof details for Section~\ref{sec:splitting}}

\begin{reptheorem}{thm:ptl-soundness}
The proof system $\sfH^0\sfL\sfS$ is sound for $\PTL$.
\end{reptheorem}
\begin{proof}\label{p:soundproof}
We show that all axioms are valid, and that the inference rules preserve truth.
Then the soundness follows by induction.
All instances of axioms of $\sfH^0$ are $\PL$-tautologies by Proposition~\ref{prop:base-completeness}, and similarly \Deriv{(E$\imp$)} is sound by Proposition~\ref{prop:equal-semantics}.
By definition of $\imp$, $\negg$ and $\timp$, the axioms of $\sfL$ and \Deriv{(E$\timp$)} are sound for all $\PTL$-formulas.
The rules \Deriv{(F$\limp$)} and \Deriv{(F$\tens$)} of $\sfS$ are valid by Proposition~\ref{prop:downward-closure}) and Proposition~\ref{prop:flatness-tens}.

For \Deriv{(Lax)}, assume $T \vDash \{\varphi,\varphi\limp\psi\}$.
Then every subteam $S \subseteq T$ satisfies $\psi$, as $T$ and $S$ always form a division of $T$ itself.
This, in turn, implies $T \vDash\vartheta \limp \psi$ for arbitrary $\vartheta$.

For \Deriv{(Ex$\limp$)}, assume for the sake of contradiction that $T$ has a division into $S,U,U'$ such that $S \vDash \psi$, $U\vDash \varphi$, but $U'\nvDash\vartheta$.
Then $S \cup U' \vDash \psi \tens \negg\vartheta$, \ie, $S \cup U' \vDash \negg(\psi \limp \vartheta)$.
However, $U \vDash \varphi$.
For this reason, $T \nvDash \varphi \limp \psi \limp \vartheta$.

For \Deriv{(C$\limp$)}, again assume for the sake of contradiction that $T$ has a division into $S$ and $U$ such that $S \vDash \psi$, but $U \nvDash \negg\varphi$, \ie, $U \vDash \varphi$.
Then, by $T= S \cup U$, we obtain $T \vDash \varphi \tens \psi$, \ie, $T \nvDash \varphi\limp\negg\psi$.

The soundness of \Deriv{(Nec$\limp$)} and \Deriv{(Dis$\limp$)} are straightforward.
\end{proof}

Recall that, if $T$ is a propositional team and $p_1,\ldots,p_n \in \PS$, then
\[
\rel(T, (p_1,\ldots,p_n)) \dfn \{ (s(p_1),\ldots,s(p_n)) \mid s \in T \} \text{.}
\]
In other words, if $(b_1,\ldots,b_n) \in \{0,1\}^n$, then $(b_1,\ldots,b_n) \in \rel(T, (p_1,\ldots,p_n))$ if and only if there exists $s \in T$ such that $s(p_1) = b_1,\ldots,s(p_n) = b_n$.
Let $[n] \dfn \{1,\ldots,n\}$.

We begin with the following lemmas in order to prove the locality property.

\begin{lemma}\label{lem:rel-equality}
The following statements are equivalent for all teams $T,T'$ and propositions $p_1,\ldots,p_n$:
\begin{enumerate}
\item $\rel(T, (p_1,\ldots,p_n)) \subseteq \rel(T', (p_1,\ldots,p_n))$
\item for all $s \in T$ there exists $s' \in T'$ such that $s(p_i) = s'(p_i)$ for all $i \in [n]$.
\end{enumerate}
\end{lemma}
\begin{proof}
1.\ $\Rightarrow$ 2.:
Suppose $s \in T$.
Then $(s(p_1),\ldots,s(p_n)) \in \rel(T,(p_1,\ldots,p_n))$ by definition of $\rel$.
Since also $(s(p_1),\ldots,s(p_n)) \in \rel(T',(p_1,\ldots,p_n))$, there exists some $s' \in T'$ such that $s(p_i) = s'(p_i)$ for all $i \in [n]$.

2.\ $\Rightarrow$ 1.:
Suppose $(b_1,\ldots,b_n) \in \rel(T, (p_1,\ldots,p_n))$.
Then there is some $s \in T$ such that $b_i = s(p_i)$ for all $i \in [n]$.
By assumption, there also exists $s' \in T'$ such that $s(p_i) = s'(p_i)$ for all $i \in [n]$.
But then $(b_1,\ldots,b_n) \in \rel(T', (p_1,\ldots,p_n))$.
\end{proof}

\begin{lemma}
Let $\rel(T,(p_1,\ldots,p_n)) = \rel(T',(p_1,\ldots,p_n))$.
Then for all $\{i_1,\ldots,i_m\} \subseteq [n]$ we have $\rel(T,(p_{i_1},\ldots,p_{i_m})) = \rel(T',(p_{i_1},\ldots,p_{i_m}))$.
\end{lemma}
\begin{proof}
Let $\{i_1,\ldots,i_m\} \subseteq [n]$; we show only $\rel(T,(p_{i_1},\ldots,p_{i_m})) \subseteq \rel(T',(p_{i_1},\ldots,p_{i_m}))$ due to symmetry.
By the previous lemma, it suffices to show that for every $s \in T$ there exists $s' \in T'$ such that $s(p_{i_j}) = s'(p_{i_j})$ for all $j \in [m]$.

Accordingly, suppose $s \in T$.
By assumption, there is $s' \in T'$ such that $s(p_{i}) = s'(p_i)$ for all $i \in [n]$.
In particular, $s(p_{i_j}) = s'(p_{i_j})$ for $j \in [m]$.
\end{proof}

\begin{repproposition}{prop:locality}[Locality]
Let $T,T'$ be propositional teams and $\varphi \in \PTL$ such that the propositions occurring in $\varphi$ are $p_1, \ldots, p_n$.
Then $\rel(T,(p_1,\ldots,p_n)) = \rel(T',(p_1,\ldots,p_n))$ implies $T \vDash \varphi \Leftrightarrow T' \vDash \varphi$ in lax semantics.
\end{repproposition}
\begin{proof}
By induction on $\varphi$.
If $\varphi = p$ for a proposition $p \in \PS$, clearly $T \vDash p \Leftrightarrow (0) \notin \rel(T,(p)) \Leftrightarrow (0) \notin \rel(T',(p)) \Leftrightarrow T' \vDash p$.

For the inductive step, assume $\varphi = \psi \timp \psi'$.
Suppose $p_{i_1},\ldots,p_{i_m}$ appear as propositions in $\psi$, and $p_{j_1},\ldots,p_{j_k}$ appear as propositions in $\psi'$, where $\{i_1,\ldots,i_m,j_1,\ldots,j_k\} \subseteq [n]$.
Then $\rel(T,(p_{i_1},\ldots,p_{i_m})) = \rel(T'(p_{i_1},\ldots,p_{i_m}))$ by the above lemma.
Likewise, $\rel(T,(p_{j_1},\ldots,p_{j_k})) = \rel(T',(p_{j_1},\ldots,p_{j_k}))$.
By induction hypothesis, $T \vDash \varphi \Leftrightarrow T' \vDash \varphi$.
The case $\varphi = \negg \psi$ is shown similarly.

Next, suppose $\varphi = \psi \limp \psi'$.
Let $\rel(T,(p_1,\ldots,p_n)) = \rel(T',(p_1,\ldots,p_n))$ and $T \vDash \psi \limp \psi'$.
For the sake of contradiction, let $\psi \limp \psi'$ be true in $T$ but false in $T'$.
Then $T' = S' \cup U'$ such that $S' \vDash \psi$ and $U' \nvDash \psi'$.
We construct subteams $S,U$ of $T$ such that $S \cup U = T$, $\rel(S,(p_1,\ldots,p_n)) = \rel(S',(p_1,\ldots,p_n))$ and $\rel(U,(p_1,\ldots,p_n)) = \rel(U',(p_1,\ldots,p_n))$.
By a similar argument as for $\timp$, then $T \nvDash \psi \limp \psi'$, contradicting the assumption.
Let
\begin{align*}
S &\dfn \{ s \in T \mid \exists s' \in S' : \forall i \in [n] : s(p_i) = s'(p_i) \}\text{,}\\
U &\dfn \{ s \in T \mid \exists s' \in U' : \forall i \in [n] : s(p_i) = s'(p_i) \}\text{.}
\end{align*}

We show that $\rel(S,(p_1,\ldots,p_n)) = \rel(S',(p_1,\ldots,p_n))$  (and analogously for $U$).
We apply Lemma~\ref{lem:rel-equality} and show that for all $s \in S$ there exists  $s' \in S'$ such that $s(p_i) = s'(p_i)$ for all $i \in [n]$, and vice versa.
As the first direction is clear, instead suppose $s' \in S'$.
Since $S' \subseteq T'$, by the assumption of the lemma we can apply Lemma~\ref{lem:rel-equality} and obtain that there is $s \in T$ such that $s(p_i) = s'(p_i)$ for all $i \in [n]$.
Then $s \in S$ by definition of $S$.

It remains to prove $T \subseteq S \cup U$.
Let $s \in T$.
As before, there exists $s' \in T'$ such that $s(p_i) = s'(p_i)$ for all $i \in [n]$.
As $s' \in S' \cup U'$, $s$ satisfies at least one of $(\exists s' \in S' : \forall i \in [n] : s(p_i) = s'(p_i))$ and $(\exists s' \in U' : \forall i \in [n] : s(p_i) = s'(p_i))$, and hence is in $S \cup U$.
\end{proof}

\begin{replemma}{lem:meta-ptl}
Let $\Omega \succeq \sfL\sfS$ be a proof system.
Them $\Omega$ has substitution in $\negg$, $\timp$ and $\limp$.
Furthermore, $\Omega$ admits the following meta-rules:
\begin{itemize}
	\item Reductio ad absurdum \Deriv{(RAA)}:
	If $\Phi \cup \{\varphi\} \vdash \{\psi, \negg \psi\}$, then $\Phi \vdash\negg\varphi$.\\
	If $\Phi \cup \{\negg\varphi\} \vdash \{\psi, \negg \psi\}$, then $\Phi \vdash \varphi$.
	\item Modus ponens in $\limp$ \Deriv{(MP$\limp$)}:
	If $\;\vdash \varphi \timp \psi$ and $\Phi \vdash \vartheta \limp \varphi$, then $\Phi \vdash \vartheta \limp \psi$.
	\item Modus ponens in $\tens$ \Deriv{(MP$\tens$)}:
	If $\;\vdash \varphi \timp \psi$ and $\Phi \vdash \vartheta \tens \varphi$, then $\Phi \vdash \vartheta \tens \psi$.
\end{itemize}
\end{replemma}
\begin{proof}
First, we derive the meta-rules in $\Omega$.
For \Deriv{(RAA)}, suppose $\Phi \cup \{\varphi\} \vdash \{\psi, \negg \psi\}$.
By Theorem~\ref{thm:ext-deduction}, $\Phi\vdash \{ \varphi \timp\negg\psi, \varphi \timp \psi \}$.
Moreover, $\vdash_\sfL (\varphi\timp\negg\psi)\timp(\psi\timp\negg\varphi)$ and $\vdash_\sfL  (\varphi\timp\negg\varphi)\timp\negg\varphi$ due to Theorem~\ref{thm:completeness-bool}.
But as $\Phi \vdash \{ \varphi \timp \psi, \psi \timp \negg \varphi \}$,
we have $\Phi \vdash \negg \varphi$.
The other case is proven analogously, using the theorem $\negg\negg\varphi\timp\varphi$ of $\sfL$.

The rule \Deriv{(MP$\limp$)} easily follows by \Deriv{(Nec$\limp$)}, \Deriv{(Dis$\limp$)} and \Deriv{(E$\timp$)}.
A derivation of \Deriv{(MP$\tens$)} can be found at the beginning of the appendix.

Next, we prove substitution in $\negg$, $\timp$ and $\limp$.
For $\negg$, suppose $\varphi = \negg \xi$ and $\xi \eqpr \psi$.
Obviously, $\{\varphi, \psi \} \vdash \xi, \negg \xi$.
By \Deriv{(RAA)}, $\varphi \vdash \negg \psi$.

For $\timp$, suppose $\varphi = \xi_1 \timp \xi_2$, $\xi_1 \eqpr \psi_1$ and $\xi_2 \eqpr \psi_2$.
Then $\{\psi_1, \varphi\} \vdash \xi_2 \vdash \psi_2$.
By the deduction theorem, $\varphi \vdash \psi_1 \timp \psi_2$.

The final case is $\limp$.
Since we demonstrated that \Deriv{(MP$\limp$)} is available, the following derivation proves substitution in $\limp$.

{
\setlength{\fitchprfwidth}{1.5in}
\fitchprf{
\pline[A ]{\psi_1 \timp \varphi_1 \thm}\\
\pline[B ]{\varphi_2 \timp \psi_2 \thm} \\
\pline[C ]{\varphi_1 \limp \varphi_2}
}
{
\pline[1 ]{\varphi_2 \timp \negg\negg\psi_2 \thm }[$\sfL$, B]\\
\pline[2 ]{\negg \varphi_1 \timp \negg \psi_1 \thm}[$\sfL$, A]\\
\pline[3 ]{\negg\negg\psi_2 \timp \psi_2 \thm}[$\sfL$]\\
\pline[4 ]{\varphi_1 \limp \negg\negg\psi_2}[\Deriv{(MP$\limp$)}, C, 1]\\
\pline[5 ]{\negg\psi_2 \limp \negg \varphi_1}[\Deriv{(C$\limp$)}]\\
\pline[6 ]{\negg\psi_2 \limp \negg \psi_1 }[\Deriv{(MP$\limp$)}, 2, 5]\\
\pline[7 ]{\psi_1 \limp \negg\negg\psi_2}[\Deriv{(C$\limp$)}]\\
\pline[\slider]{\psi_1 \limp \psi_2}[\Deriv{(MP$\limp$)}, 3, 7\qedhere]
}
}
\end{proof}

\section{Proof details for Section~\ref{sec:modal}}

\begin{lemma}\label{lem:diamond-sound}
\Deriv{(D$\Diamond\tens$)}, \ie, $\Diamond (\varphi \tens \psi) \tequiv \Diamond\varphi \tens \Diamond \psi$, is sound for $\MTL$.
\end{lemma}
\begin{proof}
Let $K = (W,R,V)$ be a Kripke structure and $T$ a team in $K$.

"$\rightarrowtriangle$":
Suppose $(K,T) \vDash \Diamond(\varphi \tens \psi)$.
Then $T$ has a successor team $T'$ that satisfies $\varphi\tens\psi$, \ie, there are $S'$ and $U'$ such that $T' = S'\cup U'$, $(K,S')\vDash \varphi$ and $(K,U')\vDash \psi$.
We define subteams $S$ and $U$ such that $T = S \cup U$, $(K,S)\vDash \Diamond\varphi$ and $(K,U)\vDash\Diamond\psi$:
\begin{align*}
S &\dfn \Set{ v \in T | \exists v' \in S' : (v,v') \in R }\text{,}\\
U &\dfn \Set{ v \in T | \exists v' \in U' : (v,v') \in R }\text{.}
\end{align*}
Every world $v \in T$ has at least one successor $v' \in T'$.
Since $S'\cup U' = T'$, either $v' \in S'$, or $v' \in U'$, or both.
By definition, $v$ is in then in $S$ or $U$.
Consequently, $T = S \cup U$.

To prove $(K, S)\vDash \Diamond \varphi$, we demonstrate that $S'$ is a successor team of $S$.
$(K,U) \vDash \Diamond \psi$ is then analogous.
First, by definition of $S$, every $v\in S$ has at least one successor in $S'$.
Likewise, every $v'\in S'$ has at least one predecessor in $S$:
Since $S' \subseteq T'$ and $T'$ is a successor team of $T$, $v'$ has some predecessor $v$ in $T$.
By definition of $S$, $v \in S$.

"$\leftarrowtriangle$":
Suppose $(K, T)\vDash \Diamond\varphi \tens \Diamond\psi$ due to subteams $S$ and $U$ of $T$ such that $T = S \cup U$, $(K,S)\vDash \Diamond\varphi$ and $(K,U)\vDash \Diamond\psi$.
Then there is a successor team $S'$ of $S$ satisfying $\varphi$, and a successor team $U'$ of $U$ satisfying $\psi$.

We show that $T' \dfn S' \cup U'$, which satisfies $\varphi \tens \psi$, itself is a successor team of $T$.
If $v \in T$, then $v \in S$ or $v \in U$, and $v$ has a successor in $S'$ or $U'$, and consequently in $T'$.
On the other hand, $v' \in T'$ implies $v' \in S'$ or $v' \in U'$.
But then $v'$ has a predecessor in $S$ or $U$, and hence in $T$.
\end{proof}

\begin{reptheorem}{thm:mtl-completeness}
The proof system $\sfH^\Box\sfL\sfS\sfM$ is sound for $\MTL$.
\end{reptheorem}
\begin{proof}
As $\sfH^\Box$ applies only to $\ML$-formulas, it is sound by Corollary~\ref{cor:base-completeness-team}.
The system $\sfL$ is easily proved sound; and the soundness of $\sfS$ is shown as in Theorem~\ref{thm:ptl-soundness}.
It remains to consider $\sfM$.

For \Deriv{(Lin$\Box$)}, clearly $(K,T)\vDash \Box\negg\varphi \Leftrightarrow (K,RT) \vDash \negg \varphi \Leftrightarrow (K,RT) \nvDash \varphi \Leftrightarrow (K,T) \nvDash \Box \varphi$.
Next, the flatness axiom \Deriv{(F$\Diamond$)} follows from the definition of successor teams.
\Deriv{(D$\Diamond\lor$)} is proved sound in Lemma~\ref{lem:diamond-sound}.
The remaining axioms \Deriv{(E$\Box$)}, \Deriv{(I$\Box$)}, \Deriv{(Dis$\Box$)} and \Deriv{(Dis$\triangle$)} and rules  \Deriv{(Nec$\Box$)} and \Deriv{(Nec$\triangle$)} are straightforward.
\end{proof}

\begin{replemma}{lem:meta-mtl}
Let $\Omega\succeq \sfL\sfS\sfM$.
Then $\Omega$ has substitution in $\timp, \negg, \limp, \Box$ and $\triangle$.
Furthermore, $\Omega$ admits the following meta-rules:
\begin{itemize}
	\item Modus ponens in $\Box$ \Deriv{(MP\,$\Box$)}:
	If $\vdash \varphi \timp \psi$ and $\Phi \vdash \Box \varphi$, then $\Phi \vdash \Box \psi$.
	\item Modus ponens in $\triangle$ \Deriv{(MP$\triangle$)}:
	If $\vdash \varphi \timp \psi$ and $\Phi \vdash \triangle \varphi$, then $\Phi \vdash \triangle \psi$.
	\item Modus ponens in $\Diamond$ \Deriv{(MP$\Diamond$)}:
	If $\vdash \varphi \timp \psi$ and $\Phi \vdash \Diamond \varphi$, then $\Phi \vdash \Diamond \psi$.
\end{itemize}
\end{replemma}
\begin{proof}
It is straightforward to prove \Deriv{(MP$\Box$)} and \Deriv{(MP$\triangle$)} from \Deriv{(Nec$\Box$)}, \Deriv{(Dis$\Box$)} resp.\ \Deriv{(Nec$\triangle$)}, \Deriv{(Dis$\triangle$)} and \Deriv{(E$\timp$)}.
Next, since $\Omega \succeq \sfL\sfS$, \Deriv{(RAA)} is available by Lemma~\ref{lem:meta-ptl}.
Consequently, the derivation for \Deriv{(MP$\Diamond$)} can be implemented as follows.

\smallskip

{
\setlength{\fitchprfwidth}{1.5in}

\fitchprf{
\pline[A ]{\Diamond \varphi}\\
\pline[B ]{\varphi \timp \psi \thm}
}
{
\pline[1 ]{\negg \psi \timp \negg \varphi \thm}[$\sfL$, B]\\
\subproof{\pline[2 ]{\triangle \negg \psi}}
{
\pline[3 ]{\triangle \negg \varphi}[\Deriv{(MP$\triangle$)}, 1, 2]\\
\pline[4 ]{\negg \triangle \negg \varphi}[Def., A]
}
\pline[5 ]{\negg\triangle\negg\psi}[\Deriv{(RAA)}, 3, 4]\\
\pline[\slider]{\Diamond\psi }[Def.]
}
}

\medskip

It remains to prove that $\Omega$ admits substitution.
The cases $\timp$, $\negg$ and $\limp$ follow from Lemma~\ref{lem:meta-ptl}, as $\Omega \succeq \sfL\sfS$.
Finally, the cases $\triangle$ and $\Box$ immediately follow from \Deriv{(MP$\triangle$)} and \Deriv{(MP$\Box$)}.
\end{proof}
 
\section{Proof details for Section~\ref{sec:qbf}}

Recall that \Deriv{(D$\exists\tens$)} is the axiom
$\exists x (\varphi \tens \psi) \tequiv \exists x \varphi \tens \exists x  \psi$ of $\sfQ$.

\begin{lemma}\label{lem:exists-sound}
\Deriv{D$\exists\tens$} is sound for $\FO(\negg)$ and $\QPTL$.
\end{lemma}
\begin{proof}
We prove only the first-order case; the proof works analogously for $\QPTL$.

"$\rightarrowtriangle$":
Suppose $(\calA, T) \vDash \exists x (\varphi \tens \psi)$, where $\calA = (A, \tau^\calA)$ is a first-order structure, $T$ is a team, and $x \in \Var$.
Then there exists $f \colon T \to \pow{A} \setminus \{\emptyset\}$ such that $(\calA, T^x_f) \vDash \varphi \tens \psi$.
Consequently, there are $S, U\subseteq T^x_f$ such that $(\calA, S) \vDash \varphi$, $(\calA, U)\vDash \psi$ and $T^x_f = S \cup U$.
For the proof, we construct a division of $T$ into subteams $S', U'$ of $T$ that satisfy $\exists x\, \varphi$ and $\exists x \, \psi$, respectively:
\begin{align*}
S' &\dfn \Set{ s \in T | \exists s' \in S \colon \forall y \in \Var \setminus \{x\} \colon s(y) = s'(y) }\text{,}\\
U' &\dfn \Set{ s \in T | \exists s' \in U \colon \forall y \in \Var \setminus \{x\} \colon s(y) = s'(y) }\text{.}
\end{align*}
In other words, $S'$ contains exactly the assignments $s\in T$ such that $s^x_a$ is in $S$ for some $a$ (and likewise $U'$).
$S'$ and $U'$ form a division of $T$: Suppose $s \in T$.
Then $s^x_a \in T^x_f$ for at least one $a$, and consequently $s^x_a \in S$ or $s^x_a \in U$.
This implies $s \in S'$ or $s \in U'$.
Next, we will prove that $S$ is a supplementing team of $S'$ (the proof for $U$ is analogous).
As then $(\calA, S') \vDash \exists x\, \varphi$ and $(\calA, U') \vDash \exists x\, \psi$, ultimately $(\calA,T) \vDash (\exists x\, \varphi)\tens(\exists x\, \psi)$.

\smallskip

We show that $S = (S')^x_g$ for $g(s) \dfn \Set{ a \in A | s^x_a \in S }$.
$g(s)$ is always non-empty, since $s \in S'$ implies $s^x_a \in S$ for some $a$ by definition of $S'$, and $g$ is a supplementing function.

In order to prove $S \subseteq (S')^x_g$, suppose $s' \in S$.
As $S \subseteq T^x_f$, then $s' = s^x_a$ for some $a \in f(s)$ and $s \in T$.
By definition of $S'$, then $s \in S'$, and since $a \in g(s)$, we have $s^x_a \in (S')^x_g$.

For $(S')^x_g \subseteq S$, let $s' \in (S')^x_g$.
Then $s' = s^x_a$ for some $s \in S'$ and $a \in g(s)$.
By definition of $g$, then $s' = s^x_a \in S$.

\medskip

"$\leftarrowtriangle$":
Suppose $(\calA,T) \vDash (\exists x\, \varphi)\tens (\exists x\, \psi)$, \ie, that $(\calA, S) \vDash \exists x \,\varphi$ and $(\calA, U)\vDash \exists x\, \psi$ for $T = S \cup U$.
Let $S^x_f$ and $U^x_g$ be supplementing teams of $S$ and $U$ such that $(\calA, S^x_f)\vDash \varphi$ and $(\calA, U^x_g)\vDash \psi$.
We prove that $S^x_f \cup U^x_g$ is a supplementing team of $T$, which implies $(\calA, T) \vDash \exists x \, (\varphi \tens \psi)$.
Consider the function $h$ on $T = S \cup U$ given by
\begin{align*}
h(s) \dfn \begin{cases}f(s) &\text{ if }s \in S \setminus U\text{,}\\
g(s) &\text{ if }s \in U \setminus S\text{,}\\
f(s) \cup g(s)& \text{ if }s \in S \cap U\text{.}
\end{cases}
\end{align*}
Clearly $h : T \to \pow{A} \setminus \{\emptyset\}$.
We demonstrate $S^x_f \cup U^x_g = T^x_h$.
For $S^x_f \subseteq T^x_h$ ($U^x_g$ is analogous), suppose $s' \in S^x_f$.
Then $s' = s^x_a$ for some $s \in S \subseteq T$ and $a \in f(s) \subseteq h(s)$.
Consequently, $s' \in T^x_h$.

Conversely, for $T^x_h \subseteq S^x_f \cup U^x_g$, let $s' \in T^x_h$, \ie, $s' = s^x_a$ for some $s \in T$ and $a \in h(s)$.
If $s \in S \setminus U$, then necessarily $a \in f(s)$, and $s^x_a \in S^x_f$.
Likewise, if $s \in U \setminus S$, then $a \in g(s)$ and $s^x_a \in U^x_g$.
Finally, if $s \in S \cap U$, then $a \in f(s)\cup g(s)$, so $s^x_a$ is either in $S^x_f$ or in $U^x_g$.
\end{proof}

\begin{lemma}
The system $\sfQ$ is sound for $\FO(\negg)$ and $\QPTL$.
\end{lemma}
\begin{proof}
For the soundness of \Deriv{D$\exists\tens$}, see the previous lemma.
\Deriv{(Lin$\forall$)} is similar to \Deriv{(Lin$\Box$)}.
The soundness of \Deriv{(F$\exists$)} is by definition of supplementing functions.
\Deriv{(E$\forall$)} follows from Proposition~\ref{prop:downward-closure}, since any supplementing team is contained in the duplicating team.
Likewise, \Deriv{(I$\forall$)} follows as the duplicating team is a supplementing team.
The rule \Deriv{(UG$\shriek$)} and the axioms \Deriv{(Dis$\forall$)} and \Deriv{(Dis$\shriek$)} are straightforward.
\end{proof}

\begin{reptheorem}{thm:soundness-fo}
$\sfH^1\sfL\sfS\sfQ$ is sound for $\QPTL$ and $\sfH\sfU\sfL\sfS\sfQ$ is sound for $\FO(\negg)$.
\end{reptheorem}
\begin{proof}
The soundness of $\sfH$ and $\sfH^1$ is shown in Corollary~\ref{cor:base-completeness-team}, that of $\sfQ$ in the above lemma, that of $\sfU$ in Lemma~\ref{lem:u-soundness}, and the remaining axioms and rules are proved sound as in Theorem~\ref{thm:ptl-soundness} on p.~\pageref{p:soundproof}.
\end{proof}

\section{Proof details for Lemma~\ref{lem:ptl-laws} (system $\sfS'$)}

In the proofs below, we sometimes omit applications of \Deriv{(MP$\tens$)} and \Deriv{(MP$\limp$)}.

\renewcommand{\thm}{\text{\tiny\; (thm)}}

{
\scriptsize

\begin{minipage}[t][][b]{.43\textwidth}
\Deriv{(Com$\tens$)}:

\medskip

{
\setlength{\fitchprfwidth}{1.2in}
\fitchprf{
\pline[A ]{\varphi \tens \psi}
}
{
	\subproof{
	\pline[1 ]{\negg(\psi\tens\varphi)}
	}
	{
		\pline[2 ]{\negg\negg(\psi \limp \negg \varphi)}[Def.]\\
		\pline[3 ]{\psi \limp \negg \varphi}[$\sfL$]\\
		\pline[4 ]{\varphi \limp \negg \psi}[\Deriv{(C$\limp$)}]\\
		\pline[5 ]{\negg\negg(\varphi\limp\negg\psi)}[$\sfL$]\\
		\pline[6 ]{\negg(\varphi\tens\psi)}[Def.]
	}
	\pline[\slider]{\psi \tens \varphi}[\Deriv{(RAA)}, A, 6]
}
}
\end{minipage}
\begin{minipage}[t][][b]{.49\textwidth}
\Deriv{(Aug$\tens$)}:

\medskip

{
	\setlength{\fitchprfwidth}{1.6in}
\fitchprf{
\pline[A ]{\varphi\tens\psi}\\
\pline[B ]{\varphi \limp \vartheta}
}
{
 	\subproof{
		\pline[1 ]{\negg(\varphi \tens (\psi \oland \vartheta))}
	}
	{
		\pline[2 ]{\negg\negg(\varphi\limp \negg (\psi \oland \vartheta))}[Def.]\\
		\pline[3 ]{\varphi \limp \negg(\psi \oland \vartheta)}[$\sfL$]\\
		\pline[4 ]{\varphi \limp (\vartheta \timp \negg\psi )}[$\sfL$]\\
		\pline[6 ]{(\varphi \limp \vartheta) \timp (\varphi \limp \negg\psi)}[\Deriv{(Dis$\limp$)}]\\
		\pline[7 ]{\varphi \limp\negg \psi}[\Deriv{(E$\timp$)}, B, 6]\\
		\pline[8 ]{\negg(\varphi \limp \negg \psi)}[Def., A]
	}
	\pline[\slider]{\varphi \tens (\psi \oland \vartheta)}[\Deriv{(RAA)}, 7, 8]
}
}
\end{minipage}

\bigskip

\begin{minipage}[t][][b]{.43\textwidth}

\Deriv{(Ass$\tens$)$^1$}:

\medskip

{
\setlength{\fitchprfwidth}{1.2in}
\fitchprf{
\pline[A ]{(\varphi\tens\psi)\tens\vartheta)}
}
{
	\pline[1 ]{\vartheta \tens (\varphi \tens\psi)}[\Deriv{(Com$\tens$)}]\\
	\pline[2 ]{\vartheta \tens (\psi \tens \varphi)}[\Deriv{(Com$\tens$)}]\\
	\pline[3 ]{(\vartheta \tens \psi) \tens \varphi}[\Deriv{(Ass$\tens$)}$^2$]\\
	\pline[4 ]{\varphi \tens (\vartheta \tens \psi)}[\Deriv{(Com$\tens$)}]\\
	\pline[\slider]{\varphi \tens (\psi \tens \vartheta)}[\Deriv{(Com$\tens$)}]
}
}
\end{minipage}
\begin{minipage}[t][][b]{.46\textwidth}

\Deriv{(Ass$\tens$)$^2$}:

\medskip

{
\setlength{\fitchprfwidth}{1.6in}
\fitchprf{
\pline[A ]{\varphi\tens(\psi\tens\vartheta)}
}
{
  \subproof{
		\pline[1 ]{\vartheta \limp \negg\negg(\varphi \limp \negg\psi)}
	}
	{
		\pline[2 ]{\vartheta \limp (\varphi \limp \negg\psi)}[$\sfL$]\\
		\pline[3 ]{\varphi \limp (\vartheta \limp \negg \psi)}[\Deriv{(Ex$\limp$)}]\\
		\pline[4 ]{\varphi \limp (\psi \limp \negg \vartheta)}[\Deriv{(C$\limp$)}]\\
		\pline[5 ]{\varphi \limp \negg\negg(\psi \limp \negg \vartheta) }[$\sfL$]\\
		\pline[6 ]{\negg(\varphi \limp \negg\negg(\psi \limp \negg \vartheta))\hspace{-10pt}}[Def., A]
	}
	\pline[7 ]{\negg(\vartheta \limp \negg\negg(\varphi \limp \negg\psi))}[\Deriv{(RAA)}, 5, 6]\\
	\pline[8 ]{\vartheta \tens (\varphi \tens \psi)}[Def.]\\
	\pline[\slider]{(\varphi \tens \psi) \tens \vartheta}[\Deriv{(Com$\tens$)}]
}
}
\end{minipage}

\bigskip

\begin{minipage}[t][][b]{.43\textwidth}

\Deriv{(Abs$\tens$)}:

\medskip

{
\setlength{\fitchprfwidth}{1.2in}
\fitchprf{
\pline[A ]{\E\alpha \tens\varphi}
}
{
  \pline[1 ]{\neg \alpha \timp \negg\negg\neg\alpha}[$\sfL$]\\
	\pline[2 ]{\neg \alpha \timp \negg\E\alpha}[Def.]\\
	\subproof{\pline[3 ]{\neg \alpha}}
	{
		\pline[4 ]{\varphi \limp \neg \alpha}[\Deriv{(I$\limp$)}]\\
		\pline[5 ]{\varphi \limp \negg \E \alpha}[\Deriv{(MP$\limp$)}, 2, 4]\\
		\pline[6 ]{\negg\negg(\varphi \limp \negg \E \alpha)}[$\sfL$]\\
		\pline[7 ]{\negg(\varphi \tens \E\alpha)}[Def.]\\
		\pline[8 ]{\varphi \tens \E\alpha}[\Deriv{(Com$\tens$)}, A]
	}
	\pline[9 ]{\negg \neg \alpha}[\Deriv{(RAA)}, 7, 8]\\
	\pline[\slider]{\E\alpha}[Def.]
}
}
\end{minipage}
\begin{minipage}[t][][b]{.46\textwidth}

\Deriv{(E$\limp$)}:

\medskip

{
\setlength{\fitchprfwidth}{1.6in}
\fitchprf{
\pline[A ]{\top \limp (\neg \alpha \timp \alpha)}
}
{
	\pline[1 ]{\top \limp \negg\negg(\neg \alpha \timp \alpha)}[$\sfL$]\\
	\pline[2 ]{\negg (\neg \alpha \timp \alpha) \limp \negg \top}[\Deriv{(C$\limp$)}]\\
	\pline[3 ]{\negg (\neg \alpha \timp \alpha) \limp \negg \alpha}[$\sfL$]\\
	\pline[4 ]{\alpha \limp \negg\negg (\neg \alpha \timp \alpha)}[\Deriv{(C$\limp$)}]\\
	\pline[5 ]{\alpha \limp  (\neg \alpha \timp \alpha)}[$\sfL$]\\
	\pline[6 ]{\alpha \tens \neg \alpha \thm}[$\sfH^0$]\\
	\pline[7 ]{\alpha \tens (\neg \alpha \oland (\neg \alpha \timp \alpha))}[\Deriv{(Aug$\tens$)}, 5, 6]\\
	\pline[8 ]{\alpha \tens (\alpha \oland \neg \alpha)}[$\sfL$]\\
	\pline[9 ]{\alpha \tens \bot}[$\sfH^0$, $\sfL$]\\
	\pline[\slider]{\alpha}[$\sfH^0$, \Deriv{(F$\tens$)}]
}
}
\end{minipage}

\bigskip

\begin{minipage}[t][][b]{.43\textwidth}

\Deriv{(D$\oland\tens$)$^1$}:

\medskip

{
\setlength{\fitchprfwidth}{1.15in}
\fitchprf{
\pline[A ]{(\alpha \oland \varphi) \tens (\alpha \oland \psi)}
}
{
	\pline[1 ]{(\alpha \oland \varphi) \tens \alpha}[$\sfL$]\\
	\pline[2 ]{\alpha \tens (\alpha \oland \varphi)}[\Deriv{(Com$\tens$)}]\\
	\pline[3 ]{\alpha \tens \alpha}[$\sfL$]\\
	\pline[4 ]{\alpha}[\Deriv{(F$\tens$)}, $\sfH^0$]\\
	\pline[5 ]{(\alpha \oland \varphi) \tens \psi}[$\sfL$, A]\\
	\pline[6 ]{\psi \tens (\alpha \oland \varphi)}[\Deriv{(Com$\tens$)}]\\
	\pline[7 ]{\psi \tens \varphi}[$\sfL$]\\
	\pline[\slider]{\alpha \oland (\varphi \tens \psi)}[\Deriv{(Com$\tens$)}, $\sfL$, 4, 7]
}
}
\end{minipage}
\begin{minipage}[t][][b]{.46\textwidth}

\Deriv{(D$\oland\tens$)$^2$}:

\medskip

{
\setlength{\fitchprfwidth}{1.63in}
\fitchprf{
\pline[A ]{\alpha \oland (\varphi \tens \psi)}
}
{
	\pline[1 ]{\alpha}[$\sfL$]\\
	\pline[2 ]{(\alpha \oland \varphi) \limp \alpha}[\Deriv{(F$\limp$)}]\\
	\pline[3 ]{(\alpha \oland \varphi) \limp (\negg(\alpha \oland \psi) \timp \negg \psi) }[$\sfL$]\\
	\pline[4 ]{\psi \limp \alpha}[\Deriv{(F$\limp$)}, 1]\\
	\pline[5 ]{\psi \limp(\negg(\alpha \oland \varphi) \timp \negg \varphi) }[$\sfL$]\\
	\subproof{\pline[6 ]{\negg ((\alpha \oland \varphi) \tens (\alpha \oland \psi))}}
	{
		\pline[7 ]{(\alpha \oland \varphi) \limp \negg(\alpha \oland \psi)}[Def., $\sfL$]\\
		\pline[8 ]{(\alpha \oland \varphi) \limp \negg \psi}[\Deriv{(Dis$\limp$)}, 3, 7]\\
		\pline[9 ]{\psi \limp \negg(\alpha \oland \varphi)}[\Deriv{(C$\limp$)}]\\
		\pline[10]{\psi \limp \negg \varphi}[\Deriv{(Dis$\limp$)}, 5, 9]\\
		\pline[11]{\psi \tens \varphi}[$\sfL$, \Deriv{(Com$\tens$)}, A]\\
		\pline[12]{\negg(\psi \limp \negg\varphi)}[Def.]
	}
	\pline[\slider]{(\alpha \oland \varphi) \tens (\alpha \oland \psi)}[\Deriv{(RAA)}, 10, 12]
}
}
\end{minipage}

\bigskip

\begin{minipage}[t][][b]{.43\textwidth}

\Deriv{(D$\ovee\tens$)$^1$}:

\medskip

{
\setlength{\fitchprfwidth}{1.15in}
\fitchprf{
\pline[A ]{(\varphi \tens \psi) \ovee (\varphi \tens \vartheta)}
}
{
\subproof{\pline[1 ]{\negg (\varphi \tens (\psi \ovee \vartheta))}}
{
	\pline[2 ]{\varphi \limp \negg(\psi \ovee \vartheta)}[Def., $\sfL$]\\
	\pline[3 ]{\varphi \limp \negg \psi}[$\sfL$]\\
	\pline[4 ]{\negg(\varphi \tens \psi)}[Def., $\sfL$]\\
	\pline[5 ]{\varphi \limp \negg \vartheta}[$\sfL$, 2]\\
	\pline[6 ]{\negg(\varphi \tens \vartheta)}[Def., $\sfL$]\\
	\pline[7 ]{\hspace{-6pt}\brokenform{\negg((\varphi \tens \psi)}
	{\formula{\ovee(\varphi \tens \vartheta))}} \hspace{-10pt} }[$\sfL$, 4, 6]
}
	\pline[\slider]{\varphi \tens (\psi \ovee \vartheta)}[\Deriv{(RAA)}, A, 7]
}
}
\end{minipage}
\begin{minipage}[t][][b]{.46\textwidth}

\Deriv{(D$\ovee\tens$)$^2$}:

\medskip

{
\setlength{\fitchprfwidth}{1.63in}
\fitchprf{
\pline[A ]{\varphi \tens (\psi \ovee \vartheta)}
}
{
	\pline[1 ]{\varphi \tens \negg(\negg\psi \oland \negg \vartheta)}[$\sfL$]\\
	\pline[2 ]{\negg (\varphi \limp \negg \negg(\negg \psi \oland \negg\vartheta))}[Def.]\\
	\subproof{\pline[3 ]{\negg ((\varphi \tens \psi) \ovee (\varphi \tens \vartheta))}}
	{
		\pline[4 ]{\negg(\varphi \tens \psi) \oland \negg(\varphi \tens \vartheta)}[$\sfL$]\\
		\pline[5 ]{\varphi \limp \negg \psi}[Def., $\sfL$, 4]\\
		\pline[6 ]{\varphi \limp \negg \vartheta}[Def., $\sfL$, 4]\\
		\pline[7 ]{\hspace{-6pt}\brokenform{\varphi \limp \big(\negg\psi \timp (\negg\vartheta}
		{\formula{\timp \negg(\psi \ovee \vartheta))\big) \thm}} }[$\sfL$, \Deriv{(Nec$\limp$)}]\\
		\pline[8 ]{\varphi \limp \negg(\psi \ovee \vartheta)}[\Deriv{(Dis$\limp$)}, 5, 6, 7]\\
		\pline[9 ]{\varphi \limp \negg\negg(\negg\psi \oland \negg \vartheta)}[$\sfL$]
	}
	\pline[\slider]{(\varphi \tens \psi) \ovee (\varphi \tens \vartheta))}[\Deriv{(RAA)}, 2, 9]
}
}
\end{minipage}

\bigskip

\begin{minipage}[t][][b]{.43\textwidth}

\Deriv{(Lax$\tens$)}:

\medskip

{
\setlength{\fitchprfwidth}{1.15in}
\fitchprf{
\pline[A ]{\varphi \tens \psi}\\
\pline[B ]{\vartheta}
}
{
\subproof{\pline[1 ]{\varphi \limp \negg\vartheta}}
{
	\pline[2 ]{\vartheta \limp \negg \varphi}[\Deriv{(C$\limp$)}]\\
	\pline[3 ]{\psi \limp \negg \varphi}[\Deriv{(Lax$\limp$)}, B, 2]\\
	\pline[4 ]{\psi \tens \varphi}[\Deriv{(Com$\tens$)}, A]\\
	\pline[5 ]{\negg(\psi \limp \negg \varphi)}[Def.]
}
	\pline[4 ]{\negg(\varphi \limp \negg\vartheta)}[\Deriv{(RAA)}, 3, 5]\\
	\pline[\slider]{\varphi \tens \vartheta}[Def.]
}
}
\end{minipage}
\begin{minipage}[t][][b]{.46\textwidth}

\Deriv{(Join$\E$)}:

\medskip

{
\setlength{\fitchprfwidth}{1.63in}
\fitchprf{
\pline[A ]{\alpha \oland \E\beta}
}
{
	\pline[1 ]{\neg(\alpha \land \beta) \timp (\alpha \imp \neg \beta)}[$\sfH^0$, \Deriv{(L4)}]\\
	\pline[2 ]{\negg(\alpha \imp \neg \beta) \timp \negg \neg(\alpha \land \beta)}[$\sfL$]\\
	\pline[3 ]{\negg(\alpha \imp \neg \beta) \timp \E(\alpha \land \beta)}[Def.]\\
	\subproof{\pline[4 ]{\alpha \imp \neg \beta}}
	{
		\pline[5 ]{\alpha}[$\sfL$, A]\\
		\pline[6 ]{\neg \beta}[$\sfH^0$]\\
		\pline[7 ]{\E\beta}[$\sfL$, A]\\
		\pline[8 ]{\negg\neg\beta}[Def.]
	}
	\pline[9 ]{\negg(\alpha \imp\neg\beta)}[\Deriv{(RAA)}, 6, 8]\\
	\pline[\slider]{\E(\alpha\land\beta)}[\Deriv{(E$\timp$)}, 3, 9]
}
}
\end{minipage}

\bigskip

\begin{minipage}[t][][b]{.43\textwidth}

\Deriv{(Sub$\E$)}:

\medskip

{
\setlength{\fitchprfwidth}{1.15in}
\fitchprf{
\pline[A ]{\alpha \imp \beta}\\
\pline[B ]{\E\alpha}
}
{
\subproof{\pline[1 ]{\neg\beta}}
{
	\pline[2 ]{\neg \alpha}[$\sfH^0$, A, 1]\\
	\pline[3 ]{\negg\neg\alpha}[Def., B]
}
	\pline[4 ]{\negg\neg\beta}[\Deriv{(RAA)}, 2, 3]\\
	\pline[\slider]{\E\beta}[Def.]
}
}
\end{minipage}
\begin{minipage}[t][][b]{.46\textwidth}

\Deriv{(Isolate$\E$)$^1$}:

\medskip

{
\setlength{\fitchprfwidth}{1.58in}
\fitchprf{
\pline[A ]{\varphi \tens (\alpha \oland \E\beta)}
}
{
	\pline[1 ]{\varphi \tens \E(\alpha \land \beta)}[\Deriv{(Join$\E$)}]\\
	\pline[2 ]{\E(\alpha\land\beta)}[\Deriv{(Com$\tens$)}, \Deriv{(Abs$\tens$)}]\\
	\pline[3 ]{\varphi \tens \alpha}[$\sfL$, A]\\
	\pline[\slider]{(\varphi \tens \alpha) \oland \E(\alpha\land\beta)}[$\sfL$, 2, 3]
}
}
\end{minipage}

\bigskip

\begin{minipage}[t][][b]{.43\textwidth}

\Deriv{(I$\tens$)}:

\medskip

{
\setlength{\fitchprfwidth}{1.25in}
\fitchprf{
\pline[A ]{\E\alpha}
}
{
\subproof{\pline[1 ]{\negg(\top \tens (\alpha \oland \E\alpha))}}
{
	\pline[2 ]{\top \limp \negg(\alpha \oland \E\alpha)}[Def., $\sfL$]\\
	\pline[3 ]{\top \limp \negg(\alpha \oland \negg\neg\alpha)\!\!}[Def.]\\
	\pline[4 ]{\top \limp (\alpha \timp \neg\alpha)}[$\sfL$]\\
	\pline[5 ]{\top \limp (\neg\neg\alpha \timp \neg\alpha)\!\!}[$\sfH^0$, $\sfL$]\\
	\pline[6 ]{\neg \alpha}[\Deriv{(E$\limp$)}]\\
	\pline[7 ]{\negg\neg \alpha}[Def., A]
}
	\pline[\slider]{\top \tens (\alpha \oland \E\alpha)}[\Deriv{(RAA)}, 6, 7]
}
}
\end{minipage}
\begin{minipage}[t][][b]{.46\textwidth}

\Deriv{(Isolate$\E$)$^2$}:

\medskip

{
\setlength{\fitchprfwidth}{1.58in}
\fitchprf{
\pline[A ]{(\varphi \tens \alpha) \oland \E(\alpha \land \beta)}
}
{
	\pline[1 ]{\varphi \tens \alpha}[$\sfL$, A]\\
	\pline[2 ]{\E(\alpha\land\beta)}[$\sfL$, A]\\
	\pline[3 ]{\top \tens((\alpha \land \beta) \oland \E(\alpha \land \beta)}[\Deriv{(I$\tens$)}]\\
	\pline[4 ]{\top \tens (\alpha \oland \E\beta)}[$\sfH^0$, $\sfL$]\\
	\pline[5 ]{(\alpha \oland \E\beta) \tens \top}[\Deriv{(Com$\tens$)}]\\
	\pline[6 ]{(\alpha \oland \E\beta) \tens (\varphi \tens \alpha)}[\Deriv{(Lax$\tens$)}, 1, 5]\\
	\pline[7 ]{(\varphi \tens \alpha) \tens (\alpha \oland \E\beta)}[\Deriv{(Com$\tens$)}]\\
	\pline[8 ]{\varphi \tens (\alpha \tens (\alpha \oland \E\beta))}[\Deriv{(Ass$\tens$)}]\\
	\pline[9 ]{\varphi \tens ((\alpha \oland \alpha) \tens (\alpha \oland \E\beta))}[$\sfL$, \Deriv{(Com$\tens$)}]\\
	\pline[10]{\varphi \tens (\alpha \oland (\alpha \tens \E\beta))}[\Deriv{(D$\oland\tens$)}]\\
	\pline[\slider]{\varphi \tens (\alpha \oland \E\beta)}[\Deriv{(Com$\tens$)}, \Deriv{(Abs$\tens$)}]
}
}
\end{minipage}

}

\section{Proof details for Lemma~\ref{lem:mtl-laws} (system $\sfM'$)}

\label{sec:aux-modal}

As for \Deriv{(MP$\tens$)} and \Deriv{(MP$\limp$)}, we mostly omit applications of \Deriv(MP$\Diamond$), \Deriv(MP$\Box$) and \Deriv(MP$\triangle$) in the derivations.

\bigskip

{
\scriptsize

\begin{minipage}[t][][b]{.43\textwidth}

\Deriv{(D$\Diamond\ovee$)$^1$}:

\medskip

{
\setlength{\fitchprfwidth}{1.33in}
\fitchprf{
\pline[A ]{\Diamond(\varphi\ovee\psi)}
}
{
\subproof{\pline[1 ]{\negg\Diamond\varphi \ovee \Diamond\psi)}}
{
	\pline[2 ]{\negg(\negg\triangle\negg\varphi \ovee \negg\triangle\negg\psi)\hspace{-9pt}}[Def.]\\
	\pline[3 ]{\triangle\negg\varphi}[$\sfL$, 2]\\
	\pline[4 ]{\triangle\negg\psi}[$\sfL$, 2]\\
	\pline[5 ]{\triangle\negg(\varphi \ovee \psi)}[\Deriv{(Dis$\triangle$)}, $\sfL$]\\
	\pline[6 ]{\negg\negg\triangle\negg(\varphi \ovee \psi)}[$\sfL$]\\
	\pline[7 ]{\negg\Diamond(\varphi \ovee \psi)}[Def.]
}
	\pline[\slider]{\Diamond\varphi \ovee \Diamond\psi}[\Deriv{(RAA)}, A, 7]
}
}
\end{minipage}
\begin{minipage}[t][][b]{.46\textwidth}

\Deriv{(D$\Diamond\ovee$)$^2$}:

\medskip

{
\setlength{\fitchprfwidth}{1.58in}
\fitchprf{
\pline[A ]{\Diamond\varphi \ovee \Diamond\psi}
}
{
\subproof{\pline[1 ]{\negg\Diamond(\varphi \ovee\psi)}}
{
	\pline[2 ]{\triangle\negg(\varphi \ovee \psi)}[Def., $\sfL$]\\
	\pline[3 ]{\triangle\negg\varphi}[$\sfL$, 2]\\
	\pline[4 ]{\triangle\negg\psi}[$\sfL$, 2]\\
	\pline[5 ]{(\negg\negg\triangle\negg\varphi) \oland (\negg\negg\triangle\negg\psi)}[$\sfL$]\\
	\pline[6 ]{(\negg \Diamond\varphi) \oland (\negg\Diamond\psi)}[Def.]\\
	\pline[7 ]{\negg(\Diamond \varphi \ovee \Diamond \psi)}[$\sfL$]
}
	\pline[\slider]{\Diamond(\varphi \ovee\psi)}[\Deriv{(RAA)}, A, 7]
}
}
\end{minipage}

\bigskip

\begin{minipage}[t][][b]{.43\textwidth}

\Deriv{(Com$\Diamond\E$)$^1$}:

\medskip

{
\setlength{\fitchprfwidth}{1.33in}
\fitchprf{
\pline[A ]{\Diamond\E\beta}
}
{
\pline[1 ]{\negg\triangle\negg\negg\neg\beta}[Def.]\\
\subproof{\pline[2 ]{\negg\E\neg\Box\neg\beta}}
{
	\pline[3 ]{\negg\negg\neg\neg\Box\neg\beta}[Def.]\\
	\pline[4 ]{\Box\neg\beta}[$\sfL$, $\sfH^0$]\\
	\pline[5 ]{\triangle\neg\beta}[\Deriv{(E$\Box$)}]\\
	\pline[6 ]{\triangle\negg\negg \neg\beta}[$\sfL$]
}
	\pline[\slider]{\E\neg\Box\neg\beta}[\Deriv{(RAA)}, 1, 6]
}
}
\end{minipage}
\begin{minipage}[t][][b]{.46\textwidth}

\Deriv{(Com$\Diamond\E$)$^2$}:

\medskip

{
\setlength{\fitchprfwidth}{1.58in}
\fitchprf{
\pline[A ]{\Diamond\varphi}\\
\pline[B ]{\E\neg\Box\neg\beta}
}
{
\pline[1 ]{\negg\neg\neg\Box\neg\beta}[Def., B]\\
\subproof{\pline[2 ]{\negg\Diamond\E\beta}}
{
	\pline[3 ]{\negg\negg\triangle\negg\negg\neg\beta}[Def.]\\
	\pline[4 ]{\triangle\neg\beta}[$\sfL$]\\
	\pline[5 ]{\Box\neg\beta}[\Deriv{(I$\Box$)}]\\
	\pline[6 ]{\neg\neg\Box\neg\beta}[$\sfH^0$]
}
	\pline[\slider]{\Diamond\E\beta}[\Deriv{(RAA)}, 1, 6]
}
}
\end{minipage}

\bigskip

\begin{minipage}[t][][b]{.43\textwidth}

\Deriv{(Aug$\Diamond$)}:

\medskip

{
\setlength{\fitchprfwidth}{1.2in}
\fitchprf{
\pline[A ]{\Diamond\varphi}\\
\pline[B ]{\triangle\psi}
}
{
\subproof{\pline[1 ]{\negg\Diamond(\varphi\oland\psi)}}
{
	\pline[2 ]{\triangle\negg(\varphi \oland \psi)}[Def., $\sfL$]\\
	\pline[3 ]{\triangle\negg\varphi}[\Deriv{(Dis$\triangle$)}, $\sfL$, B, 2]\\
	\pline[4 ]{\negg\negg\triangle\negg\varphi}[$\sfL$]\\
	\pline[5 ]{\negg\Diamond\varphi}[Def.]
}
	\pline[\slider]{\Diamond(\varphi\oland\psi)}[\Deriv{(RAA)}, A, 5]
}
}
\end{minipage}
\begin{minipage}[t][][b]{.46\textwidth}

\Deriv{(Join$\Diamond$)}:

\medskip

{
\setlength{\fitchprfwidth}{1.58in}
\fitchprf{
\pline[A ]{\Diamond\alpha}\\
\pline[B ]{\Diamond\E\alpha}
}
{
\subproof{\pline[1 ]{\negg\Diamond(\alpha \oland \E\alpha)}}
{
	\pline[2 ]{\triangle\negg(\alpha \oland \E\alpha)}[Def., $\sfL$]\\
	\pline[3 ]{\triangle(\alpha \timp \neg\alpha)}[Def., $\sfL$]\\
	\pline[4 ]{\Diamond(\alpha \oland (\alpha \timp \neg\alpha))}[\Deriv{(Aug$\Diamond$)}, A, 3]\\
	\pline[5 ]{\Diamond(\alpha \oland \neg\alpha)}[$\sfL$]\\
	\pline[6 ]{\Diamond\bot}[$\sfL$, $\sfH^0$]\\
	\pline[7 ]{\bot}[\Deriv{(F$\Diamond$)}, $\sfH^\Box$]\\
	\pline[8 ]{\Box\neg\alpha}[$\sfH^\Box$]\\
	\pline[9 ]{\triangle\neg\alpha}[\Deriv{(E$\Box$)}, A, 8]\\
	\pline[10]{\triangle\negg\negg\neg\alpha}[$\sfL$]\\
	\pline[11]{\negg\triangle\negg\negg\neg\alpha}[Def., B]
}
	\pline[\slider]{\Diamond(\alpha\oland\E\alpha)}[\Deriv{(RAA)}, 10, 11]
}
}
\end{minipage}

\bigskip

\begin{minipage}[t][][b]{.43\textwidth}

\Deriv{(D$\Box{\timp}$)}:

\medskip

{
\setlength{\fitchprfwidth}{1.2in}
\fitchprf{
\pline[A ]{\Box\varphi \timp \Box \psi}
}
{
\subproof{\pline[1 ]{\negg\Box(\varphi\timp\psi)}}
{
	\pline[2 ]{\Box\negg(\varphi \timp \psi)}[\Deriv{(Lin$\Box$)}]\\
	\pline[3 ]{\Box\varphi}[$\sfL$, 2]\\
	\pline[4 ]{\Box\negg\psi}[$\sfL$, 2]\\
	\pline[5 ]{\negg\Box\psi}[\Deriv{(Lin$\Box$)}]\\
	\pline[6 ]{\Box \varphi \oland \negg \Box\psi}[$\sfL$, 3, 5]\\
	\pline[7 ]{\negg(\Box \varphi \timp \Box\psi)}[$\sfL$]
}
	\pline[\slider]{\Box(\varphi\timp\psi)}[\Deriv{(RAA)}, A, 7]
}
}
\end{minipage}
\begin{minipage}[t][][b]{.46\textwidth}

\Deriv{($\Diamond$Isolate$\E$)$^1$}:

\medskip

{
\setlength{\fitchprfwidth}{1.58in}
\fitchprf{
\pline[A ]{\Diamond(\alpha\oland\E\beta)}
}
{
	\pline[1 ]{\Diamond\alpha}[$\sfL$]\\
	\pline[2 ]{\Diamond\E(\alpha \land \beta)}[\Deriv{(Join$\E$)}, A]\\
	\pline[3 ]{\E\neg\Box\neg(\alpha \land \beta)}[\Deriv{(Com$\Diamond\E$)}]\\
	\pline[\slider]{\Diamond \alpha \oland \E\neg\Box\neg(\alpha \land \beta)}[$\sfL$, 1, 3]
}
}
\end{minipage}

\bigskip

\Deriv{($\Diamond$Isolate$\E$)$^2$}:

\medskip

{
\setlength{\fitchprfwidth}{2.8in}
\fitchprf{
\pline[A ]{\Diamond\alpha \oland \E\neg\Box\neg(\alpha\land\beta)}
}
{
\subproof{\pline[1 ]{\neg\Box\neg(\alpha\land\beta) \oland \E\neg\Box\neg(\alpha\land\beta)}}
{
	\pline[2 ]{\Diamond(\alpha\land\beta)}[\Deriv{(F$\Diamond$)}, $\sfL$, 1]\\
	\pline[3 ]{\E\neg\Box\neg(\alpha\land\beta)}[$\sfL$, 1]\\
	\pline[4 ]{\Diamond\E(\alpha\land\beta)}[\Deriv{(Com$\Diamond\E$)}]\\
	\pline[5 ]{\Diamond\big((\alpha\land\beta) \oland \E(\alpha\land\beta) \big)}[\Deriv{(Join$\Diamond$)}, 2, 4]\\
	\pline[6 ]{\Diamond(\alpha \oland \E\beta)}[\Deriv{(Sub$\E$)}, $\sfH^0$, $\sfL$]
}
	\pline[7 ]{\big(\neg\Box\neg(\alpha\land\beta) \oland \E\neg\Box\neg(\alpha\land\beta)\big) \timp \Diamond(\alpha \oland \E\beta) \thm}[Ded. Thm.]\\
	\pline[8 ]{\E\neg\Box\neg(\alpha\land\beta)}[$\sfL$, A]\\
	\pline[9 ]{\top \tens (\neg\Box\neg(\alpha\land\beta) \oland \E\neg\Box\neg(\alpha\land\beta))}[\Deriv{(I$\tens$)}]\\
	\pline[10]{\top \tens \Diamond(\alpha \oland \E\beta)}[\Deriv{(MP$\tens$)}, 7, 9]\\
	\pline[11]{\Diamond(\alpha \oland \E\beta) \tens \top}[\Deriv{(Com$\tens$)}]\\
	\pline[12]{\Diamond\alpha}[$\sfL$, A]\\
	\pline[13]{\Diamond(\alpha\oland \alpha)}[$\sfL$]\\
	\pline[14]{\Diamond(\alpha \oland \E\beta) \tens \Diamond (\alpha\oland \alpha)}[\Deriv{(Lax$\tens$)}, 11, 13]\\
	\pline[15]{\Diamond\big((\alpha \oland \E\beta) \tens (\alpha\oland \alpha)\big)}[\Deriv{(D$\Diamond\tens$)}]\\
	\pline[16]{\Diamond(\alpha \oland (\E\beta \tens \alpha))}[\Deriv{(D$\oland\tens$)}]\\
	\pline[\slider]{\Diamond(\alpha \oland \E\beta)}[\Deriv{(Abs$\tens$)}, $\sfL$]
}
}

}
  
\end{document}